\documentclass[acmsmall]{acmart}

\usepackage{verbatim}

\usepackage[normalem]{ulem}
\usepackage{xcolor}
\usepackage[T1]{fontenc}
\usepackage{epsfig,endnotes,footnote}
\usepackage{array,multirow}
\usepackage{amsmath}
\usepackage{enumitem}
\usepackage{booktabs}
\usepackage{xspace}
\usepackage{relsize}
\usepackage{hyphenat}
\usepackage{algorithm}
\usepackage{algorithmicx}
\usepackage{algpseudocode}
\usepackage{rotating}
\usepackage{graphicx}
\usepackage{subcaption}
\DeclareGraphicsExtensions{.pdf,.jpg,.png}
\usepackage{multirow}

\usepackage{url}
\usepackage{wrapfig}

\usepackage{tabularx}
\usepackage{mathpartir}
\usepackage[title]{appendix}

\hypersetup{
    bookmarks=true,         
    unicode=false,          
    pdftoolbar=true,        
    pdfmenubar=true,        
    pdffitwindow=true,      
    pdftitle={},    
    pdfauthor={},     
    pdfsubject={},   
    pdfnewwindow=true,      
    pdfkeywords={}, 
    colorlinks=true,       
    linkcolor=blue,          
    citecolor=blue,        
    filecolor=magenta,      
    urlcolor=cyan           
}

\newcommand{\m}[1]{\mathsf{#1}}
\newcommand{\mt}[1]{\mathit{#1}}
\newcommand{\paragraphb}[1]{\vspace{0.05in}\noindent{\bf #1}\xspace}
\newcommand{\bnfdef}{::=}
\newcommand{\bnfalt}{\,|\,}
\newcommand{\rulename}[1]{\textsc{#1}}
\newcommand{\ifthen}[3]{\m{if}\ #1\ \m{then}\ #2\ \m{else}\ #3}

\newcommand{\timestamp}{\tau}
\newcommand{\context}{\kappa}
\newcommand{\nvmem}{N}
\newcommand{\vmem}{V}
\newcommand{\nint}{N_{\mt{int}}}
\newcommand{\ncont}{N_{\mt{cont}}}
\newcommand{\nrb}{N_{\mt{rb}}}

\newcommand{\cmd}{c}
\newcommand{\loc}{\mt{loc}}
\newcommand{\val}{\mt{val}}

\newcommand{\resetm}{\mathit{reset}}
\newcommand{\runof}[1]{\mt{Run}(#1)}
\newcommand{\inputs}{\mathcal{I}}
\newcommand{\war}{\mathit{WAR}}
\newcommand{\rio}{\mathit{RIO}}
\newcommand{\tntrel}{\mathrel{\approx_\mt{tnt}}}
\newcommand{\dom}{\m{dom}}

\newcommand{\tshared}{\mathcal{T}\!\textit{s}}
\newcommand{\tlocal}{\mathcal{T}\!\iota}
\newcommand{\tpriv}{\mathcal{T}\!\textit{p}}
\newcommand{\tlocaln}{\mathcal{T}\!\iota_N}
\newcommand{\tlocalv}{\mathcal{T}\!\iota_V}

\newcommand{\stepsto}{\longrightarrow}
\newcommand{\Stepsto}[1]{\stackrel{#1}{\Longrightarrow}}
\newcommand{\SeqStepsto}[1]{\stackrel{#1}{\longrightarrow}}
\newcommand{\MStepsto}[1]{\stackrel{#1}{\Longrightarrow^*}}
\newcommand{\MSeqStepsto}[1]{\stackrel{#1}{\longrightarrow^*}}
\newcommand{\proj}[2]{#1|_{#2}}

\newcommand{\ee}{\mathcal{E}}

\newcommand{\denote}[1]{[\![ #1 ]\!]}

\newcommand{\erase}[1]{{#1}^-}

\newcommand{\oleq}{\leqslant}
\newcommand{\oleqm}{\leqslant^m}
\newcommand{\oleqmc}{\leqslant_c^m}

\newcommand{\Vtaint}{\Vdash_\mt{taint}}

\newcommand{\cache}{\mathcal{C}}

\newcommand{\ulog}{\mathcal{L}}
\newcommand{\loggedl}{\mt{LL}}
\newcommand{\ulctx}{\kappa_\mt{UL}}
\newcommand{\ulStepsto}[1]{\stackrel{#1}{\Longrightarrow}_\mt{UL}}
\newcommand{\ulMStepsto}[1]{\stackrel{#1}{\Longrightarrow}^*_\mt{UL}}

\newcommand{\rlctx}{\kappa_\mt{RL}}
\newcommand{\rlStepsto}[1]{\stackrel{#1}{\Longrightarrow}_\mt{RL}}
\newcommand{\rlMStepsto}[1]{\stackrel{#1}{\Longrightarrow}^*_\mt{RL}}

\newcommand{\tskctx}{\kappa_\mt{TSK}}
\newcommand{\tskStepsto}[1]{\stackrel{#1}{\Longrightarrow}_\mt{TSK}}

\newcommand{\rel}{\looparrowright}
\newcommand{\dstate}{\Sigma}

\newcommand{\ustate}{\Sigma_{\mt{UL}}}
\newcommand{\rstate}{\Sigma_{\mt{RL}}}
\newcommand{\tstate}{\Sigma_{\mt{TSK}}}

\newcommand{\dcon}{\context}

\newcommand{\ucon}{\context_{\mt{UL}}}
\newcommand{\rcon}{\context_{\mt{RL}}}
\newcommand{\tcon}{\context_{\mt{TSK}}}

\newcommand{\trel}{\leadsto}

\theoremstyle{plain}
\newtheorem{thm}{Theorem}
\newtheorem{lem}[thm]{Lemma}
\newtheorem{cor}[thm]{Corollary}

\newtheorem{defn}[thm]{Definition}

\newenvironment{proofsketch}{\vspace{-1pt}\textsc{Proof (sketch). }\hspace*{0.25em}}{ \hspace*{\fill} \qed}

\newcommand{\notes}[2]{} 
\newcommand{\limin}[1]{\notes{magenta}{Limin says: #1}}

\newcommand{\brandon}[1]{\notes{orange}{Brandon says: #1}}

\newif\ifproofs
\proofstrue

\graphicspath{{graphics/}}

\startPage{1}

\setcopyright{none}

\begin{document}

\title[Foundations of Intermittent Computing]{Towards a Formal Foundation of Intermittent Computing}

\author{Milijana Surbatovich}

\affiliation{
  \institution{Carnegie Mellon University}            
  \country{USA}                    
}
\email{milijans@andrew.cmu.edu}          

\author{Limin Jia}

\affiliation{
  \institution{Carnegie Mellon University}            
  \country{USA}                    
}
\email{liminjia@cmu.edu}          

\author{Brandon Lucia}

\affiliation{
  \institution{Carnegie Mellon University}            
  \country{USA}                    
}
\email{blucia@cmu.edu}          

\begin{abstract}
Intermittently powered devices enable new applications in harsh or inaccessible 
environments, such as space or in-body implants, but also 
introduce problems in programmability and correctness.  
Researchers have developed programming models to ensure that programs make
progress and do not produce
erroneous results due to memory inconsistencies caused by
intermittent executions. 
As the technology has matured, more and more features are
added to intermittently powered devices, such as I/O. Prior work
has shown that all existing intermittent execution models have
problems with repeated device or sensor inputs (RIO). RIOs could leave
intermittent executions in an inconsistent state.
Such problems and the
proliferation of existing intermittent execution models necessitate a
formal foundation for intermittent computing. 

In this paper, we formalize 
intermittent execution models, their correctness
 properties with respect to memory consistency and inputs, and identify the invariants needed to prove systems
 correct. We prove equivalence between several existing
 intermittent systems. 

To address RIO problems, we define an algorithm for identifying variables
 affected by RIOs that need to be restored after reboot and
 prove the algorithm correct. Finally, we
implement the algorithm in a novel intermittent runtime system that is correct with respect to
 input operations and evaluate its performance. 

\end{abstract}





\maketitle

 \section{Introduction}
\label{sec:intro}

Battery-less, energy-harvesting devices (EHDs) are an emerging class of
embedded computing device that operate entirely using energy extracted from
their environment, such as light energy from a solar panel or energy from
radio waves using an antenna.
Free from a battery, these devices enable new applications in
IoT~\cite{pible,permamote,capybara,flicker}, civil infrastructure
sensing~\cite{camaroptera}, in-body medical sensing~\cite{proteus}, and space
exploration~\cite{capybara,kicksat, oec-constellations}.  
We study EHDs that compute {\em intermittently} as energy is available.
The device slowly harvests energy into a capacitor. 
After storing sufficient energy to make meaningful progress, the
device operates, quickly consuming the energy.
After exhausting the stored energy, the device powers off, awaiting more
energy.  Software executes according to an {\em intermittent
execution model}, where programs make progress during an active period that is
preceded by and followed by an inactive recharge period and a
reboot~\cite{dino,alpaca,chain,ratchet,hibernus,quickrecall,mementos}. 
A reboot clears volatile state (registers and SRAM) and preserves non-volatile
state (FRAM~\cite{tiwolverine} and Flash). This execution pattern 
is illustrated in Figure~\ref{fig:gen-overview}~(a).

Unpredictably-timed power failures create several challenges for an
intermittent execution model, including how to maintain forward progress,
ensure memory consistency, manage I/O and concurrency, and correctly interface
with hardware. Figure~\ref{fig:gen-overview}~(a) illustrates how executing a program intermittently 
can result in a memory state inconsistent with any continuous execution of  
a program, if the program contains certain memory access patterns or repeats input operations. 
As energy-harvesting devices have matured, an increasing variety of new
programming models and runtime systems for intermittent execution,
with varying technical approaches to addressing these key challenges, have
emerged~\cite{dino,chain,alpaca,hibernus,
hibernusplusplus,chinchilla,coati,ratchet,quickrecall,mementos,sonic,
incidentalcomputing,whatsnext,capybara,samoyed, tics}. 

The proliferation of diverse intermittent execution models 
presents 
a developer with a confusing space of implementation options and raises the
question of how to specify and compare the behavioral properties of different
systems. 
Models differ subtly, and a program written for one model may make assumptions
not met by another. 
Moreover, no existing software or hardware system has
clear, formally defined behaviour. 
Such a characterization and formalism is a key building block for defining and
proving correctness properties, developing tools to find and fix
intermittence-specific bugs, and understanding the fundamental similarities and
differences between models.  A primary motivation for using intermittent computing systems 
is their deployability in remote, inaccessible environments. 
Updating the device may 
be difficult or even impossible, so a buggy runtime that corrupts the memory state
can make the device practically useless. The lack of specifications of intermittent systems
is a key impediment to their deployment, particularly in applications that demand high
reliability or security.

In this work, we lay the groundwork for provably correct 
intermittent computing by formalizing
the semantics of several classes of intermittent execution models, focusing on 
the foundational correctness issue of {\em ensuring memory consistency} in 
the presence of non-deterministic input operations. Figure~\ref{fig:gen-overview}~(b)
provides an overview of our contributions. An intuitive correctness property 
is that 
an intermittent execution's behavior should be equivalent to some
continuously-powered execution. Prior work~\cite{ibis} has shown
(confirmed by our formalism) that a majority of
existing intermittent systems {\em do not satisfy a reasonable correctness
condition} in the presence of input operations that may change as a program runs intermittently.
We articulate the changes to the execution model that are necessary to correctly handle
the behavior of input operations in an intermittent execution. 
We start by formalizing a {\em checkpoint-based intermittent execution
model} based on DINO~\cite{dino}.  A checkpoint model saves important state 
during execution and, after a power failure, restarts from that point
in the execution by restoring the checkpoint on reboot.  
We then formalize the 
behavior of variants of
checkpoint-based systems~\cite{ratchet}, including both redo- and undo-logging
checkpoint strategies, and a task-based execution model~\cite{alpaca}.  We
relate these systems via bi-simulation,
showing that the correctness properties of one system hold for the others.
Finally, we propose a new checkpoint-based intermittent execution model that provably
handles input behavior correctly, while allowing re-execution of inputs (crucial for data 
freshness). As far as we know, we are the first to
propose such an execution model. 

\begin{figure}[t]
  \centering
  \includegraphics[width=1.00\textwidth]{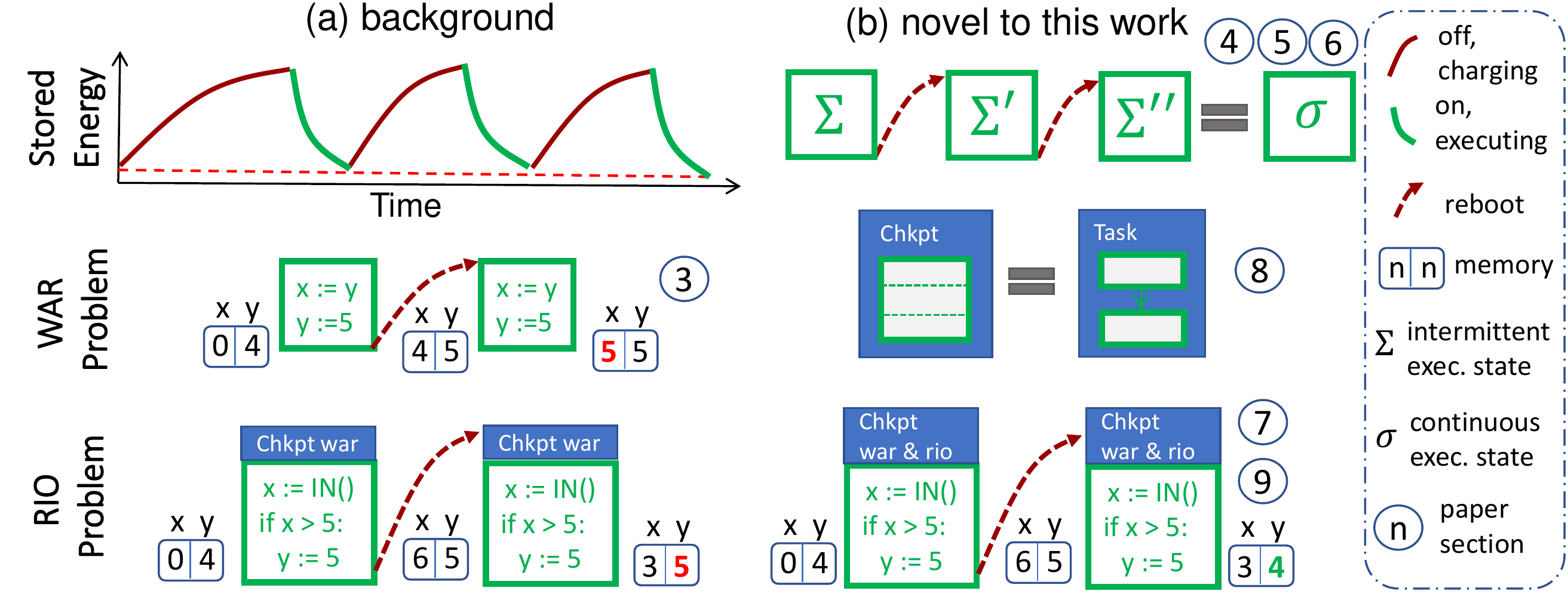}
\vspace{-20pt}
  \caption{(a) The intermittent execution model. Re-executions can result in an inconsistent memory state due 
  to write-after-read dependencies and repeated input operations. 
  (b) We provide a formal correctness theorem, prove equivalence between systems, and implement a runtime with a provably correct algorithm}
  \label{fig:gen-overview}
\end{figure}

We acknowledge that intermittent systems are similar to systems
that handle crashes and
failures in file systems or databases~\cite{idempotent,
  justdo, reachability}, but there are key differences.
As intermittent execution 
  targets highly resource-constrained, embedded systems, and failure is 
  inevitable and frequent, the checkpoint and 
  recovery mechanisms are typically implemented to save and restore the least 
  state possible. Additionally, the correctness of the system 
  depends on the state of both the volatile and non-volatile memory, and programs 
  are often driven by non-deterministic sensor inputs. Existing formalisms 
  for verifying crash consistency typically do not model the internal state of the recovery 
  mechanism~\cite{reachability}, which is necessary to capture the behaviour of the systems we model,
  model nonvolatile state only~\cite{yggdrasil, fscq},
   or do not explicitly model peripherals. These differences mean
    that existing automatic verifiers or logic frameworks~\cite{reachability, yggdrasil, fscq, ft-rr} 
    are not directly suitable for analyzing 
    intermittent execution models, though the correctness specifications 
    and invariants that we identify are useful for extending 
    such tools to work for intermittent computing.   
Identifying such invariants for the correctness proofs is nontrivial.  These
invariants deepen our understanding of intermittent systems. This work is the
foundation for formalizing more complex intermittent behaviour, 
including event-driven and reactive execution~\cite{coati,ink}, or guaranteeing 
forward progress. 
We make the following contributions:

\begin{itemize}[topsep=0pt,leftmargin=.1in] 
\item A novel, formal intermittent execution model with inputs,
  a correctness theorem, and sufficient conditions for
  correctness. 
\item A provably sound algorithm to collect variables to checkpoint
  and an implementation for an existing intermittent system.
\item Formalized variants of intermittent-execution models and
  proofs of the equivalence of these systems via bi-simulation.
\item An experimental evaluation of our algorithm implementation on real hardware showing that our 
technique has low time and space overheads while requiring little to no programmer effort. 
\end{itemize}

Due to space, we relegate detailed formalism and proofs to the 
appendices.

 \section{Scope and Related Work}
Fully specifying intermittent system behaviour requires reasoning 
about diverse properties. These properties include memory consistency, 
forward progress, timeliness, and correct handling of concurrency. 
Our paper addresses 
memory consistency. An intermittent program that is guaranteed to finish and
 always processes only fresh data is incorrect if it operates on 
 a memory state inconsistent with any continuous execution of the program.
  While we focus on memory consistency, we can 
 extend the formalism to cover other correctness properties in the future. 
 In the remainder of this section, 
 we describe these other properties at the high level and sketch 
 what extensions are necessary for our framework to capture these properties. We 
 then discuss how our framework relates to existing research, particularly 
 in intermittent computing and verified crash consistency. 

\subsection{Scope: What the Paper is \emph{Not} About}
\label{sec:other-props}
We focus on memory consistency with re-executed inputs. Reasoning about other 
desirable properties presupposes that the underlying 
system memory is correct. As intermittent applications are often 
sensor-driven, ensuring correctness with input operations is paramount. 
The full set of properties mentioned above is a long term goal that 
  can be reached by building on top of our current framework. 
  Each additional piece requires nontrivial 
  theoretical and implementation components.

\paragraphb{Forward progress} To make forward progress, a program 
executing intermittently must be able to execute 
the region between any two adjacent checkpoints (or any task) with the amount 
of energy in the device's buffer. Otherwise the program will get 
stuck, partially executing the region, recharging energy, and rebooting forever. Current intermittent 
systems assume that the largest task or checkpoint region will be cheap enough 
to finish~\cite{ratchet, alpaca}. Our correctness theorem is sound relative to this assumption and 
does not itself prove forward progress. 
Formalizing and {\em guaranteeing} forward progress 
requires a persistent energy model, which is likely to be complex because the 
amount of energy a sequence of instructions consumes depends on the 
state of the entire board, not the processor alone. A region between checkpoints could finish 
on a processor in isolation,
 but may not if, e.g., the radio is enabled.  
 CleanCut~\cite{cleancut} is 
 a compiler tool that provides a probabilistic energy model to guide programmers 
 in sizing tasks, but it offers no guarantees and does 
 not consider the full state of the board. Samoyed~\cite{samoyed} 
 allows programmers to specify cheaper alternatives to algorithms that 
 the system can switch to at runtime, if it seems a program is not making progress.

\paragraphb{Timeliness} As power can 
be off for an arbitrary period of time, 
sensor data collected before a power failure can be stale and useless after a reboot.
 To avoid processing stale 
data, prior systems have either required external persistent timekeepers
~\cite{reliable-time, remanence} so that a programmer 
can specify explicit timing annotations that the system can check at runtime 
~\cite{mayfly, tics}, or required that the programmer place sensor calls 
and uses requiring fresh data in the same checkpoint region or task. We 
assume the latter approach in this work, and find that while it allows 
timely processing of data, it also introduces memory inconsistencies, making current 
systems that take this approach incorrect. Our formalism
aids us in developing a runtime that allows consistent re-execution of inputs. 

Guaranteeing timely consumption of data requires 
additional language and type constructs to specify which inputs and uses are time-critical, along with 
 either static checking algorithms to disallow programs that may incorrectly 
consume stale data or additional runtime mechanisms to 
ensure that a program will not consume stale data. 

\paragraphb{Concurrency} While most current 
intermittent systems use single-core micro-controllers and have no parallelism,
 recent work~\cite{coati} supported interrupt-based concurrency through transactions. 
 Modeling concurrency requires modeling interrupts
and asynchronous events and updating the language with 
synchronization commands. 

\subsection{Related Work}
The ideas presented in this paper are related to work in 
intermittent systems, fault tolerance 
and crash consistency in files systems, and formal persistent memory models.
We first discuss the most related works in verified crash consistency and persistent memory models, 
and then how our work relates to existing intermittent systems, particularly 
those that deal with inputs or reactivity.

\noindent{\bf Crash consistency}
The failure and recovery problems of 
intermittent systems are similar to those of crash 
consistency on concurrent programs and file systems. 
A file system can crash at any time and must not exhibit 
unspecified behaviour after recovering. Developing formal specifications and verifiable file 
systems is an important research goal~\cite{jpl-mgc}, particularly to guarantee correctness in the presence of 
crashes~\cite{fscq, fs-cc-models, 
yggdrasil, ft-rr, flashix, flashix2}.  

Bornholt et al.~\cite{fs-cc-models} create a framework 
for generating crash consistency models. A crash consistency 
model specifies the allowed behaviour of a file system across crashes. 
Their 
crash consistency theorem relates the crashy fs trace to a canonical program 
trace with no crashes. In contrast, as we model non-deterministic sensor inputs, there 
is no single canonical trace even for executions with no crashes.  
The model consists of litmus tests 
and an operational semantics of the file system that 
models both volatile core state and durable disk state. Our semantics additionally 
model checkpoints. 

The 
verification tool Yggdrasil~\cite{yggdrasil} uses crash refinement 
to aid programmers in developing verified file systems. 
Programmers must write specification and 
consistency invariants of their system. Then the 
verification is process is modular, allowing developers to swap in different 
implementation of system components as long as they meet the specification.
The focus of our work is on defining correct specifications of intermittent 
system behaviour, including whether different implementations are in fact equivalent.
While our correctness theorem is similar to crash refinement, we do not use 
Yggdrasil to verify our specifications as Yggdrasil uses file system abstractions, 
e.g., inode layouts and disc models, that don't apply to intermittent systems, 
which interact directly with memory.
We additionally model the effects of inputs. 

Crash Hoare Logic (CHL)~\cite{fscq} and fault-tolerant resource reasoning~\cite{ft-rr} are proof
 automation tools that extend Hoare triples with crash 
 conditions to verify file system implementations.
Using CHL has a high programmer proof burden because the 
programmer must specify the correctness invariants and recovery 
procedures and prove the recovery procedure 
correct. Thus, a large portion of our work is a prerequisite to using CHL; we 
define intermittent correctness invariants, which is non-trivial. We additionally show that existing recovery procedures 
are in fact {\em in}correct. Using CHL for proof automation once we have defined intermittent correctness is not 
immediately possible as 
 CHL  
 does not provide primitives for checkpoints and only 
 explicitly models non-volatile state, not the mixed-volatility state typical on 
 an intermittent system.
Additionally, the crash conditions for the Hoare triples should capture the intermediate states 
at which a crash could occur, and must be specified for every procedure. For a 
set of file system procedures, capturing these intermediate states is not onerous, 
as each procedure interacts 
 with only a few blocks of the disk. In contrast,  we model intermittent 
 execution traces of programs, which makes enumerating crash conditions complicated 
 and time-consuming. 
Ntzik et al. \cite{ft-rr} 
 do consider both volatile and non-volatile
 resources, though they 
 also require enumerating the non-volatile states possible 
 after a procedure crashes. 
 The authors use their framework to prove the 
 soundness of an ARIES recovery mechanism. They model updates to 
 the durable state at page granularity, and undoing a transaction 
 requires rolling back all updates to pages modified by the transaction. 
 In contrast, the intermittent systems we model are designed to roll back 
 the minimum set of updates necessary to (ostensibly) guarantee correctness. 
 Neither of these frameworks model non-deterministic sensor inputs.

Unlike the works above, which 
deal with verifying the file system itself, Koskinen et al. 
~\cite{reachability} automatically verify crash 
recoverability at the program level. Our approach is most 
similar to this work. 
Our correctness theorem is similar, defining correctness in 
terms of a simulation relation 
and observational equivalence between the continuous and intermittent executions.
In contrast to our work, this model assumes
that underlying system operations will be correct. Their method 
analyzes control-flow and reduces crash recoverability 
to reachability: if control-flow cannot reach an error state, it will be correct. The definition 
of a recovery
mechanism is that given a state $q_k$ in the 
original program, after a crash the mechanism brings the program back 
to $q_k$ after transitioning through some recovery states. While this is clearly the desired behaviour 
of a recovery mechanism, this definition does not consider 
the internal state of the recovery mechanism and is not expressive enough to capture 
the behaviour of existing intermittent checkpoint systems.
After fully executing the recovery procedure, an intermittent program is not in an 
equivalent state to before the crash.  Identifying 
what differences are allowable in the recovered state so that 
further execution eventually brings the program to a consistent state is a key contribution of this work.

\noindent{\bf Persistent memory formalisms}
Persistent memory has been used for whole systems~\cite{wholesystempersistence}, 
entirely non-volatile 
processors~\cite{nvp,nvp2}, and heap structures~\cite{nvheaps, mnemosyne}.
There are formalisms exploring persistency models~\cite{mp,mp2} for 
reasoning about data on non-volatile systems and parallel persistency~\cite{parallel-pm}. Other work looks at defining 
linearizability~\cite{persistent-linearizability} for persistent objects on concurrent systems.  
While these are useful correctness properties, current intermittent hardware is single-core and has no thread-level concurrency.

Weak persistency semantics have been formalized for TSO memory models~\cite{raad1}, for 
ARMv8~\cite{raad2}, and for Intel x86~\cite{raad3}. 
In~\cite{raad2}, 
the authors introduce a declarative semantics for reasoning about persistency. Among memory persistency 
formalisms, our approach is most similar to this one, but we are at a higher level; there are differences in 
scope and the language features provided. 
These persistency models reason about the allowable 
differences between the order in which instructions execute and 
the order they persist to memory on multi-threaded programs. This scope introduces (needed) 
complexity into the declarative semantics, but the devices we target do not 
have multi-threading and expose no difference between execution and persist order. We do 
not currently benefit from this complexity, but in future work we may need to integrate 
with these models to guarantee assumptions, e.g., checkpoint atomicity, that are currently 
upheld by the simple hardware. 
Moreover, our modeling language provides inputs and checkpoints.

\paragraphb{Runtime systems for crash consistency}
Runtime systems that attempt to provide crash consistency 
on database systems have similar functionality to runtime systems for 
intermittent execution, but generally do not provide re-execution of inputs, 
necessary for data freshness.

JustDo logging~\cite{justdo} targets hybrid persistent systems. JustDo
 explicitly avoids 
re-executing code for better performance. 
iDo~\cite{ido}, also targeting hybrid systems, identifies idempotent
instruction sequences to reduce the number of locations to be logged.
Idempotence has also been used as a correctness criterion for fault
tolerance in distributed systems~\cite{ft-idempotent}. Idempotent
processing~\cite{idempotent, idempotent2} has been posed as an
alternative recovery mechanism to checkpoint-logging and
re-execution, but does not allow re-executing inputs, which
sometimes is necessary for intermittent systems to provide fresh sensor readings.
\limin{check the previous sentence}

Other work~\cite{delay-free} provides a construction to 
automatically make accesses to shared 
memory and algorithms persistent. 
Intermittent systems need all executing code to be 
checkpointed or in transactions, not just 
shared data structures.

\noindent{\bf Runtimes for Intermittent Systems}
In this paper, we explicitly model DINO~\cite{dino} as a basic checkpointing system, 
Ratchet~\cite{ratchet} for the idempotent region variant, Alpaca~\cite{alpaca} as an example 
of task-based redo logging, and Chinchilla~\cite{chinchilla} for undo logging. 
Hibernus~\cite{hibernus, hibernusplusplus} is a 
\emph{just-in-time} checkpoint system that dynamically inserts checkpoints and does not re-execute code, but suffers timeliness violations.
Mayfly~\cite{mayfly} is the first work to describe the timeliness
problem and implements a 
programming model 
to enforce 
timeliness 
using an external timekeeper 
and explicit programmer annotations. Capybara~\cite{capybara} is a
reconfigurable energy-harvesting platform that allows 
 flexible
atomicity and reactive events. Homerun~\cite{homerun} also
explores atomicity for I/O events.  Coati~\cite{coati} and
InK~\cite{ink} explore event-driven intermittent
systems.  None of these works provide formal definitions or
guarantees of correctness for either memory consistency or timely 
processing of inputs.

EDB~\cite{edb} and Ekho~\cite{ekho} are frameworks for debugging intermittent
systems, and ScEpTIC~\cite{sceptic} is a tool for detecting bugs caused by
write-after-read patterns. 
 The EH
Model~\cite{ehmodel} provides a way of reasoning about the architectural and
software consequences of energy availability and intermittent system design
choices.

Dahiya et al.~\cite{avis} create a formal model for verifying 
via translation validation that 
instrumented intermittent programs are equivalent to continuous ones. 
They do not consider repeated input operations or how checkpoints must behave 
for  programs to be correct.

\noindent{\bf Inputs on intermittent systems}
IBIS~\cite{ibis} identifies and characterizes bugs caused by repeated inputs in 
intermittent systems.
The authors provide only a bug detection tool, not a correct runtime system, 
nor formal correctness invariants. We provide a formal proof of a sound 
version of the algorithm the authors use in their tool, as well as a correct 
runtime system. We discuss in detail the differences and similarities of the 
algorithm presented in this paper versus the algorithm in the IBIS tool in 
Section~\ref{sec:ibis-diff}.

Developed most recently, TICS~\cite{tics} is a runtime system that uses an external timekeeper 
and programmer annotations to avoid consuming stale data. Rather than 
regather data, the runtime reruns expiration checks after rebooting, 
so that any stale data will not be processed. This approach avoids 
any consistency errors associated with re-executing inputs, but can 
also miss processing any input events if power failures are frequent.
Moreover, this approach requires external time-keeping hardware.

Samoyed~\cite{samoyed}, Sytare~\cite{sytare} and
RESTOP~\cite{restop} look at retaining the peripheral state of input
devices, not at memory correctness issues caused by repeated
input operations.

 \section{Background and Motivation}
\label{sec:overview}

Our work is motivated by emerging intermittent execution models and their
 varied correctness definitions. 
In this section, we review the fundamentals of checkpoint-based intermittent execution and 
show by example how existing models are not correct in the presence of input operations.   
Any intermittent execution model must ensure forward progress and preserve
state. An intermittent execution progresses only when energy is available.
Power fails when energy is exhausted, erasing the device's execution context
and volatile state, including registers and all data stored in volatile
memory.  By default, the system then restarts from the start of {\tt main()}
and naively-written code makes no forward progress. 
To make progress, an
intermittent system can periodically save its
execution context and restart from that execution context on 
reboot; a common mechanism for saving state is a statically placed {\em checkpoint}. 
A checkpoint is an operation that stores some memory state and execution 
context in non-volatile memory, preserving it across a power failure~\cite{ratchet,dino, mementos, 
hibernus, hibernusplusplus,quickrecall,idetic,chinchilla}.

To be correct, an intermittent execution should generate the same result as a
continuous execution. Multiple partial executions followed by a complete execution should have the same
behaviour as {\em some} continuous execution.  Code between checkpoints must execute
\emph{idempotently}.  Unfortunately, a checkpoint system that saves only
volatile execution context and volatile memory~\cite{mementos} may be
incorrect.
Prior work~\cite{dino,ratchet} identified that write-after-read (WAR)
dependencies on non-volatile locations cause non-idempotent behavior and adjust
the checkpointed data accordingly.  We next discuss how WAR dependencies cause
problems.

\subsection{Write-After-Read (WAR) Dependencies}
\label{sec:overview-war}

An intermittent execution may produce an incorrect result if the execution
writes a value into non-volatile memory before a power failure and the
execution reads that updated value after rebooting. 
Consider the example shown in Figure~\ref{fig:rio-example} (a). On the left of
the figure is a small program. The next column shows execution traces
illustrating the WAR problem.  In the initial execution, the branch at line 1
is taken.  After executing through line 4, there is a power failure. The
column then shows the re-execution of the code.  The re-execution assumes that
the checkpoint restores only volatile state (i.e., control state) and retains
non-volatile variables' values.  This time, the
execution completes and yields state $N'_4$. A continuous program
would finish with the memory in state $N_4$.
The state in $N'_4$ contains a different value
for $x$.  The re-execution does not idempotently update $x$ because $x$ depends
on $w$.
When the re-execution reads $w$ into $x$ at line
2, the read produces the (incorrect) value of $w$ written before the power
failure. 
The example shows that the re-execution is non-idempotent
because $w$ is involved in a WAR dependence, which we call a {\bf WAR Variable}.  

\newcounter{linenum}
\newcommand{\refnum}[1]{\stepcounter{#1} \arabic{#1} }

\setcounter{linenum}{0}
\begin{figure*}[tbhp]
\centering
\includegraphics[width=1.0\textwidth]{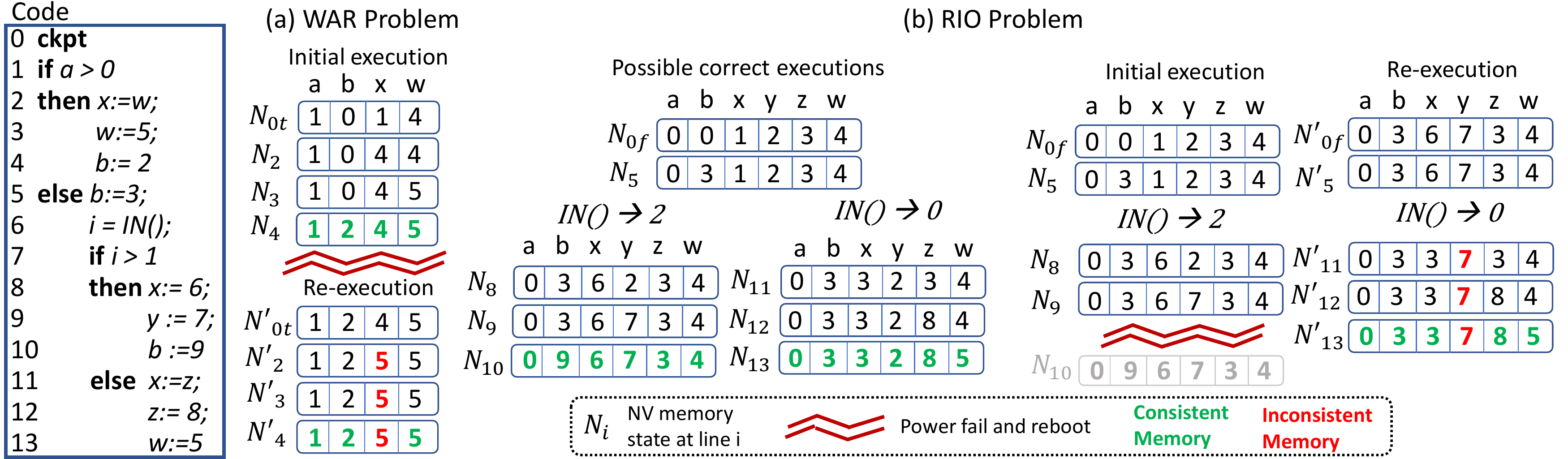}
\vspace{-10pt}
\caption{An example program illustrating WAR and RIO problems}
\label{fig:rio-example}
\end{figure*}
To ensure that code containing WAR dependencies executes idempotently, existing
checkpoint-based systems must add {\em potentially inconsistent} variables to
the checkpoint~\cite{dino,chinchilla} and restore those variables with the
checkpoint after a power failure. 
In the example, $w$ is potentially inconsistent because a write to $w$ before the
power failure may be visible to a read after restarting from the checkpoint. At
line 0, a checkpoint system needs to save a {\em version} of $w$ with the
checkpoint.
Prior work observed that not all WAR dependences lead to
inconsistency~\cite{ratchet,dino,alpaca}.  If $w$ had been written before being
read, e.g. if line 1 was instead $w:=3; \m{if}~a>0$, then $w$ need not be
checkpointed.  The read on line 1 would always see a consistent value.

Checkpointing WAR variables that are not write-dominated is the
current state-of-art in ensuring idempotent re-execution, which 
is \emph{insufficient} for correctness, due to RIOs~\cite{ibis}.

\subsection{Repeated Inputs Cause Incorrect Behavior}

\label{sec:overview-rio}
Applications that target low-power embedded systems rely heavily on peripheral
devices, such as sensors and radios.  
A program stores in a variable the result of an input operation.  In an
intermittent execution, a program may execute an input operation before a power
failure, and then repeat that input operation after a failure, in both cases
fetching a fresh, usable value.  However, repeating the input operation can lead to
incorrect behavior when a program's control- or data-flow depends on the result
of that input operation. We refer to a re-executed input as a 
{\bf RIO}: a {\bf R}epeated {\bf I}nput {\bf O}peration that can generate a different 
value each execution. 

Continuing with Figure~\ref{fig:rio-example} (b), the columns to the right show
how a RIO causes incorrect behavior.  The example now assumes that all
variables are in non-volatile memory and that the system now checkpoints
variables involved in potentially inconsistent WAR dependencies at line 0,
saving and restoring $w$ and $z$.  The starting memory $N_{0f}$ has $a \mapsto
0$, so the initial branch is not taken.  Instruction $\mathit{i} = \m{IN}()$;
on line 6 reads an input value (e.g., from a sensor). Depending on the sensor
reading, a continuous execution could correctly end with either state $N_{10}$ or
$N_{13}$.

An intermittent execution may produce a result different from both 
$N_{10}$ and $N_{13}$.
Such an execution may first get an input greater than $1$ at line 6, causing
the branch at line 7 to be taken.
Power then fails and the program restarts. After the restart, the input is
less than $1$, and the branch at line 7 is not taken.
The final state is $N'_{13}$, which is inconsistent with all correct,
continuously-powered outcomes because the RIO is not idempotent. 

An intermittent execution with inputs is correct if it corresponds
to a continuously-powered execution, regardless of the inputs.  Here, the
RIO causes different branches to be taken and non-volatile variable $y$ is
written on only one of them.  Checkpoint systems that version WAR  variables
do not handle $y$'s RIO problem because $y$ is not a WAR
variable. {\em No existing checkpoint or task-based intermittent execution system
correctly handles these non-idempotent RIOs.}
A correct checkpoint system must store $y$'s value at the checkpoint and
restore it on reboot. We refer to variables that are affected by RIOs
(e.g., $y$) as {\bf RIO variables}.
Preventing inputs from re-executing, as by placing a 
checkpoint immediately after the operation, is not an adequate solution 
as some inputs must be re-executed to be timely~\cite{mayfly}. 
This paper fills the gap left by RIOs, formally and with a practical
system implementation that makes intermittent systems robust to RIOs.

 \section{System Assumptions and Formal Model}
\label{sec:framework}

We define a language to 
model checkpoint-based intermittent execution with inputs, providing the syntax and semantics for both
continuously-powered and intermittent executions. 
First, we explain the lower-level system assumptions to justify our choice of modeling language. 

\paragraphb{Target System Assumptions} Our target intermittent systems use
low-end microcontrollers (MCUs) such as the TI MSP430FR
series~\cite{tiwolverine}. These are single-threaded, single-issue, in-order
compute cores.  The MCUs have embedded, on-chip volatile SRAM or DRAM and non-volatile
Flash, FRAM, or STT-MRAM.  These architectures often lack caches or have only a simple
write-through cache to avoid repeated non-volatile memory accesses. Unlike prior work in
persistent memory targeting more complex
architectures~\cite{raad1,raad2, raad3,persistent-linearizability, parallel-pm}, we
need not reason about concurrency or persist order due to
write-back caches or other microarchitectural optimizations.  We thus
realistically assume that execution and persist order are the same and that the
compiler never re-orders an instruction past a checkpoint. 

\paragraphb{Syntax} 
Our simple language includes accesses to volatile memory, 
accesses to non-volatile memory, and branch statements. We
include arrays but omit general pointer arithmetic. We also omit functions
calls and unbounded loops. 
These omissions do not affect our ability to capture
 the behavior of existing intermittent execution models. Existing systems~\cite{dino, alpaca} do not 
allow recursive function calls, so any code in a function body can be inlined. Including general pointer arithmetic would not change the correctness invariants we 
present, as the definitions consider memory locations directly, but would complicate the 
implementation of any checkpoint algorithm (discussed in Section~\ref{sec:implementation}), as the alias sets of the memory locations in the definitions 
would need to be tracked as well. 
Unbounded loops can be handled by extending our infrastructure with 
loop invariants, which do not introduce technical difficulties but 
unnecessarily complicate the presentation.
Though simple, this modeling language suffices to illustrate clearly the key challenges in defining memory-consistent
intermittent execution models. 
 
 We summarize the syntax in Figure~\ref{fig:syntax}. 
\begin{figure}[tbp]
\centering
  \small
\(
\begin{array}{llcl@{\quad}llcl}

\textit{Values} & v & \bnfdef & n \bnfalt \m{true} \bnfalt \m{false} \bnfalt \m{in}(\timestamp)

& 
\textit{Configuration} & \Sigma& \bnfdef & (\timestamp,
                                                      \context,
                                                      \nvmem, \vmem, \cmd) 
\\

\textit{Expressions} & e & \bnfdef & x \bnfalt v \bnfalt e_1\;
                                     \m{bop}\; e_2 \bnfalt a[e']

&\textit{Cont. config.} & \sigma& \bnfdef & (\timestamp,\nvmem, \vmem, \cmd)

\\

\textit{Instructions} & \iota &  \bnfdef&  x:= e
                                          \bnfalt a[e]:= e' \bnfalt x:=\m{IN}() \bnfalt
                                          & \textit{Volatile mem.} & \vmem & : & M
 \\
&& & \m{skip} \bnfalt \m{checkpoint}(\omega) \bnfalt \m{reboot}(n)

& \textit{Non-vol. mem.} & \nvmem & : & M
\\
\textit{Commands} & \cmd & \bnfdef& \iota  \bnfalt \iota;\cmd \bnfalt
                                    \m{if}\ e\ \m{then}\ \cmd_1\
                                    \m{else}\ \cmd_2
& \textit{Context} & \context & \bnfdef & (\nvmem, \vmem, \cmd)
\\
\textit{Memory loc.} & \loc & \bnfdef & x\bnfalt  a[n] 
& \textit{Read obs.} & r & \bnfdef & \m{rd}\ \loc\ v \bnfalt  r, r
\\

\textit{Chckpnted loc.} & \omega & \bnfdef & \omega, x,\bnfalt \omega, a^n
&  \textit{Observation} & o & \bnfdef & [r] \bnfalt \m{in}(\timestamp) \bnfalt \m{reboot}                                               
\\
\textit{Mem. mapping} & M & \bnfdef & \m{Loc} \rightarrow \m{Val}
& & & & \bnfalt \m{checkpoint} 
\end{array}
\)
\vspace{-10pt}
\caption{Syntax and Semantic Constructs}
\label{fig:syntax}
\end{figure}
We write $v$ to denote values, which can be
numbers $n$, the boolean values $\m{true}$ and $\m{false}$, and inputs $\m{in}(\timestamp)$, 
representing the input gathered at time $\tau$.
Expressions, denoted $e$, can be variables, values, binary operations
of expressions, or an array element.  Array lengths are fixed 
and all array indices are assumed in 
bounds; this assumption is necessary for correctness and memory safety for real 
C code is orthogonal~\cite{cyclone-regions-pldi02} and beyond our scope.
Instructions, denoted $\iota$, consist of assignments to 
variables and arrays, checkpointing, rebooting, $\m{skip}$, and synchronous input operations $\m{IN}()$. We write $\omega$ to denote the set of
non-volatile variables and arrays that must be saved with a
checkpoint to avoid inconsistency. 
We call these variables \emph{checkpointed locations}. In the example
in Figure~\ref{fig:rio-example}, $\m{checkpoint}(\{w,y,z\})$ would
precede the if statement on line 1, which include both WAR and RIO
variables. \limin{check if these variables are correct} \brandon{Why
  don't we put the variables in the ckpt expression in the figure?  It
  would be clearer and we could make direct reference.  I suppose when
  we add rio, then the checkpoint has to change, so that's one reason
  not to.} We write $a^n$ to represent all the locations in the array
$a$. That is: each $a^n$ in $\omega$ represents the set of locations
$\{a[1], \cdots, a[n]\}$. We often omit the bounds $n$ and write $a$
directly. We assume that checkpoint operations are manually inserted
into code (e.g., like DINO~\cite{dino}).  Appendix~B.2
 details the algorithm to
compute $\omega$ for WAR variables, as in existing systems.
Section~\ref{sec:rio-collection} describes our novel algorithm for
computing $\omega$ for RIOs.
A program is a command $\cmd$, which is an atomic instruction, a
sequence of instructions, or an if branching statement. We lift all the
branches to the top-level for ease of explanation. Any program with general
branching statements can be re-written to our language and bounded loops can be un-rolled to
{\tt if} statements. 

\paragraphb{Semantics for Intermittent Execution} 
We focus on intermittent execution semantics. The rules 
for continuously-powered execution semantics are standard and can be found in Appendix~A.2.  
 First, we
define the necessary runtime constructs in Figure~\ref{fig:syntax}. Memory is a mapping
from a location, which is either a variable or an array index, to a
value.  We distinguish between volatile and non-volatile memory, which are disjoint. 
The method of specifying where a variable resides varies; systems may provide 
abstractions~\cite{alpaca,capybara}, do automatic compiler analysis~\cite{chinchilla}, or assume that all data is non-volatile~\cite{ratchet}.
A configuration $\Sigma$ is a tuple consisting of 
a
timestamp $\tau$, a checkpoint context $\context$, 
 non-volatile memory state, volatile memory state, and a command to be
 executed $\cmd$. 
The checkpoint context $\context$ consists of the non-volatile data, volatile data, and 
command saved at the last checkpoint. The timestamp is the logical
time at the current configuration. 
Executing commands and evaluating expressions generate observations, 
which are memory reads, input reads, or occurrences of a checkpoint or reboot instruction. These observations 
are used for facilitating definitions of correctness.

The semantic rules are of the form:
$(\timestamp, \context, \nvmem, \vmem, \cmd) \Stepsto{O}
(\timestamp',\context', \nvmem', \vmem', \cmd')$, where $O$ is a list
of observations. We write $\Stepsto{}$ to denote an intermittent execution, 
and $\SeqStepsto{}$ to denote a continuous execution. 
We show rules specific to intermittent execution 
in Figure~\ref{fig:semantic-rules}.
\begin{figure}[t]
\begin{mathpar}
\small
\inferrule[I/O-CP-CkPt]{ }{
  (\timestamp,\context, \nvmem, \vmem, \m{checkpoint}(\omega);\cmd) 
 \Stepsto{\m{checkpoint}} (\timestamp + 1, (\proj{\nvmem}{\omega}, \vmem, \cmd), \nvmem, \vmem, \cmd) 
}
\and
\inferrule[I/O-CP-PowerFail]{ \m{pick}(n)}{
  (\timestamp, \context, \nvmem, \vmem, \cmd) 
 \Stepsto{}  (\timestamp + 1, \context, \nvmem, \resetm(\vmem), \m{reboot}(n)) 
}
\and
\inferrule[I/O-CP-Reboot]{ \context=(\nvmem, \vmem, \cmd)}{
  (\timestamp, \context, \nvmem', \vmem', \m{reboot}(n)) 
 \Stepsto{\m{reboot}}  (\timestamp + n, \context, \nvmem'\lhd\nvmem, \vmem, \cmd) 
}
\end{mathpar}
\vspace{-10pt}
\caption{Selected semantic rules}
\label{fig:semantic-rules}
\end{figure}
The rule \rulename{I/O-CP-CkPt} states that at a checkpoint, the
context $\context$ is updated to include the portion of the current
non-volatile memory whose domain is $\omega$, the current state of
volatile memory, and the current command to be executed. 
We write $\proj{m}{\omega}$ to denote the part of $m$, whose domain is
 $\omega$. 
The system generates the
$\m{checkpoint}$ observation and proceeds to execute the command
after the checkpoint. This operation is assumed to be atomic, implemented 
 at the low-level with an atomic flag update and double buffering to ensure there is always a valid last checkpoint.
The \rulename{I/O-CP-PowerFail} can execute at any
step, as power failures can occur at any time, and is thus non-deterministic.
When power fails, volatile memory is reset, the current
command is lost, and a positive integer $n$ is picked at random for the subsequent
reboot instruction. The system then transitions to reboot.
On reboot, \rulename{I/O-CP-Reboot}
partially restores non-volatile memory using the checkpoint. This rule applies even 
before the first checkpoint is reached 
as $\context$ is a piece of memory initialized with the starting $\vmem$ and $\cmd$ and empty non-volatile portion.
We write $m_1\lhd m_2$ to denote the memory
resulting from updating locations in $m_1$ with the values of those locations in $m_2$.
The reboot restores volatile memory and the command to 
the values in the checkpoint's context. 
A $\m{reboot}$ is added to the observation sequence, and the timestamp increases to
$\timestamp + n$. The increase captures the idea that a power failure can have an
arbitrary duration. 

We write $N, V \vdash e \Downarrow_{r} v$ to mean that with memories $N$
 and $V$, expression $e$ evaluates to value $v$ with observation
 $r$.  For example, $N_{\mt{0f}}, V \vdash w \Downarrow_{\m{rd}\, w\,
   4} 4$ is a sub-derivation when line 2 of the program in
 Figure~\ref{fig:rio-example} is executed.
The rules are standard, so we omit them for space.

 \section{Formally Correct Intermittent Execution}
\label{sec:inputs}

\begin{figure}[t!]
  \centering
  \includegraphics[width=\textwidth]{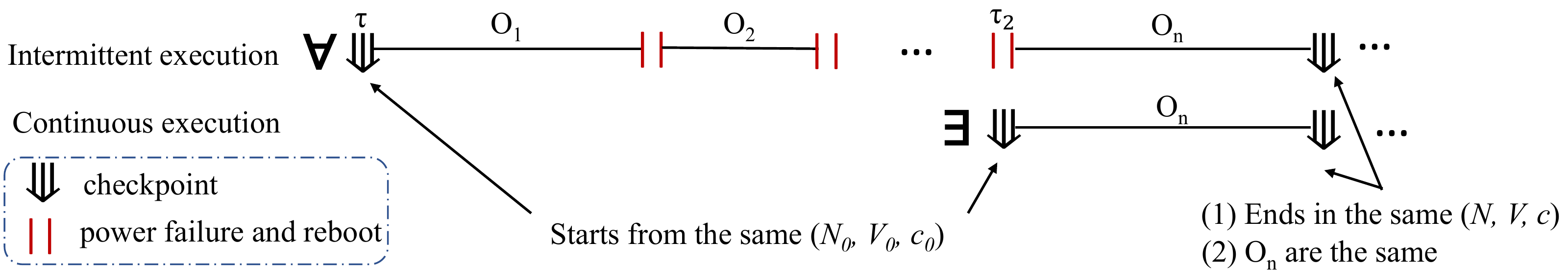}
\vspace{-20pt}
  \caption{Illustrating the correctness definition}
  \label{fig:correctness-rio}
\end{figure}
A program can be executed correctly in an intermittent model if and only if
any completed intermittent execution trace of the program corresponds to a
continuous execution, w.r.t. the program context and the observation sequence. 
More precisely, comparing an intermittent and continuous execution, the program 
contexts, including volatile and non-volatile memory and the command to be executed 
must be the same by the end of the program. To show this, we examine execution 
segments between checkpoints.
Each partial re-execution
from a checkpoint can observe a different value produced by the same
input.  Consequently, observation sequences from partial executions are {\em
not necessarily} prefixes of the same continuous execution. 
To be correct, the observation sequence of the \emph{ final} re-execution segment in an
intermittent execution must match a continuous execution with the same input
results. 
We illustrate this in Figure~\ref{fig:correctness-rio}. Time advances
from left to right. The top line  is an intermittent
execution trace. We detail a segment between two checkpoints, marked
by down arrows. Multiple power failures and reboots are present in
these segments, demarcated by red parallel bars. The observed memory reads
$O_i$ are shown on top of the line. For each such intermittent
execution, the correctness property dictates the existence of a
continuous execution---the second line---such that the read accesses 
from the latest reboot to the checkpoint ($O_n$) match the read accesses
of that continuous execution. Furthermore, the ending configuration of both
executions at the checkpoint are the same (excepting
the extra context $\context$ in the intermittent configuration). 
The above holds for all execution segments, including the last.

To formalize this definition, we introduce additional notation and 
constructs to relate intermittent and continuous 
program contexts and observation sequences. We define the erasure of the configuration: $\erase{(\tau, \context,
\nvmem, \vmem, \cmd)} = (\tau,  \nvmem, \vmem, \cmd)$ to 
relate the configurations at the checkpoints. 
We next formally relate the observation sequences.
\begin{mathpar}
  \small
\inferrule[I-Rb-Base]{ }{
  O \oleqm O
}
\and 
\inferrule[I-Rb-Ind]{ O_1' \oleqm O_2}{
  O_1, \m{reboot}, O_1' \oleqm O_2
}
\and
\inferrule[Cp-Base]{ O_1\oleqm O_2}{
  O_1 \oleqmc O_2
}
\and
\inferrule[Cp-Ind]{O_1 \oleqm O_2 \\ O_1' \oleqmc O_2'}{
  O_1, \m{checkpoint}, O_1' \oleqmc O_2, O_2'
}

\end{mathpar}
The rules use $\oleqm$ and $\oleqmc$ to express the prefix requirements of the
observation sequence of the intermittent execution ($O_1$) to the observation of the continuous execution ($O_2$).  
$\oleqm$ expresses a relation 
between an $O_1$ that may include reboots to $O_2$, 
and $\oleqmc$ expresses a relation between an $O_1$
that may include both reboots and checkpoints to $O_2$. 
The crucial aspect of the observation relation is in rule \rulename{I-Rb-Ind}. An 
intermittent observation consisting of two observation prefixes $O_1$ and $O_1'$ separated by a reboot 
relates to the continuous observation $O_2$ if the latter prefix relates to the continuous observation. 
The intuition of this rule is that the observation prefix of an intermittent execution before a reboot 
may read values produced by input operations. After the reboot, the input 
operations may return different values, so the old prefix should be discarded. The 
intermittent and continuous executions need only agree on observations 
after the most recent reboot. 

 The correctness of an intermittent execution model is defined as follows:
\begin{defn}[Correctness of Intermittent Execution]
  \label{defn:io-correctness}
  ~\\
  A program $\cmd$  can be
 correctly intermittently executed if 
 for all $\timestamp$,
  $\nvmem$, $\vmem$, $O_1$ s.t.
  $(\timestamp, \emptyset, \nvmem, \vmem, \cmd) \MStepsto{O_1}
  \Sigma$, where the program in $\Sigma$ is $\m{skip}$ (i.e., the
  program terminated), 
 then $\exists O_2, \timestamp_2, \sigma$ s.t.
  $(\timestamp_2,\nvmem, \vmem, \cmd) \MSeqStepsto{O_2} \sigma$,
  $\timestamp_2\geq\timestamp$, $O_1\oleqmc O_2$, and
  $\sigma = \erase{(\Sigma)}$.
  \end{defn}

Figure~\ref{fig:correctness-rio-example} illustrates the relations in
the correctness definition by
revisiting the code from Figure~\ref{fig:rio-example}. Assume there is
a checkpoint immediately preceding the branch on $a$, which saves the state of
$\{w,y,z\}$.
Consider a power failure after the assignment to $y$ on line 9, which lasts for 4 timestamps. The
column on the left shows the intermittent execution state and the right shows
the continuous execution state, starting at a later time, time 8.  The final state of the executions ($N_5'$ and
$N_5$) are equal, despite differences in their execution paths and intermediate
states. Further, the observation sequences relate: $\m{checkpoint}, \m{in}(1), \m{reboot}, \m{in}(9), \m{rd}\,z\, 3 \oleqmc
\m{in}(9), \m{rd}\,z\, 3 $.  There happens to be no
reads before the reboot. If there were, they would not need to match
the reads on the right, as they are not in the final (successful)
re-execution. 
To ensure this correctness property, a checkpointed set must include
{\em both} WAR and RIO variables, which we will explain in the next section.

\begin{wrapfigure}{L}{0.53\textwidth}
  \centering
  \includegraphics[width=0.53\textwidth]{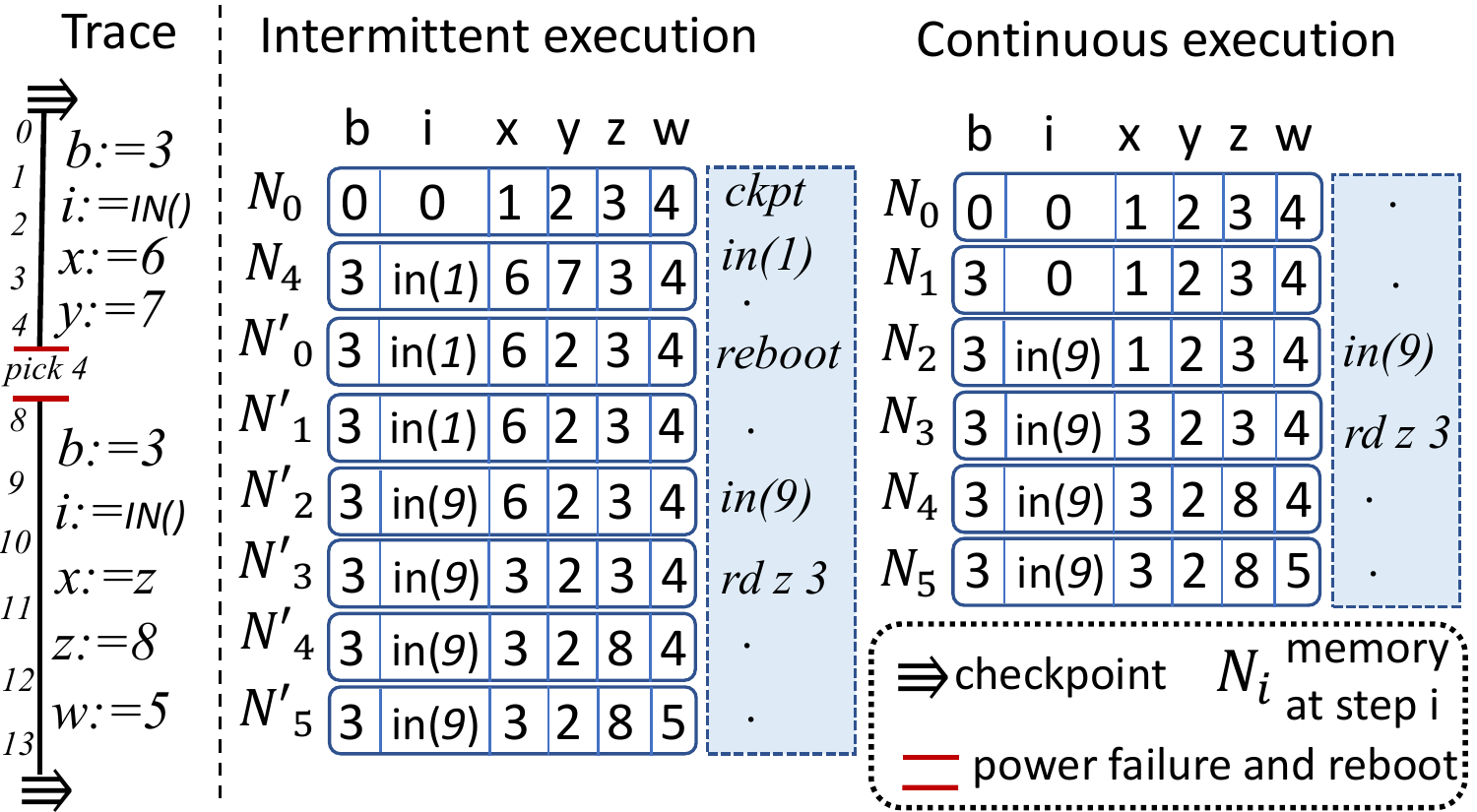}
 \vspace{-20pt}
  \caption{Illustrating the correctness definition by example}
  \label{fig:correctness-rio-example}
\end{wrapfigure}

\section{Proving Memory Consistency}
Non-volatile memory updates of an intermittent execution can diverge from non-volatile memory updates of a continuous execution in two key ways. 
A continuous execution updates memory through the writes of a single execution of the program. An intermittent execution 
updates memory through multiple executions of a prefix of the program and by overwriting some set of non-volatile memory locations after 
a reboot. To be correct, an intermittent execution model must ensure that any inconsistencies caused by these different update traces 
have resolved by the next checkpoint. In this section, we define for the first time invariants on the checkpointed set and non-volatile memory updates 
that allow us to prove an intermittent execution model correct.
A runtime system that upholds these invariants will provably update memory consistently, no matter the algorithm used in implementation.

\subsection{Locations to Checkpoint}
\label{sec:rio-check}

The correctness of an intermittent execution model requires restoring at reboot
a set of non-volatile memory locations $\omega$ (e.g., $\{y, z, w\}$ in Figure~\ref{fig:correctness-rio-example}). 
We refer to the minimal set of potentially inconsistent non-volatile
memory locations as $\omega^{\dagger}$ and observe that precisely computing this
set is difficult in general. 
It is safe to over-approximate $\omega^{\dagger}$.  
One safe over-approximation is all of
non-volatile memory, $\omega^\mt{all} = \m{dom}(N)$, but $\omega^\mt{all}$ is inefficient
because it requires unnecessarily checkpointing many variables.  Existing systems 
use
ostensibly less conservative over-approximations, such as a subset
of the variables
involved in WAR dependence,
$\omega^\mt{WAR}$~\cite{dino,alpaca,ratchet,clank,chinchilla}, but as 
we have shown, this set misses any variables that are inconsistent 
due to RIOs.   
Next, we describe how to to check that all potentially inconsistently written locations are
 checkpointed. We then define an algorithm to statically analyze the program
 and add those locations to the checkpoint in Section~\ref{sec:rio-collection}.

To express that a checkpoint includes 
the subset of potentially inconsistent variables, we introduce 
two judgments:  $\Vdash_\war c:\m{ok}$, which checks 
variables inconsistent due to WARs and $\Vdash_\rio c:\m{ok}$ 
which checks variables inconsistent due to RIOs.
We formally define rules for this judgment here. Rules for the WAR checking judgment are in Appendix~B.1.

 Our checking algorithm leverages taint-tracking to find branch operations that depend on an
 input, then ensures that if a non-volatile location is written on any
 path of such a branch, that location is either written on all paths of
 the branch or is in the checkpointed set. 
 We call variables that are written on all paths regardless of inputs {\em
   must-write variables}.

 There are two top-level judgments for commands, depending on whether control is currently input-dependent --- i.e., \emph{tainted} --- or not: 
 $N;M\Vtaint c: \m{ok}$ and $N;I;M\Vdash_\rio c: \m{ok}$ respectively. Likewise there are two judgments for instructions: 
  $N;M' \Vtaint \iota:\m{ok}$ or $N;I;M \Vdash_\rio \iota:\m{ok}$. $N$ is the set of
 checkpointed variables, $I$ is the set of variables that depend on inputs
 (control and data dependence), and $M$ is the set of variables written prior to
 executing $c$. The $\Vtaint$ judgment does not use $I$ as the judgment itself carries the information that 
 control is tainted. The $\Vtaint$ check for instructions ensures that all writes access 
 checkpointed variables in $N$ or are must-write variables, which
 captures how RIOs' effects may transitively taint variables through
 dependences.  We show rule \rulename{RIO-Assign-tainted} as an example.

 \begin{mathpar}
 \small

 \inferrule[RIO-Cp]{\omega;\emptyset;\emptyset \Vdash_\rio \cmd: \m{ok}}{
   N;I;M \Vdash_\rio \m{checkpoint} (\omega);\cmd : \m{ok}
 }

 \and
  \inferrule[RIO-If-NDep]{I \cap rd(e) = \emptyset 
 \\
 N;I;M \Vdash_\rio \cmd_i: \m{ok}   \\ i \in [1, 2] }{
  N;I;M \Vdash_\rio \m{if}\ e\ \m{then}\ \cmd_1 \m{else}\ \cmd_2: \m{ok} 
 }
 \and
 \inferrule*[left=RIO-If-Dep]{I \cap rd(e) \neq \emptyset 
 \\
   M \Vdash^{\mt{mstWt}} \m{if}\ e\ \m{then}\ \cmd_1\ \m{else}\ \cmd_2:M' 
   \\\\   N;  M' \Vtaint \cmd_i: \m{ok} ~i \in [1, 2]}{
     N;I;M \Vdash_\rio \m{if}\ e\ \m{then}\ \cmd_1\ \m{else}\ \cmd_2: \m{ok} 
 }
\and
 \inferrule[RIO-Assign-tainted]{ x \in (M \cup N)}{
     N;M \Vtaint x:= e : \m{ok} 
   }
 \end{mathpar}

 We explain selected rules for commands. Rule \rulename{RIO-Cp}
 applies to a command starting with a
 checkpoint and checks the remaining command $\cmd$ using the checkpoint's
  checkpointed set $\omega$ and an empty $I$ and $M$.
 Rule \rulename{RIO-If-NDep} checks a branch that is input-independent.
 \rulename{RIO-If-Dep} checks an input-dependent branch 
 , identifying its must-write variables up to the next
 checkpoint using auxiliary judgment $M \Vdash^\mt{mstWt} c: M'$ (not shown). The resulting $M'$ includes variables written on the path up to 
 the branch (\{i, b\} in the example) and any variables that must be written on all paths from the branch ($\{x \}$ in the example). 
 The key difference between these two rules is that \rulename{RIO-If-NDep}
 checks its sub-commands with $\Vdash_\rio$ and the must-write set $M$, whereas 
 \rulename{RIO-If-Dep} checks its sub-commands with $\Vtaint$ and $M'$. The intuition for this 
 is that an input-independent branch will always evaluate the same way. Any variable written will be 
 written on any re-execution, even if it is not written on all paths. 
 An input-dependent branch may not take the same path, however, so the sub-command must be checked with the 
 must-write set of all paths from the branch. 

 Returning to Figure~\ref{fig:rio-example}, the rules check the if statement at
 line 7 using judgment $\Vtaint$. The $M'$ is $\{i, b, x\}$. To satisfy
 the check, $N$ must include $y$, $z$, $w$. 

\subsection{Defining the Effect of Input on Execution Prefixes}
Input can cause an intermittent execution's memory state to vary across re-executions, behaviour impossible 
on a continuous execution.
We introduce notation showing how input interacts with a program execution.
 $O|_{\m{in}}$ denotes the sequence of input values in
observation sequence $O$. Let "trace" refer to the sequence of execution states with observations annotated on top of each transition generated 
by an intermittent execution. 
We define $\runof{\sigma,\inputs, \cmd}$ to be a trace starting at
$\sigma$, ending in command $\cmd$ with the input sequence
$\inputs$ and without checkpoints. Formally: 
$$\runof{\sigma,\inputs,\cmd} = \{T \bnfalt T = \sigma{\MSeqStepsto{O}}
(\timestamp,  \nvmem, \vmem, \cmd)\land
O~\m{contains~no~checkpoints}\land\inputs = O|_{\m{in}}\}
$$

We write $\runof{\sigma,\inputs, \mt{CP}}$ to denote the trace ending in the
nearest checkpoint.  Note that given a set of inputs, $\runof{\sigma,\inputs,\cmd}$ is always a
singleton set with a uniquely determined trace: 
inputs already executed from the {\em previous} checkpoint to the current execution point are fixed, and any 
execution to the current point with those same inputs will yield the same trace. There are, however, multiple
possible traces from an arbitrary point $\cmd$ to the {\em next checkpoint} due to inputs yet to execute. 

We write $\mt{Wt}(T)$ to be the write set of a trace and $\mt{FstWt}(T)$ to be the 
set of variables written before they are read in a trace.
We define the set of locations that must be written on any input:
\[\begin{array}{ll}
\mt{MstWt}(N, V, \cmd) = \{\loc \bnfalt \forall\tau, N, V, \cmd, \forall \inputs, \forall T \in \runof{(\timestamp, N, V, \cmd),\inputs,\mt{CP}}, 
\loc \in \mt{Wt}(T)\}
\end{array}
\]
The must-write set of the example program is $\{b, x\}$. 
We define the set of locations that must be first written before being read, no matter the input value, below:
\[
\begin{array}{ll}
\mt{MFstWt}(N, V, \cmd) = 
\{\loc \bnfalt \forall\tau, N, V, \cmd, \forall \inputs, 
\forall T \in \runof{(\timestamp, N, V, \cmd),\inputs,\mt{CP}},
\loc \in \mt{FstWt}(T)
\}
\end{array}
\]
This set contains locations in non-volatile memories that are first written
(not read) on all possible runs through the checkpointed region, starting from
$\cmd$.  We call the set the \emph{must-first-write} set, or $\mt{MFstWt}$. In the example program, the 
$\mt{MFstWt} = \mt{MstWt}$, but the distinction is important as the must-first-write set will not 
include variables that have a (non-write-dominated) WAR dependence.

\subsection{Invariants Relating Memories}
\label{sec:rio-relations}
As shown in Figure~\ref{fig:correctness-rio-example}, the non-volatile memories
are not always the same. To prove the intermittent execution model correct, we
need to identify relations between the intermittent execution and the
continuously-powered execution configurations.  
Therefore, we define relations between memories in intermittent and continuous executions,
as illustrated in Figure~\ref{fig:correctness-war-proofs}.  We relate the
memory states at the same execution point with dashed lines and relate an
intermittent execution's memory state at any point to a continuous execution's
memory at its initial point with solid lines. These relations describe how
memory locations are allowed to differ, while still converging to the same
memory and observation by the next checkpoint, the key invariant for correctness.

\begin{figure}[tb]
  \centering
  \includegraphics[width=1.00\textwidth]{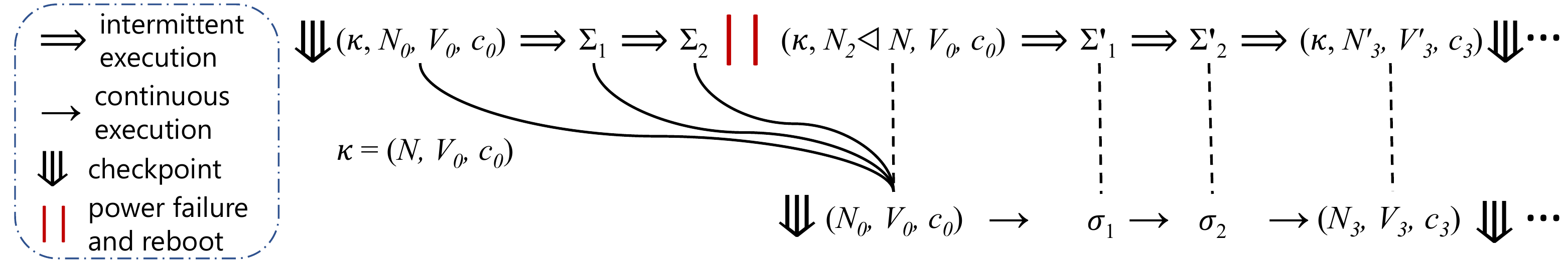}
 \vspace{-20pt}
  \caption{Illustrating the correctness proofs and invariants}
  \label{fig:correctness-war-proofs}
\end{figure}

\paragraphb{Arbitrary-point Memory Relation (Solid Line)}
We define the relation of the memory state of an intermittent execution $\nint$ at an arbitrary point  
to the memory state of a continuous execution $\ncont$ at its initial point:
\begin{defn}[Related memories between current and initial execution point]
\label{def:io-initial}  
  $N_{\mt{ckpt}}, N, V, c \vdash \nint \sim \ncont$ iff $\m{dom}(\nint)= \m{dom}(\ncont)$, 
$\forall \loc \in \nint$ s.t. $\nint(\loc) \neq \ncont(\loc)$, $\loc\in
N_{\mt{ckpt}} \cup \mt{MFstWt}(N, V, \cmd)$
\end{defn} 
For correctness, all locations that differ between $\nint$ and $\ncont$ must be in the
checkpointed set $N_{\mt{ckpt}}$ or in the must-first-write set of the initial
command $\cmd$. 
The checkpoint reverts writes to data in the checkpointed set on each reboot.
Re-execution after a reboot over-writes each variable in the must-first-write
set because {\em every} path to the next checkpoint writes to variables in
that set, regardless of input.
In Figure~\ref{fig:correctness-rio-example}, variables {$b,i,x,y$} differ
between $N_0$ and $N_4$.  
Given the initial memory state, execution always takes the false path from the
first branch and always writes $b$ and $i$.  All paths write $x$ and
it is always first written to.  Finally, $y$ is checkpointed.  Together these
actions reconcile all differences between $N_0$ and $N_4$.

\paragraphb{Same-point Memory Relation (Dashed Line)}
We next define the relation of the memory state of an intermittent
execution $\nint$ and a
continuous execution $\ncont$ at the same point : 
\begin{defn}[Related memories at the same execution point]
  \label{def:io-same_exec}\small
  $\timestamp, N, V, c, c', \inputs \vdash \nint \sim \ncont$ iff

$\m{dom}(\nint)= \m{dom}(\ncont)$ and

$\forall \loc \in \nint$ s.t. $\nint(\loc) \neq \ncont(\loc)$, 
\hspace{-5pt} \begin{itemize}
\item  $\loc\in\mt{MFstWt}(N, V, c)$
\item let $\{T\} = \runof{(\timestamp, N, V, c),\inputs, c'}$ and the last state of $T$ is
  $(\timestamp',  \nint, V, c')$,
 $\loc\in\mt{MstWt}(\nint, V, c')$  and $\loc\notin \mt{Wt}(T)$
\end{itemize}
\end{defn} 
The relation is parameterized with a timestamp $\timestamp$
and an input sequence $\inputs$. The relation uses $\timestamp$ and
$\inputs$ to define the singleton set $\runof{\sigma,\inputs, c'}$ from
the initial point to the current point.
The relation has access to all
parameters because at each point in the execution, all prior input
values are already concrete and timestamps are given.    
Beyond sharing the same domain, the
definition restricts locations that differ between the memories.  If such a
location is not written from the initial execution point to the current one,
then the location must be in the must-first-write set of the entire trace {\em
and } must be written between the current execution point and the next
checkpoint.
The intuition is that differing locations must be written to on all
possible paths through the remainder of the trace for the
intermittent and continuous traces to converge to the same state.
Moreover, re-execution following any path dictated by fresh input values
should not read locations that differ, which would cause
non-idempotent re-execution (i.e., first written to in every
execution). 

In Figure~\ref{fig:correctness-rio-example}, the continuous and intermittent
states differ at corresponding points $N_0$ and $N_0'$, $N_1$ and $N_1'$, and
$N_2$ and $N_2'$.  Starting with $N_0$, the execution writes $b,i,x$ regardless
of input.  
After stepping to states $N_1, N_1'$, $b$ cannot differ because the execution from its initial point to the current point wrote to $b$.  
After stepping to states $N_2, N_2'$, $i$ cannot differ because the execution
from its initial point to the current point wrote to $i$.  
$i$'s written value is the same in both executions because we choose a
continuous execution that reads the input at time $\mathit{9}$, which the
intermittent execution also reads. 
$x$ is not yet written, but 
will be on all paths.  
At states $N_3$ and $N_3'$, all locations must be the same. 
$x,b,i$ have been written between the initial and current execution point and
none are written between the current execution point and the next checkpoint. 
$z$ is written to on the current path, but is not in the must-first-write set
of the entire trace, nor are $w$ and $y$.

\subsection{Proving Correctness}
\label{sec:correctness-proofs}

We prove the following theorem:
\begin{thm}[Correctness]
  \label{thm:io-correctness-top}
  If $\Vdash_\war \cmd: \m{ok}$,
   $\Vdash_\rio \cmd: \m{ok}$ 
  then $\cmd$ can be correctly intermittently executed.
  \end{thm}

We actually prove a stronger
theorem that relates an intermittent execution up to checkpoints to 
a corresponding continuous execution. Only at each checkpoint, are the
memories guaranteed to sync up between the two executions. 

 The proof requires augmenting the semantics with variable taint tracking, dynamically marking all
 input-dependent locations. 
 We leverage standard taint tracking semantics rules and omit them here. 
 The proof follows the structure in Figure~\ref{fig:correctness-war-proofs} and 
  requires the following 
 properties: (1) arbitrary intermittent configurations relate to the continuous
 initial configuration, (2) each intermittent configuration relates to a
 continuous configuration at the same execution point, and (3), after reboot, we
 can switch from the relation illustrated by the solid line to that of the
dashed line and after checkpoint, we can switch from dashed line to
 solid line.

With properties (1) -- (3) established, the proof of Theorem~\ref{thm:io-correctness-top} 
is by induction over the
structure of the intermittent execution trace. First over the number of
checkpoints to show that, from checkpoint to checkpoint, the resulting memories
are the same and memory reads are idempotent. For each segment between
checkpoints, we induct over the number of reboots and use the relations in the
previous section to relate memory at each execution step. 
Note that no existing intermittent execution model checks $\Vdash_\rio \cmd:
\m{ok}$; none meets a reasonable correctness definition in the presence of I/O,
which is one of the key results of this work.  We show the full proofs in
Appendix~F.6.

\section{Collecting Exclusive May-Writes} 
\label{sec:rio-collection}
Given our correctness definition, an intermittent execution model must collect
and checkpoint not only WAR variables, but also RIO variables. 
Our algorithm identifies a (safe) conservative over-approximation of this set.
If a variable might be written on one side of a branch and not the other, and
the branch condition might change from one re-execution to the next due to a
RIO, then the variable should be checkpointed.
In other words, RIO variables are in the \emph{exclusive may-write}
set for some command $\cmd$ -- i.e., the set of variables that may be written
on exclusively one side of some future branch, but that will not be written
unconditionally.
Using the exclusive may-write set, a simple rewriting algorithm can transform a
program with empty checkpoint sets into a program that correctly checkpoints
RIO variables. 

Our algorithm computes exclusive may-write sets, and identifies input-dependent
(or \emph{tainted}) branches. 
Given a branch $\m{if}\ e\ \m{then}\ \cmd_1\ \m{else}\ \cmd_2$, if $e$ is not
(transitively) input-dependent, the branch's outcome is the same on every
re-execution.  The first branch in Figure~\ref{fig:rio-example} is
never taken under $N_{0f}$,
and the write to $b$ happens regardless of line 6's input; the write is not 
in the exclusive may-write set.  
If $e$ is input-dependent, the branch outcome depends on input and later writes
are candidates for exclusive may-write. In Figure~\ref{fig:rio-example}, the branch at line 7
is input-dependent and its exclusive may-write set is $\{w,y,z\}$, each of
which are written on one, but not both sides of the branch.
Our algorithm uses
taint analysis to identify input-dependent branches, and adds to $\omega$
exclusive may-write variables for input-dependent branches.  

\paragraphb{Collection per instruction.}
Two sets of rules collect the exclusive may-write set $X$, must-write set
$M$, and input-dependent variable set $I$ for instructions: $X;M;I
\Vdash_\rio \iota: X';M';I'$ and $X;M \Vtaint \iota: X';M'$. The $\Vdash_\rio$ rules apply 
to $\iota$ in commands from an input-independent branch and $\Vtaint$ rules apply to
$\iota$ in commands from an input-dependent branch. The primary distinction for 
instruction level rules is that $I$ does not need to be collected in the $\Vtaint$ rules 
as our branches never merge.
We explain selected $\Vdash_\rio$ rules; rules for
$X;M \Vtaint \iota: X';M'$ are similar with taint tracking removed. 

\begin{mathpar}
\small
\inferrule*[right=I/O-Get]{ }{
    X;M;I \Vdash_\rio x:= \m{IN}() : X;M\cup x;I \cup x
  }
  \and
  \inferrule*[right=I/O-Assign-dep]{ I \cap rd(e) \neq \emptyset}{
  X;M;I \Vdash_\rio x:= e : X;M \cup x;I \cup x
}
 \and
\inferrule[I/O-dep-clear]{ I \cap rd(e) = \emptyset \\ x \in I}{
   X;M;I \Vdash_\rio x:= e : X;M \cup x;I \setminus x
 }
\and

\inferrule[I/O-Arr-loc]{ I \cap rd(e) \neq \emptyset}{
  X;M;I \Vdash_\rio a[e]:= e' : X \cup a;M;I \cup a 
}
\end{mathpar}
All assignments add the variable $x$ to $M$, since $x$ must be written
on the current command. Rule \rulename{I/O-Get} adds $x$ to $I$.
If any assignment has an expression that
reads a value in $I$, the assigned location is also added to $I$, 
as taint propagates to $x$ (rule \rulename{I/O-Assign-dep}). 
Conversely, assigning a location in $I$ to an 
input-independent expression removes that location
 from $I$, effectively clearing its taint (rule
 \rulename{I/O-dep-clear}). 
Propagating taint to an array element would cause the entire
array $a$ to be conservatively tainted. 
If an array index is tainted, then $a$ is added to $X$ because 
the written array element may differ in each re-execution (rule
\rulename{I/O-Arr-loc}).  

\paragraphb{Collection for commands}
Two sets of rules define collection and rewriting for commands:
$X;M;I \Vdash_\rio \cmd \SeqStepsto{} \cmd': X'$ and
$X;M \Vtaint \cmd \SeqStepsto{} \cmd': X';M'$, with the same
distinction as the $\iota$ rules between $\Vdash_\rio $ and $\Vtaint$. 
These rules compute the
exclusive may-write set $X'$ and must-write set $M'$ up to a checkpoint in $\cmd$ 
and rewrite the checkpoint command to use the collected $X'$ as $\omega$.
Rewriting an instruction $\iota$ directly uses the rules we introduced
in the previous paragraph to collect relevant variants (e.g, $X$, $M$) and rewrites of itself. 
Much of the complexity for commands is collecting exclusive
may-write and must-write sets from (nested)
branches.  We show key rules:
\begin{mathpar}
\small
\mprset{flushleft}

\inferrule*[right=I/O-If-Dep]{I \cap rd(e) \neq \emptyset 
\quad  X;M\Vtaint \cmd_i \SeqStepsto{} \cmd_i': X_i;M_i   
~~ i \in [1, 2]}{
X;M;I \Vdash_\rio \m{if}\ e\ \m{then}\ \cmd_1\ \m{else}\ \cmd_2 \SeqStepsto{ } 
\m{if}\ e\ \m{then}\ \cmd_1'\ \m{else}\ \cmd_2': (X_1 \cup X_2 \cup M_1 \cup M_2) \setminus (M_1 \cap M_2)
}
\and
\inferrule*[right=I/O-If-Tainted]{
\emptyset;M \Vtaint \cmd_i \SeqStepsto{} \cmd_i': X_i;M_i \\ i \in [1, 2]}{
X;M\Vtaint \m{if}\ e\ \m{then}\ \cmd_1\ \m{else}\ \cmd_2 \SeqStepsto{ } 
\m{if}\ e\ \m{then}\ \cmd_1'\ \m{else}\ \cmd_2'
\\ \qquad\qquad\quad : (X_1 \cup X_2 \cup M_1 \cup M_2) \setminus (M_1 \cap M_2) ;  (M_1 \cap M_2)
}
 \and
 \inferrule*[right=CP-Tainted]{\emptyset;\emptyset;\emptyset \Vdash_\rio 
 \cmd \SeqStepsto{} \cmd': X'}{
   X;M \Vtaint \m{checkpoint()};\cmd \SeqStepsto{}
   \m{checkpoint(X')};\cmd' : X; M
 }
\end{mathpar}

\rulename{I/O-If-Dep} applies to an input-dependent
branch encountered when control is not yet tainted. 
The rule switches from $\Vdash_\rio$  to $\Vtaint$ in the premises 
because the condition expression is input dependent. 
The exclusive may-write set is the union of exclusive may-write and must-write sets from each
side of the branch, minus the intersection of the must-write sets. 
\rulename{I/O-If-Tainted} applies to branches encountered while in $\Vtaint$
This rule must also collect the branch's must-write set,
which is the intersection of the must-write sets from both sides of the branch,
unioned with the must-write set from before the if statement. 

\begin{wrapfigure}{R}{0.50\textwidth}
  \centering

  \includegraphics[width=0.48\columnwidth]{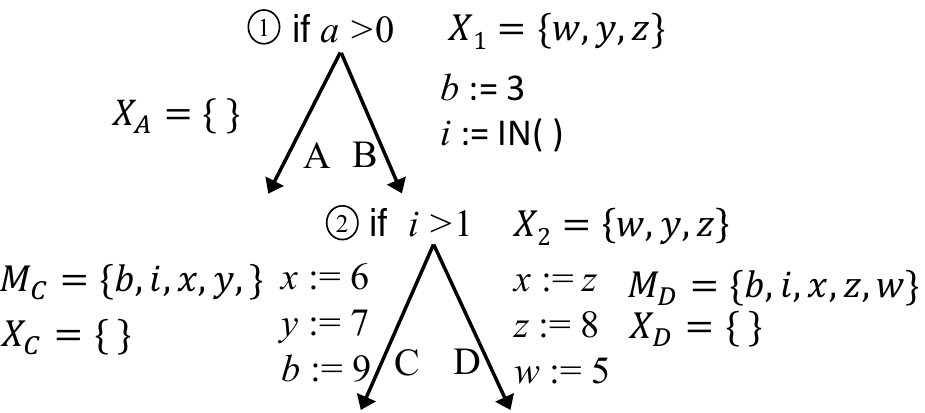}

\vspace{-10pt}
\caption{Example of exclusive may-write set collection.}
  \label{fig:exclusive-may}
\end{wrapfigure}

\paragraphb{Exclusive-May-Write Collection Example}
Figure~\ref{fig:exclusive-may} illustrates exclusive may-write collection,
using code from Figure~\ref{fig:rio-example}.  We letter each branch outcome
path and number each branch instruction.  Collection begins in the
$\Vdash_\rio$ rule, as the branch is not tainted.  For path $A$, $X_A$ is
$\emptyset$.  For path $B$, the must-write set is $\{b,i\}$ before branching on
$i$.  We apply $\Vtaint$ rules to paths $C$ and $D$ because $i$ is
input-dependent.  
$X_C$ and $X_D$ are empty, $M_C = \{b,i,x,y\}$, and $M_D = \{b,i, x,z,w\}$.
Then the must-write set for branch 2 is the intersection of $M_C$ and $M_D$,
which is $\{b, i, x\}$. The may-write set of paths $C$ and $D$ is the union of
the segments' $X$ and $M$ sets: $\{b,i,w,x,y,z\}$.  The exclusive may-write set
removes variables that must be written on both paths (i.e., $\{b,i, x\}$); so
$X_2=\{w,y,z\}$.  Collection propagates $X_B=X_2$ and the final exclusive
may-write set for branch 1 is the union of those for paths $A$ and $B$: $X_1=
\{w,y,z\}$. For this code snippet, $\{w,y,z\}$ have to be checkpointed.

\paragraphb{Inserting $\omega$}
At a checkpoint, collection inserts the exclusive
may-write set returned from rewriting $\cmd$ as $\omega$ (rule \rulename{CP-Tainted}).
 The rules re-write the command after a checkpoint using the judgment
 $\Vdash_\rio$, {\em even if control is tainted}.
 This re-writing is correct because after checkpointing, the input operation
 will never be re-executed and execution is
 deterministic until the next input operation. 

\paragraphb{Correctness of the Algorithm}
We prove that the collection rules are safe with regard to the
checking rules, which is a condition in Theorem~\ref{thm:io-correctness-top}.
\begin{lem}
If $~\Vdash_\rio \cmd \longrightarrow \cmd': X$ then $\Vdash_\rio
\cmd': \m{ok}$. 
\end{lem}
The lemma states that if a command has been rewritten using the collection algorithm, 
the rewritten command is safe with respect to RIOs.  
Note that separating the collection algorithm from the checking $\Vdash_\rio
\cmd': \m{ok}$ enables modular proofs. A different collection
algorithm does not change the overall correctness proof as long as it
can be shown to be safe w.r.t. the checking rules. 

 \section{Equivalences Between Systems}
\label{sec:variants}
The checkpoint system presented is based on DINO~\cite{dino}, but many others
exist.  We additionally formalize four alternative implementations---undo
logging~\cite{chinchilla}, redo logging, idempotent regions~\cite{ratchet}, and
a task-based execution model Alpaca~\cite{alpaca}. Instead of reproving the
correctness theorems for each system, we define and prove a bi-simulation
relation between the basic system and each alternative showing that they are
equivalent. While these systems differ in mechanism and performance, the
equivalence result shows that their correctness criteria are the same, allowing us to implement the algorithm in Section~\ref{sec:rio-collection}
for the more performant Alpaca. We
relegate bi-simulation for undo logging and idempotent regions to Appendix~D, as DINO 
implements a conservative form of undo-logging, and idempotent regions differ from the basic 
model by constraints on checkpoint placement. 

\subsection{Redo Logging}
\label{sec:redo-log}
State restoration 
can be implemented with
redo-logging, which
works by logging memory updates during execution and committing
the log to memory upon reaching a checkpoint. A redo logging context $\rlctx$ is of the form $(\ulog, \vmem, \cmd, \omega)$, where $\ulog$ is a log and $\omega$ has 
the same meaning as before.

\begin{wrapfigure}{L}{0.50\textwidth}
  \centering
  \includegraphics[width=0.48\textwidth]{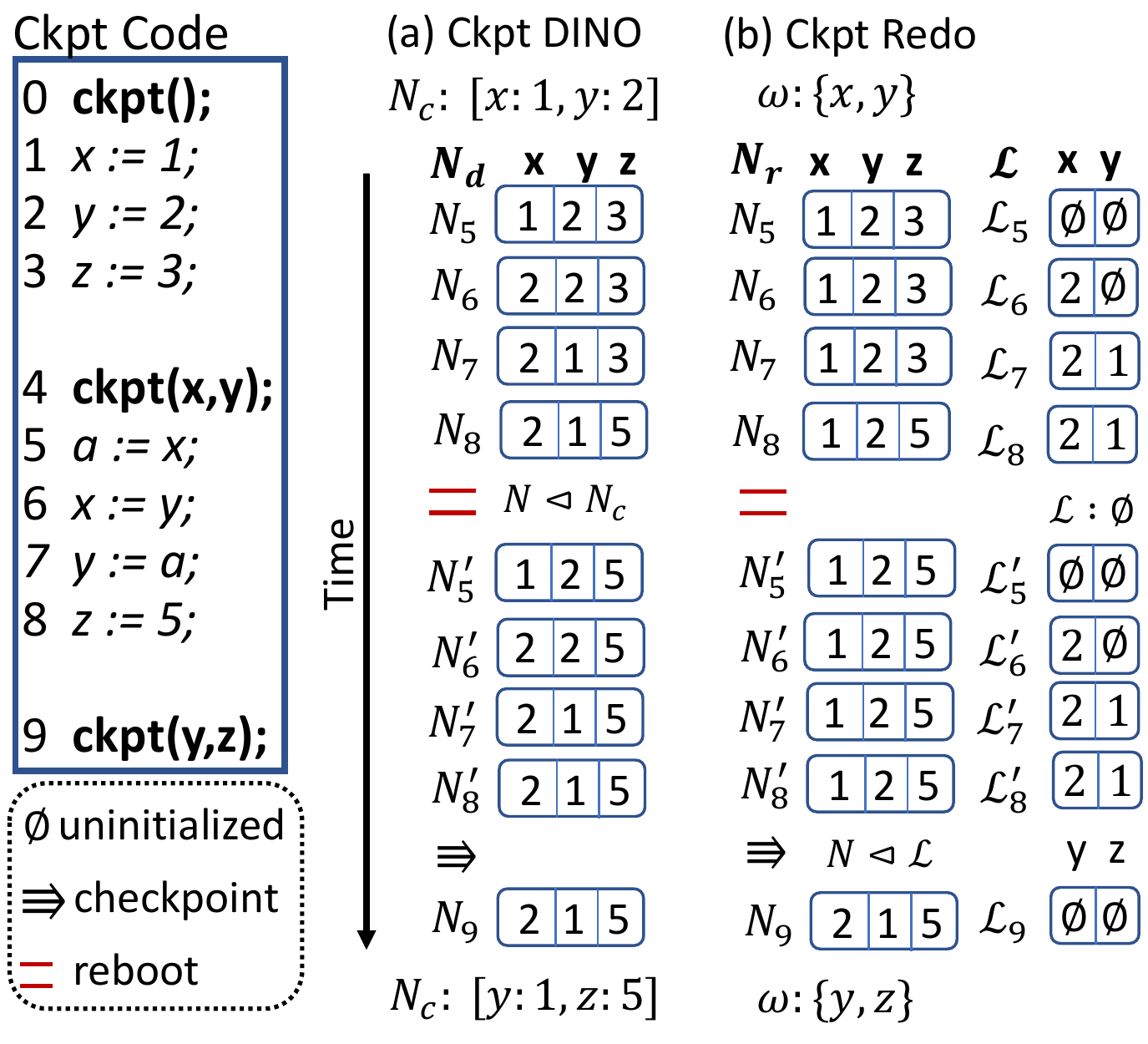}
  \vspace{-5pt}
  \caption{Illustrating the relation of DINO to redo-logging}
  \label{fig:equiv-redo}
\end{wrapfigure}

We illustrate 
the key behaviour of redo-logging and its relation to 
the basic model in Figure~\ref{fig:equiv-redo}. {\tt Ckpt Code} is a simple program 
that initializes the variables $x, y, z$, checkpoints $\{x, y\}$, swaps their values using a volatile variable $a$, and takes another checkpoint 
before continuing with the rest of the program. Columns (a-c) on the right show the program's state at each point in the execution.  Column (a) shows the execution of the basic model starting from the checkpoint at line 4. 
Column (b) shows a redo-logging execution. Volatile memory $\vmem$ and command $\cmd$ are equivalent at each step and omitted.
DINO starts with $\nvmem_c$ containing the values of $x,y$ at the checkpoint. Redo-logging starts 
with an empty log and $\omega = \{x, y\}$. In any checkpoint region, the 
domain of the log will be $\omega$.   
If the program contains an assignment to a non-volatile location in $\omega$ (lines 6,7), the update is placed
directly into the log, leaving non-volatile memory untouched.
Otherwise the variable is updated directly in non-volatile memory (line 8). On reboot, the redo log clears but leaves
non-volatile memory untouched, as all updates to locations in $\omega$
reside in the log only. The basic model updates non-volatile locations in $\omega$ 
to the values from the checkpoint. As the program re-executes, updates to locations in $\omega$ are redone, either directly to 
non-volatile memory (a) or to the log (b). When the program reaches the next checkpoint on line 9, the Redo model applies the log 
to non-volatile memory, committing the changes. 

A key part of the bi-simulation relation between DINO and redo-logging is that if the
value of a location in redo-log non-volatile memory $\nvmem_r$ is not
equal to the same location in DINO's non-volatile memory $\nvmem_d$, then
that location is in the domain of $\omega$ and therefore in the log. The value
in the log is equal to the value in $\nvmem_d$, as both reflect its latest
update. 
Consequently, $\nvmem_d = \nvmem_r \lhd \ulog$. 
At reboots and 
checkpoints $\ulog$ is empty, and
 the non-volatile memory, volatile memory, and command of both models are 
equal. We formalize this relation and prove bi-simulation in Appendix D.4.

\subsection{Task-Based Systems}
\label{sec:tasks}

Task-based systems~\cite{alpaca,capybara,coati,chain, mayfly} require
the programmer to structure an intermittent program as a series of
transaction-like tasks. Updates within a task are not be visible to
other tasks (including re-executions of the same task) until the task
commits. Task semantics rely on either 
undo or redo logging. 
The key
differences between checkpoint- and task-based systems are in their
memory abstraction and control structure.

\[
  \small
\begin{array}{llcl}

\textit{Task Map} & T & \bnfdef & \cdot\bnfalt T, i\mapsto (\omega, \cmd) 
\\
\textit{Instr.} & \iota & \bnfdef & \cdots\bnfalt \m{toTask}(i)

\end{array}
\]

A task-based
program has no checkpoints, and is instead a series of tasks. The context is augmented with a
task map $T$, from task IDs to a checkpoint set $\omega$ and command
for the corresponding task. The context $\tskctx$ for intermittent
execution is a pair $(T, i)$ consisting of the task map and the ID of the
current task. A
special instruction $\m{toTask}(j)$ ends the current
task, transitioning to task $j$ (rule \rulename{TSK-Trans}).

\begin{mathpar}
  \small
  \centering

 \inferrule*[right=TSK-Trans]{ 
  \tskctx = (T,i)
  \\ T(j) = (\omega, \cmd)
  }{
    (\tskctx, \tshared, \tpriv, \tlocal, \m{toTask}(j)) 
   \tskStepsto{\m{transition}}  ((T,j), \tshared \lhd \tpriv, \emptyset, \tlocal, \cmd) 
 } 
\end{mathpar}

The task-system has a memory abstraction of task-shared memory 
$\tshared$, task-local memory $\tlocal$, and task-private
memory $\tpriv$. A programmer assigns variables accessed in
multiple tasks to task-shared memory and variables used only in a
single task to task-local. Task shared memory must be non-volatile,
and task-local can be split into volatile and non-volatile sections,
$\tlocalv$ and $\tlocaln$ respectively.  Task-private variables are hidden from the
programmer and used to implement logging.  A task must initialize task-local variables
by writing them before reading them.  Like basic checkpoints, the
command in each task is well-formed given the checkpoint set:  $\omega\Vdash_\war \cmd: \m{ok}$  and  $\omega\Vdash_\rio \cmd: \m{ok}$.

\[
  \small
\begin{array}{llcl}
\textit{Commands} & \cmd & \bnfdef & \cdots \bnfalt \m{goto}~\ell
\\
\textit{Code context} & \Psi & \bnfdef & \Psi \cdot \bnfalt \ell_i:\m{checkpoint(\omega);\cmd}
\end{array}
\]
\begin{wrapfigure}{R}{0.55\textwidth}
  \centering
  \includegraphics[width=0.53\textwidth]{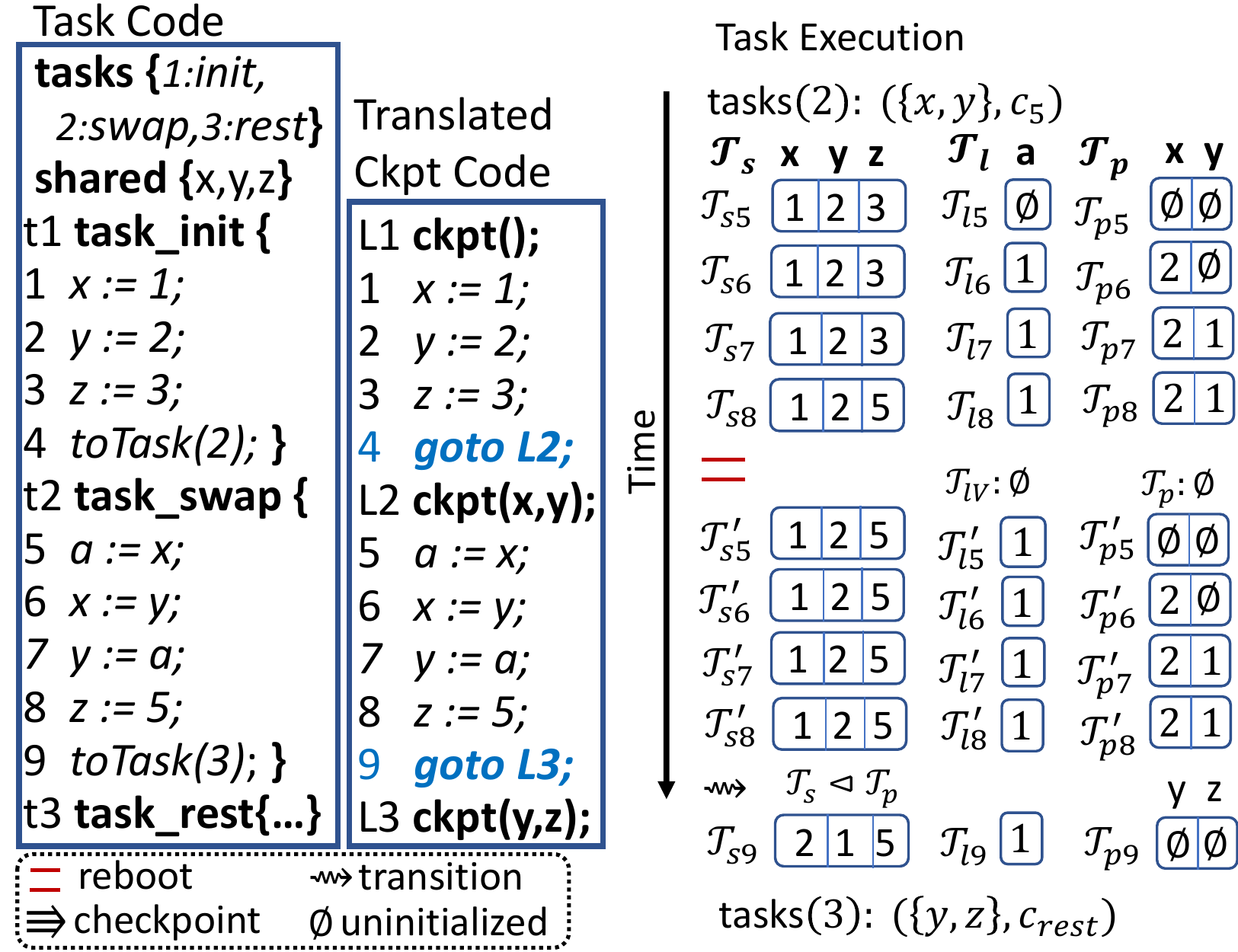}
 \vspace{-5pt}
  \caption{Relating checkpoint redo-logging to tasks}
  \label{fig:equiv-task}
\end{wrapfigure}
    To prove equivalence between a task-based and a checkpoint-based
    system, we first translate task transitions
    to checkpoint commands. We augment the
    basic language with a $\m{goto}$ command and 

a code context $\Psi$ that includes a set of labeled
    program points, each beginning with a checkpoint. 
Correctness still holds; extending the main proofs is trivial as $\m{goto}$
does not change memory.

We specify a translation relation from tasks to a redo-logging code
context, written $T \trel \Psi$. 
The key idea is that each task can
be translated to a checkpoint followed by the translated task command:
$i \mapsto (\omega, \cmd_t) \trel
\ell_i:\m{checkpoint}(\omega);\cmd_r$. Here
 $\cmd_r = \llbracket \cmd_t\rrbracket $ and task transitions are
 translated to $\m{goto}$s:
$\llbracket \m{toTask}(i) \rrbracket \trel \m{goto}~\ell_i$. The
translation of the rest of the constructs recursively translates
the sub-terms and returns the same construct when an instruction
is reached. 
We show a task-based version of the program to swap $x$ and $y$ and its translation 
to a checkpoint program in Figure~\ref{fig:equiv-task}.
The program is a series of three tasks: {\tt init}, {\tt swap}, and {\tt rest}. The 
variables $x,y,z$ are shared between the tasks. Variable $a$ is local to task {\tt swap}. In the translation, each {\tt toTask} is 
replaced by a {\tt goto} whose label points to a checkpoint followed by the translated task command.

Using these constructs and the translated program, we relate a configuration of a task-based
system to that of redo-logging. We show the formal relation below and an execution of the task program on the right side of Figure~\ref{fig:equiv-task}.
The redo-log execution of the translated program is the same as in column (b) of Figure~\ref{fig:equiv-redo}
 as the only difference to the original {\tt Ckpt} code is the addition of the blue {\tt goto} instructions.
 \[
  \small  
  \inferrule*{
  \rstate = (\rlctx, \nvmem_r, \vmem_r, \cmd_r)
  \\ \tstate = (\tskctx, \tshared, \tpriv, \tlocal, \cmd_t) 
  \\ \rcon = (\ulog, \vmem_c, \cmd_c, \omega_r)
  \\ \tcon = (T, i) \\ T(i) = (\omega_t, \cmd_\mt{tt})
  \\ T \trel \Psi
  \\ \llbracket \cmd_t \rrbracket = \cmd_r
  \\ \llbracket \cmd_\mt{tt} \rrbracket = \cmd_c
  \\ \omega_t = \omega_r
  \\ \tpriv = \ulog
  \\ \tlocal = \tlocalv, \tlocaln
  \\ \nvmem_r = \tshared \cup \tlocaln
  \\ \vmem_r \approx \tlocalv
   \\\m{dom}(\vmem_r) \subseteq \dom(\tlocalv) 
  }{
  \tstate \rel \Psi, \rstate
  }
  \]

The redo-log $\ulog$ and $\tpriv$ are equivalent. Any updates 
to locations in $\omega$ will be placed into $\tpriv$, not $\tshared$. If $\tlocal$ is entirely volatile, 
then $\tshared$ and $\nvmem_r$ will also be equivalent. Otherwise, if $a$ is stored in a non-volatile location, then $\nvmem_r$
will be equal to the union of $\tshared$ and $\tlocaln$. If $a$ in a volatile location, $\tlocalv$ will be equal to the 
redo log volatile memory, once $a$ has been initialized ($\vmem_r \approx \tlocalv$). This qualification is necessary as $\tlocalv$
is cleared on reboot (after line 8), whereas redo logging restores the checkpointed volatile memory. 
Translated commands will always be well-formed w.r.t. task-local memory, however, so there can be 
no \emph{read} to a volatile memory location before it is initialized. Thus any memory accesses 
on the two systems will be equivalent. This property does not hold for any arbitrary redo-log program; consider a redo-log program where 
the assignment to $a$ occurred before the checkpoint --- the first access to $a$ after the checkpoint would
be a read. 
At line 9, task {\tt swap} transitions to task {\tt rest}.
A task transition commits
$\tpriv$ to $\tshared$, resets $\tpriv$, and transfers control to the specified task, switching the
task reference in the context to the new task. The translated program jumps 
to the label $L3$, corresponding to the command $\m{checkpoint(y,z); 
\llbracket\cmd_{rest}\rrbracket}$.
When it executes this checkpoint, it updates $\nvmem_r$ with the log 
and the clears the log. As $\tpriv$ and $\ulog$ are equivalent, the updated non-volatile memory is still 
equivalent to the union of the task-shared and non-volatile task-local memories. 
Furthermore, as no writes are left to occur, $\vmem_r = \tlocalv$. 
We prove equivalence in  Appendix~E. 

 \section{Implementation}\label{sec:implementation}
We implemented the exclusive
may-write (EMW) collection algorithm in Section~\ref{sec:rio-collection}, which consists of both
write-set collection and taint-tracking, and combined
its output with Alpaca's 
runtime system, yielding an
intermittent execution runtime with safe access to I/O. 
We built two variants:
\emph{EMW}, which backs up 
EMW sets for all branches (correct, but conservative), and
\emph{taint-optimized EMW}, which calculates the EMW set for only input-dependent
branches. 
EMW requires no code changes, but 
backs up some unnecessary variables.
Taint-optimized EMW requires very minor code changes to annotate input
operations, but backs up a smaller, far less conservative variable set.  

\subsection{System}
We implemented EMW collection in LLVM~\cite{llvm} and to Alpaca, we added
support to back up EMW variables. Alpaca is a task-based system and can use
either undo or redo logging~\cite{alpacaarxiv}. As Section~\ref{sec:variants}
shows, a task-based program translates into an equivalent checkpoint-based
program with either style of logging. 
We use undo-log Alpaca since it is the most efficient Alpaca
variant. 

\subsection{Limitations Due to C Features}
The algorithm in Section~\ref{sec:rio-collection} is sound for our
simple modeling
language. However, Alpaca extends C, which has several features not present in
the modeling language, such as merging branches, arbitrary pointers, and non-recursive functions (which do not checkpoint state), 
leading to a few differences between the formal statement
of the algorithm and the implementation. Merging branches and functions do 
not require changes to the algorithm. All paths can be explored even if a branch merges
A write in an unconditionally executed block 
will execute on all paths and be in the must-write set. Functions are treated as inlined, leveraging the lack
of recursion. 

\paragraphb{Taint tracking}
Pointers and functions complicate taint tracking. Taint-optimized EMW collection is sound only if taint does not propagate
indirectly, as through pointer arithmetic (e.g., $y$ points to $N$, the address of a tainted location, $x$ points to $N-1$, $x++$. The algorithm would miss that x points 
to a tainted location). Our implementation propagates taint through function parameters and return
values.  A call with a tainted parameter taints the corresponding argument.  A
tainted return value taints the store of the return value in the
function's caller. A function may taint a reference parameter, and our
implemented algorithm taints the corresponding parameter in the function's caller (similar
to ~\cite{ibis}).
These aliasing limitations of taint tracking do not affect the soundness of taint-agnostic EMW
collection, and none of our test programs had indirect taint propagation, which would compromise soundness.

\paragraphb{Write set collection}
To compute write sets, we assume that task-shared variables must be stored 
to directly, and cannot be aliased through a task local pointer. All Alpaca 
applications followed this behaviour. This direct access of task-shared variables allows the algorithm to 
compute may and must write sets precisely, apart from arrays.
This limitation is due to our prototype implementation and is not inherent 
to the formal algorithm. To extend the prototype to compute safe EMW sets with complex aliasing, must-write sets should 
include must-alias only, and may-write sets should include may-alias locations. 
As in the formal algorithm, any array written to on a tainted branch is conservatively (safely)
put into the EMW set.

\subsection{Algorithm Implementation}
\label{sec:ibis-diff}
We implement taint tracking as a fixed-point dataflow analysis, propagating
data taint in a traversal.  At the end of each traversal, the algorithm
examines any instructions that introduced inter-procedural dataflow, and adds
any new sinks to a worklist from which to start future traversals.  When no new
inter-procedural flows are identified, the algorithm is at a fixed point and
stops.  The analysis returns a list of tainted instructions. As we
mentioned earlier, we could under-taint, though we did not observe any under-tainting.

EMW collection is a separate analysis directly implemented from our formal
description. 
Taint-enabled EMW collection narrows the scope of the analysis, using the taint
tracking analysis result and calculating the EMW sets for conditionals that
are tainted only.  EMW collection returns a map from a function to its 
computed EMW sets.

We modify Alpaca's undo-logging compiler analysis to use the EMW set
information.  Alpaca maintains a per-function set of WAR variables to undo-log
in their called task.  We modify Alpaca to include a function's EMW set with
the function's WAR variables, which are then passed together to Alpaca's
existing undo-logging instrumentation pass, which allocates undo log storage,
creates checkpoint metadata, and adds undo-logging instrumentation.  

\paragraphb{Comparison to IBIS' algorithm}
The specification of the algorithm in IBIS is unsound. Even assuming perfect pointer aliasing and 
taint propagation, IBIS could still miss bugs. IBIS detects RIO bugs 
by calculating the may-write sets of paths off tainted branches and comparing them. If 
the may-write sets are equal it reports no bug. Consider 
$\m{if}(tainted~e)~x:=1; \m{if}(e2)~y:=1~\m{else}~z:=1;
$ $\m{else}~x:=1; y:=1;z:=1$. IBIS would report 
no bugs as the may-write sets are the same, but $y$ and $z$ are in the EMW set and thus
potentially inconsistent. Additionally, conservatism, whether due to implementation decisions such as aliasing 
or inherent in static analyses (such as 
opaque path conditions) hampers the usability of IBIS. Any variable falsely identified as potentially inconsistent 
generates a confusing false-positive bug report that the programmer must 
reason through, whereas the EMW runtime safely adds it to the checkpoint at little runtime cost.

 \section{Evaluation} \label{sec:eval}

The goal of the evaluation is to show that modifying Alpaca to
correctly support input operations is practically efficient.  We
evaluate Alpaca's baseline system, a variant that checkpoints all data
identified by our exclusive may-write (EMW) analysis without taint
tracking, and a variant that checkpoints all data identified by our
EMW analysis refined with taint-tracking support.  Our data show that
our analysis provides correctness, through checkpointing both WAR and RIO variables, with low run time and memory
overheads.  EMW alone has very low overheads with {\em no programming
  effort} \limin{no programmer annotation effort?} and EMW plus taint
analysis has {\em virtually no time overhead} and very low memory
overhead, but asks the programmer to annotate input operations.  To
demonstrate the programmability benefit of our analysis, we perform
case studies, showing that it is non-trivial (sometimes complicated)
to fix RIO bugs manually, even using a state-of-the-art bug detection
tool (IBIS~\cite{ibis}), but trivial using our analysis.

\subsection{Benchmarks}
We use benchmarks from the IBIS paper~\cite{ibis}, obtained from the authors, as they run on Alpaca and have input bugs.
There are 11 programs: 7 drivers and
low-level applications from TI-RTOS~\cite{tirtos} and 4 from Alpaca~\cite{alpaca}. The TI-RTOS
programs are {\tt bmp} a pressure sensor driver, {\tt hdc} a humidity sensor
driver, {\tt elink}, a radio implementation, {\tt mpu}, a magnetometer
driver, {\tt opt}, an optical sensor driver, {\tt temp}, a temperature sensor
driver, and {\tt wsn}, a sensor data aggregator.  The Alpaca programs are {\tt
ar}, activity recognition, {\tt bc}, bit counting, {\tt cem}, a compressive logger, and {\tt cuckoo}, a cuckoo filter. IBIS
found RIO bugs in {\tt mpu}, {\tt opt}, {\tt temp}, and {\tt wsn}.

\begin{wrapfigure}{R}{0.52\textwidth}
  \centering
  \includegraphics[width=.52\textwidth]{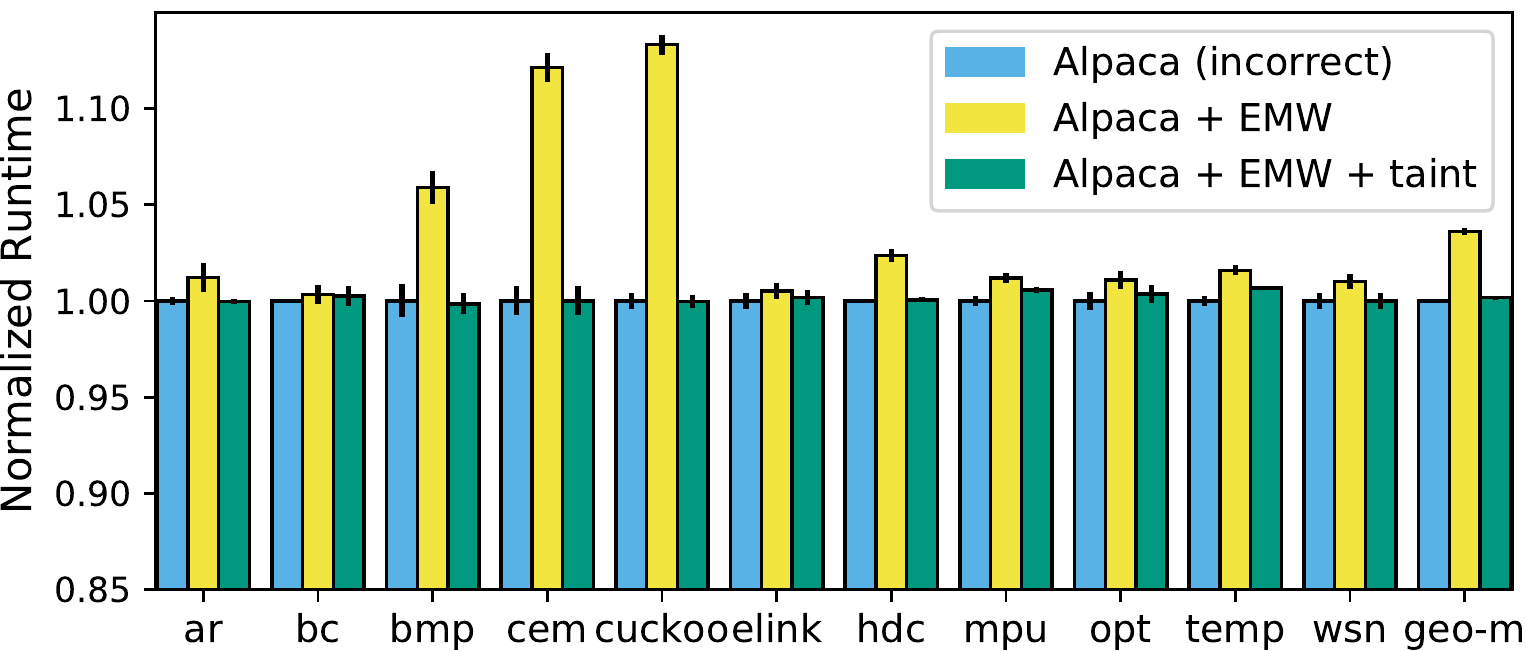}
  \vspace{-20pt}
  \caption{Normalized runtimes of Alpaca, Alpaca with EMW sets, and Alpaca with taint-optimized EMW sets}
  \label{fig:runtimes}
  \end{wrapfigure}

\subsection{Performance Overhead of EMW Tracking} Checkpointing data added to
$\omega$ by EMW analysis guarantees correctness and causes only low run time and memory
overheads. 
 Figure~\ref{fig:runtimes} shows run time normalized to Alpaca for
plain Alpaca (blue), Alpaca with checkpointing for EMW sets (yellow), and Alpaca
with checkpointing for I/O-tainted EMW data only (green).  Each bar averages
100 run times on continuous power with fixed inputs, and error bars are a 95\%
confidence interval.  Plain Alpaca is fastest, {\em but incorrect} because it
does not back up variables made inconsistent by RIOs.  EMW with taint-tracking has virtually no time overhead
(0 - 0.7\%) because the analysis checkpoints the few variables from the EMW
set required for correctness.  Checkpointing full EMW sets has higher overhead,
ranging from negligible to nearly 15\% for {\tt cuckoo}. The higher overheads
of EMW alone demonstrate the need for 
taint tracking, which
eliminates overheads, requiring only that the programmer annotate input
operations.

\begin{wrapfigure}{R}{0.52\textwidth}
  \centering
  \includegraphics[width=.52\textwidth]{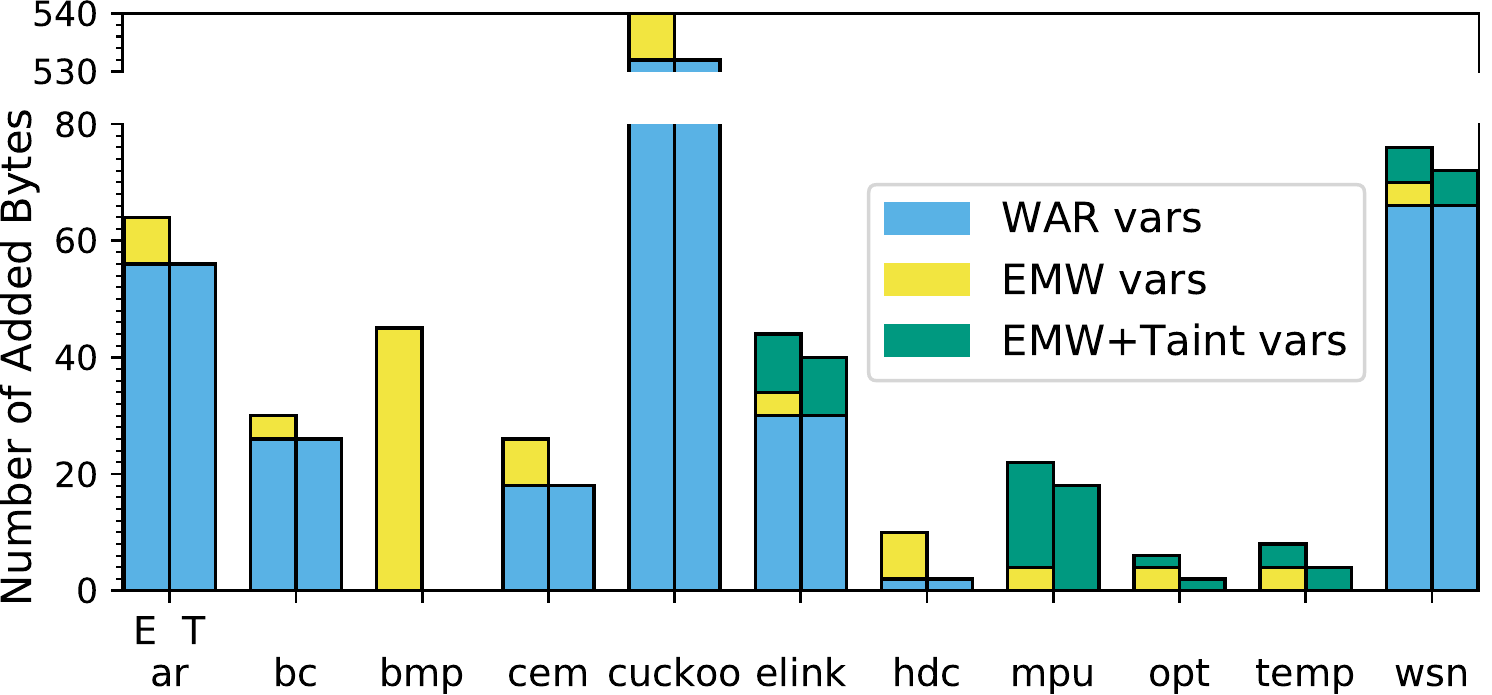}
  \vspace{-20pt}
  \caption{Space needed to back up variables identified by EMW and taint-optimized EMW analysis, 
  by category}
  \label{fig:mem}
  \end{wrapfigure}

\paragraphb{Checkpoint and Memory Overheads} 
Taint-optimized EMW analysis checkpoints only the few variables necessary to
avoid RIOs and WARs, while using EMW analysis conservatively requires
checkpointing many more variables, at a higher memory overhead. 
Figure~\ref{fig:mem} shows the bytes needed to back up variables identified by EMW
and taint-optimized EMW. The left bar in each pair is for EMW alone,
and the right is for taint-optimized EMW. Each bar is broken up into bytes due to variables in
a WAR dependence (blue), untainted variables conservatively in the
exclusive-may-write set (yellow), and tainted variables in the exclusive may-write
set (green). Any array is double-buffered, potentially causing a large difference in log size ({\tt bmp})
Together, the yellow and green segments in an EMW bar include all tainted and
untainted EMW variables. The taint-optimized EMW bar eliminates the yellow
segment, identifying input-tainted EMW variables only and illustrating the
conservatism in EMW alone that requires checkpointing more variables.

Figure~\ref{fig:memovhd} quantifies the normalized memory overhead caused by the increase in 
logged variables, accounting
for all checkpoint storage and metadata (including
array-size dependent metadata~\cite{alpaca}).  Taint-optimized EMW reduces
memory overheads significantly compared to EMW alone.

\paragraphb{Programmability Benefits of EMW}
Using taint-optimized EMW analysis is a simpler solution for repeated I/O
than manually changing code.
Prior work detects RIO bugs using an ad hoc approximation of our
taint-optimized EMW analysis~\cite{ibis}, suggesting that the programmer fix
bugs.
Taint-optimized EMW has low overheads and requires the programmer
to annotate input operations only, which is simple.
Manually finding and fixing bugs is relatively more complex.  
IBIS~\cite{ibis} advises re-initializing 
I/O-tainted variables at the start of the task that taints them. 
This strategy moves each variable into the {\em must-first-write} set, causing
it to be written on every task execution and eliminating the
need to include it in $\omega$.   However, unconditional initialization in a task may
change a program's meaning if the value overwritten by the initialization is important.
Instead, a programmer could create their own backup copy of the variable and save
its value on the first write in the task; doing so amounts to manually applying
undo-logging, guided by IBIS's bug report.

  \begin{wrapfigure}{R}{0.52\textwidth}
    \includegraphics[width=.52\textwidth]{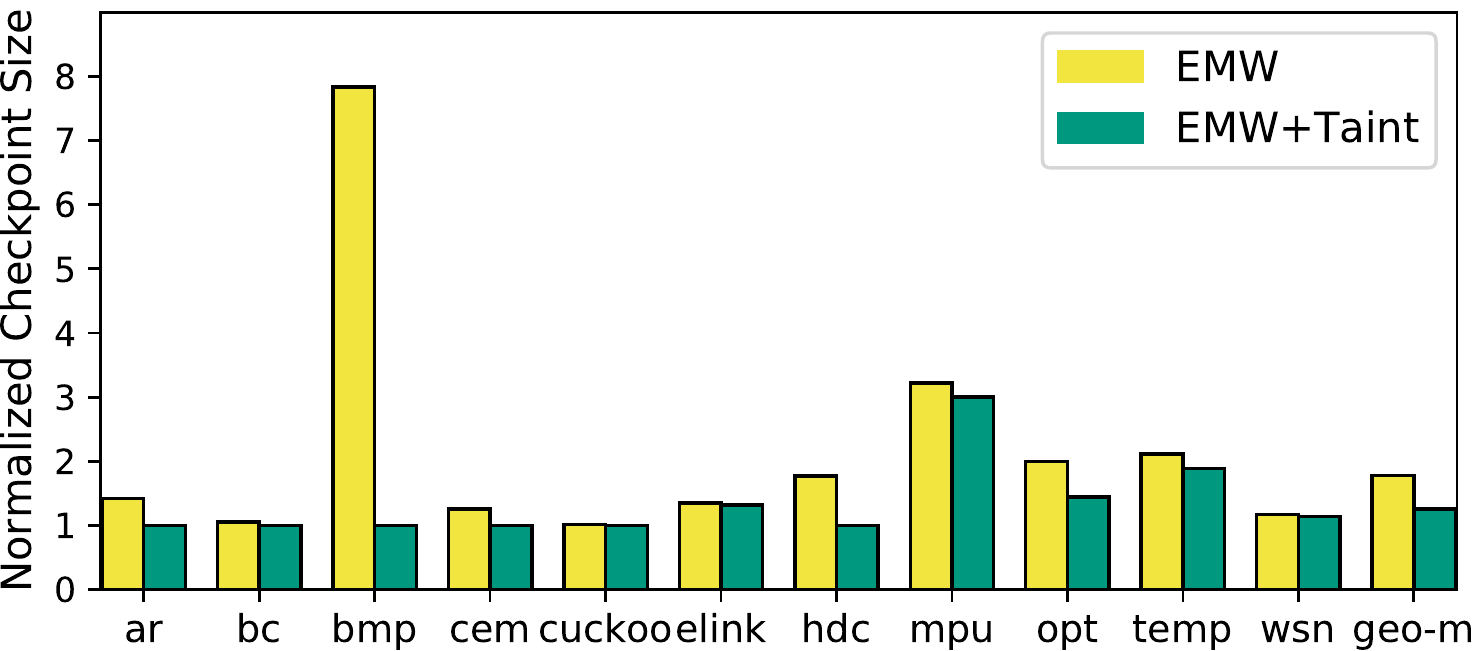}
    \vspace{-20pt}
    \caption{Normalized checkpoint size overhead caused by EMW tracking
    and taint-optimized EMW tracking.}
    \label{fig:memovhd}
    \end{wrapfigure}

Manual undo-logging requires modifying each use, and
adding an initialization and a backup operation, requiring changing at least $\sum
(\forall \loc \in \m{EMW}, 2 + \m{uses}(\loc))$ lines of code.
Concretely, {\tt mpu} required 22 changes, {\tt opt}
required four, {\tt temp} required seven, and {\tt wsn} required 17.  
As reported by IBIS, only these benchmarks had RIO bugs that required code changes to fix. 

Manual fixing is not only onerous, but ultimately duplicative: manual
backup introduces a WAR dependence on a task-shared variable
and any non-idempotent EMW variable will be added to the checkpoint anyway by
the WAR analysis after fixing manually. IBIS 
may also miss RIO bugs. Fixing only IBIS' reported bugs may result in still incorrect code. 
 Our taint-optimized EMW analysis
eliminates the risk of manual code fixes and directly backs up the
necessary data, ensuring correctness with reasonable overhead
and essentially no programming burden.

\section{Extending the Framework}
We use the presented framework to define and prove correctness with respect to 
memory consistency and input operations. There are more properties 
that must be reasoned about to truly develop provably correct, reliable intermittent systems.
This framework serves as a foundation that can be built upon to reach that 
ultimate goal.  We discuss the 
strategies to extend the framework to include other programming models or correctness properties.
Some programming models can be modeled by straightforward extensions, such 
as changes to the execution state or logging mechanisms (Section \ref{sec:variants}). Others 
require more significant changes to the current framework, such as guaranteeing forward progress. 

\subsection{Extending memory consistency to other programming models and architectures}
\label{sec:straight-ext}
Extending the memory consistency theorem to cover systems with different 
checkpointing algorithms or architectural state is straightforward.
The proof of the memory consistency theorem is built around two memory relations 1) 
current-to-initial execution point and 2) same execution point (Section ~\ref{sec:rio-relations}). 
If the intermittent and continuous states satisfy relation 1, the states remain 
related after each reboot. If they satisfy 2,  the states will have the same memory 
by the next checkpoint. To extend the framework to cover a new execution model,
 one can show that these relations hold directly, 
or one can show equivalence to the basic model as in Section~\ref{sec:variants}. 
A model with differing checkpoint placement 
or additional language instructions may satisfy these relations directly, such as \emph{just-in-time}
checkpointing ~\cite{hibernus,samoyed} or idempotent regions~\cite{ratchet}. 
If a model changes the state tuple or runtime constructs---e.g., the context or memory layout, 
additional architectural components---it is more practical to show correctness via bi-simulation, 
rather than tweaking the parameters of relations 1 and 2. 
In Section~\ref{sec:variants}, 
we showed equivalence to systems with differing program models. 
Below we sketch an extension to a different architecture, one that uses a write-back cache.

A write-back cache commits an update to memory only when the update is evicted from the cache, 
either due to an eviction policy or through explicit flushing. Thus, the order in which updates execute may differ 
from the order they are evicted and committed. Our target hardware has a write-through cache, allowing us to assume identical persist 
and execution order in the presentation of the basic model (Section~\ref{sec:framework}), but 
devices with write-back caches are reasonable future targets. We show that the framework can be simply extended
to handle this change in architecture.
The state tuple of the basic model should be extended with a cache $\cache$. An update to a location $l$ 
may be made to the cache instead of $\nvmem$, and a future eviction from the cache updates $\nvmem$ with 
the value, i.e., $\nvmem \lhd \cache[l]$. 
Note that this is almost exactly the behaviour of Redo logging (Section~\ref{sec:redo-log}); 
updates to checkpointed locations are logged, and the log is committed at a checkpoint. 
Updates to $\nvmem$  occur out of execution order. The key change to the semantics is 
that the commit must flush the newly added cache as well as the log, acting as a 
serialization point. Additionally, items in $\cache$ can be flushed before 
reaching a checkpoint. This behaviour does not introduce new inconsistencies, however, as 
region re-execution is idempotent. Consider an execution that caches updates 
$x$ and $y$, persisting only $y$ before failing. $y$ is after $x$ in execution order but 
before $x$ in persist order. If the first access to $y$ after reboot is a write, 
the previously persisted value is never accessed.  If the access is a read, 
then $y$ had a WAR dependence. $y$ is thus in $\omega$ and the update 
would have been made to the redo log.  Adding a write-back cache to a 
sequential redo-log model thus requires only minor changes to the semantics and 
bi-simulation relation, as the model already commits at checkpoints and 
safely redoes partial flushes. 

Undo-logging with a write-back cache is more complex. 
In the basic model presented, the back-up copies of variables are created at the checkpoint, 
so flushing after a checkpoint persists both cached updates from the previous region 
and all the backup copies. A more performant undo log that creates backup copies on demand 
 could become inconsistent if the update to general non-volatile memory 
persists before the update to the log, and power then fails.  A simple way to make caching correct 
is that any update to the log must be flushed immediately, so the log update always is persisted 
before the general update, but this destroys much of the benefit of having a cache for any 
variable in $\omega$. This issue of persisting log data before program data is a 
known and well studied problem with logging on persistent memory~\cite{raad2, crafty, atlas}.

Updating the checkpoint semantics to flush (and fence) the cache makes a checkpoint 
a serialization point between checkpointed regions. The WAR:ok and RIO:ok checks are static, 
and thus put any variable that could potentially be inconsistent 
(including those visible to a re-execution of the current checkpoint region) 
into the checkpoint set. Any out-of-order persists or variables from partial flushes 
are thus either over-written during the course of re-execution, were made to a log 
(redo model) or are undone when applying the checkpoint set (basic, undo-log model). 

As energy-harvesting devices develop to target more complex architectures, the modular next step is 
to add a hardware layer to the model to abstract away ordering details from the higher-level 
theorem. The instructions 
to fence and flush updates are ISA specific~\cite{raad2, raad3}. As shown above, 
changes to the architecture need not dramatically effect the memory consistency theorem. 
The interface with the hardware layer should provide certain assumptions, 
e.g., checkpoints are a serialization point, updates are linearizable. The proofs of these assumptions 
would not change the main correctness theorem, but can instead draw from formalizations in 
prior work~\cite{raad1, raad2, raad3, atlas}.

\subsection{Towards correctness properties beyond memory consistency}
Ensuring that programs are correct with respect to memory consistency 
is a crucial first step towards reliable intermittent computation, 
but there remain others, such as ensuring progress, correct timing, and concurrency 
(Section ~\ref{sec:other-props}). The presented 
framework can be extended modularly to reason about these properties.  

\paragraphb{Forward Progress}
To be able to guarantee forward progress, any possible trace between checkpoints must not 
consume more energy than can fit in the energy buffer. Reasoning about the energy consumption 
of a trace requires creating an energy model. This energy model must model the full system, 
including peripherals, as  energy consumption 
depends on all components on the board, not just the CPU.
The prior work CleanCut~\cite{cleancut} develops an energy model 
to guide programmers in creating appropriately sized tasks, but it is probabilistic, 
and furthermore does not consider the full system. Developing such a full-system, non-probabilistic energy model 
is a complex problem. 
While this energy model is necessary to reason about 
the forward progress property, it does not change the memory consistency correctness property presented in this 
work. Rather, an additional energy layer should be added to the framework. To be correct, any intermittent execution must 
correspond w.r.t. memory to some continuous execution, and additionally a trace between checkpoints 
must always take less energy than can fit in the energy buffer.
Thus, we anticipate that energy modelling can be added 
to the current framework modularly. 

\paragraphb{Concurrency through Interrupts}
While there are not yet multi-core intermittently-powered devices, some research~\cite{coati, ink}
addresses interrupt driven computation. In such execution models, inputs can be asynchronous and ephemeral
---after reboot, interrupts may not occur as they did before power-failure. Intermittent 
executions must be consistent, as in defined in the theorem presented here, but they must also correctly deal with 
concurrency, The 
updates to memory of any interrupt handlers---including those partially executed--and 
the main thread of program execution must be linearizable. Linearizability for persistent 
memory is a well-studied problem~\cite{ido, justdo, raad2,persistent-linearizability}.
Extending the framework requires adding interrupts to the semantics 
and adding linearizability to the proof. 

\paragraphb{Time-sensitivity}
In this work, as in ~\cite{reachability}, the continuous, non-crashy execution to which the 
intermittent, crashy execution corresponds can pause for arbitrary amounts of time, while 
the system recovers to a consistent state. Allowing these arbitrary pauses at any location 
can make the correctness definition too weak for programs whose behaviour depends on 
highly timing-sensitive input processing. The value an input operation returns depends on the time that it was 
gathered --- arbitrary pauses within a sequence of input operations can 
produce program behaviour not possible on a continuous execution without pauses. 
To be correct w.r.t timing of inputs, an intermittent execution 
must not only correspond to some continuous execution, but 
that continuous execution must be one without pauses in time-sensitive regions. 
Which regions are time-sensitive is frequently 
application dependent. Robust reasoning about time-sensitivity requires adding 
language constructs to describe the time-constraints on data, as well as mechanisms 
to preserve the constraints. As with the properties above, adding time-sensitivity to 
the framework 
does not change the underlying memory-consistency theorem, but adds another constraint 
to correct intermittent execution. 
\section{Conclusion}
We provide the first formal framework for examining the correctness of
intermittent systems, w.r.t memory consistency. We show the framework's usefulness by using it
to formalize intermittent systems with input operations, showing
that many existing systems do not meet reasonable correctness
criteria, and using the correctness invariants to implement a correct
runtime system. We further extend the framework to show that a variety
of existing systems are equivalent, indicating that the same correctness
properties hold for all the modeled systems. This framework lays the foundation for formally defining intermittent system correctness,
a crucial step towards the development of provably correct, reliable applications for intermittent systems.
 Future work should extend the framework to define properties beyond memory consistency, such as timeliness, forward progress, or concurrency.


\section*{Acknowledgements}
We thank the anonymous reviewers for their feedback, and members of 
the Abstract Research Lab for their insightful comments on initial drafts. 
We would like to thank Naomi Spargo for 
formalizing the theorem for the correctness of intermittent systems without inputs (Appendix C.2) 
in the Coq proof assistant, available at \url{https://github.com/misstaggart/intermittent_formalism}.
This work was generously funded through National Science Foundation Award 2007998
and National Science Foundation CAREER Award 1751029.

\begin{appendices}
\section{Syntax and Semantics of the Basic Checkpointing System}
\label{sec:checkpoint}
\subsection{Syntax}

We define a core calculus for
checkpoints.
We write $\omega$ to denote the set of
global variables and arrays that need to be stored across
checkpoints. Note that each array has a pre-defined, fixed bound. $a^n$ indicates
that array $a$ has length $n$. We omit the bounds for simplicity.
\[
\begin{array}{llcl}

\textit{values} & v & \bnfdef & n \bnfalt \m{true} \bnfalt \m{false} 
\\
\textit{expressions} & e & \bnfdef & x \bnfalt v \bnfalt e_1\;
                                     \m{bop}\; e_2 \bnfalt a[e] 
\\
\textit{war variables} & \omega & \bnfdef & \cdots \bnfalt \omega, x
                                            \bnfalt \omega, a^n
\\
\textit{instructions} & \iota &  \bnfdef& \m{skip} \bnfalt x:= e \bnfalt a[e]:= e' \bnfalt \m{checkpoint}(\omega) \bnfalt \m{reboot} 
\\
\textit{commands} & \cmd & \bnfdef& \iota  \bnfalt \iota;\cmd \bnfalt \m{if}\ e\ \m{then}\ \cmd_1\ \m{else}\ \cmd_2
\end{array}
\]

We distinguish between volatile and non-volatile memory. We write
$\context$ to denote checkpointed data, which is a triple consisting
of checkpointed non-volatile memory, volatile memory, and command to
execute at the time of the checkpoint instruction. Note that the
checkpointed non-volatile memory could be empty. 

\[
\begin{array}{llcl}
\textit{Memory locations} & \loc & \bnfdef & x\bnfalt  a[n] 
\\
\textit{Memory mapping} & M & \bnfdef & \m{Loc} \rightarrow \m{Val}
\\
\textit{Non-volatile memory} & \nvmem & : & M
\\
\textit{Volatile memory} & \vmem & : & M
\\
\textit{Continuous Conf.} & \sigma& \bnfdef & (\nvmem, \vmem,
                                           \cmd)
\\
\textit{Context} & \context & \bnfdef & (\nvmem, \vmem, \cmd)
\\
\textit{Intermittent Conf.} & \Sigma& \bnfdef & (\context, \nvmem, \vmem,
                                           \cmd)
\\
\textit{Read Observation} & r & \bnfdef & \m{rd}\ \loc\ v \bnfalt  r, r
\\
\textit{Observation} & o & \bnfdef & [r]\bnfalt \m{reboot} \bnfalt \m{checkpoint}
\\
\textit{Observation sequence} & O &  \bnfdef & \cdot \bnfalt O, o
\end{array}
\]

\subsection{Continuously-Powered Operational Semantics}

An observation sequence is determined by the following rules. We write
$N, V \vdash e \Downarrow_{r} v$ to denote that with memories $N$
and $V$, expression $e$ evaluates to value $v$ with observation
$r$. We write $\denote{v_1 \m{bop}\ v_2 }$ to denote the result of
computing a binary operation with values $v_1$, $v_2$.   

\begin{mathpar}
\inferrule*[right=Val]{ }{
   N, V \vdash v \Downarrow_{\cdot} v  
 }
\and
\inferrule*[right=BinOp]{ i \in [1, 2]\\N,V \vdash e_i
  \Downarrow_{r_i} v_i
  }{
  N, V \vdash e_1 \m{bop}\ e_2 \Downarrow_{r_1, r_2} \denote{v_1 \m{bop}\ v_2}
}
\and
\inferrule*[right=Rd-Var]{ N \cup V(x) = v}{
  N, V \vdash x \Downarrow_{\m{rd}\,x\,v} v
}
\and
\inferrule*[right=Rd-Arr]{ N,V \vdash e \Downarrow_{r_e} v_e
\\ N\cup V(a[v_e]) = v
}{
  N, V \vdash a[e] \Downarrow_{r_e, \m{rd}\,a[v_e]\, v} v
}

\end{mathpar}

We write $\proj{m}{\omega}$ to denote the part of $m$, whose domain is
 $\omega$. $a^n$ represents all the locations in the array $a$. That
 is: each $a^n$ in $\omega$ represents the set of locations
 $\{a[1], \cdots, a[n]\}$.
We write
$m[\loc\mapsto v]$ to denote the memory that is the same as $m$ except
that $\loc$ is mapped to $v$. We write $m_1\lhd m_2$ to denote the memory
resulted from updating $m_1$ with $m_2$.

We write $\sigma\SeqStepsto{O} \sigma'$ to
denote the semantics of sequential executions. The rules are the
same as those for the intermittent execution except that  the $\m{checkpoint}$
instruction behaves the same as $\m{skip}$, and that the state does
not need the checkpointed context $\context$, and that there are no
rules for fail or reboot.

~\\\noindent\framebox{$(\nvmem, \vmem, \cmd) 
 \SeqStepsto{O} (\nvmem', \vmem', \cmd')$}

\begin{mathpar}
\inferrule*[right=NV-Assign]{x \in \m{dom}(N)\\N, V \vdash e \Downarrow_{r} v }{
  (\nvmem, \vmem, x:=e) 
 \SeqStepsto{[r]}  (\nvmem[x\mapsto v], \vmem, \m{skip}) 
}%
\quad
\inferrule*[right=V-Assign]{x \in \m{dom}(V)\\N, V \vdash e \Downarrow_{r} v }{
  (\nvmem, \vmem, x:=e) 
 \SeqStepsto{[r]}  (\nvmem, \vmem[x\mapsto v], \m{skip}) 
}
\and
\inferrule*[right=Assign-Arr]{N, V \vdash e \Downarrow_{r} v \\N, V \vdash e' \Downarrow_{r'} v' }{
  (\nvmem, \vmem, a[e]:=e') 
 \SeqStepsto{[r, r']}  (\nvmem[a[v]\mapsto v'], \vmem, \m{skip}) 
}
\and
\inferrule*[right=CheckPoint]{ }{
  (\nvmem, \vmem, \m{checkpoint}(\omega);\cmd) 
 \SeqStepsto{\m{checkpoint}}  (\nvmem, \vmem, \cmd) 
}
\and
\inferrule*[right=Skip]{ }{
  (\nvmem, \vmem, \m{skip};\cmd) 
 \SeqStepsto{}  (\nvmem, \vmem, \cmd) 
}
\and
\inferrule*[right=Seq]{  (\nvmem, \vmem, \iota) 
 \SeqStepsto{o}  (\nvmem', \vmem', \m{skip})}{
  (\nvmem, \vmem, \iota;c) 
 \SeqStepsto{o}  (\nvmem', \vmem', c) 
}
\and
\inferrule*[right=If-T]{N, V \vdash e \Downarrow_{r} \m{true}}{
  (\nvmem, \vmem, \m{if}\ e\ \m{then}\ \cmd_1\ \m{else}\ \cmd_2) 
 \SeqStepsto{[r]}  (\nvmem, \vmem, \cmd_1) 
}
\and
\inferrule*[right=If-F]{N, V \vdash e \Downarrow_{r} \m{false}}{
  (\nvmem, \vmem, \m{if}\ e\ \m{then}\ \cmd_1\ \m{else}\ \cmd_2) 
 \SeqStepsto{[r]}  (\nvmem, \vmem, \cmd_2) 
}
\end{mathpar}

\subsection{Intermittent Operational Semantics}

We write $(\context, \nvmem, \vmem, \cmd) \Stepsto{O} (\context',
\nvmem', \vmem', \cmd')$ to denote the small-step operational
semantics of the core calculus. The rules are summarized below. 

~\\\noindent\framebox{$(\context, \nvmem, \vmem, \cmd) 
 \Stepsto{O} (\context', \nvmem', \vmem', \cmd')$}

\begin{mathpar}

\inferrule*[right=CP-PowerFail]{ }{
  (\context, \nvmem, \vmem, \cmd) 
 \Stepsto{}  (\context, \nvmem, \resetm(\vmem), \m{reboot}) 
}
\and
\inferrule*[right=CP-CheckPoint]{ }{
  (\context, \nvmem, \vmem, \m{checkpoint}(\omega);\cmd) 
 \Stepsto{\m{checkpoint}}  ((\proj{\nvmem}{\omega}, \vmem, \cmd), \nvmem, \vmem, \cmd) 
}
\and
\inferrule*[right=CP-Reboot]{ \context=(\nvmem, \vmem, \cmd)}{
  (\context, \nvmem', \vmem', \m{reboot}) 
 \Stepsto{\m{reboot}}  (\context, \nvmem'\lhd\nvmem, \vmem, \cmd) 
}
\end{mathpar}
\begin{mathpar}
\inferrule*[right=CP-NV-Assign]{x \in \m{dom}(N)\\N, V \vdash e \Downarrow_{r} v }{
  (\context, \nvmem, \vmem, x:=e) 
 \Stepsto{[r]}  (\context, \nvmem[x\mapsto v], \vmem, \m{skip}) 
}
\and
\inferrule*[right=CP-Assign-Arr]{N, V \vdash e \Downarrow_{r} v \\N, V \vdash e' \Downarrow_{r'} v' }{
  (\context, \nvmem, \vmem, a[e]:=e') 
 \Stepsto{[r, r']}  (\context, \nvmem[a[v]\mapsto v'], \vmem, \m{skip}) 
}
\and
\inferrule*[right=CP-V-Assign]{x \in \m{dom}(V)\\N, V \vdash e \Downarrow_{r} v }{
  (\context, \nvmem, \vmem, x:=e) 
 \Stepsto{[r]}  (\context, \nvmem, \vmem[x\mapsto v], \m{skip}) 
}
\and
\inferrule*[right=CP-Skip]{ }{
  (\context, \nvmem, \vmem, \m{skip};\cmd) 
 \Stepsto{}  (\context, \nvmem, \vmem, \cmd) 
}
\and
\inferrule*[right=CP-Seq]{  (\context, \nvmem, \vmem, \iota) 
 \Stepsto{o}  (\context, \nvmem', \vmem', \m{skip})}{
  (\context, \nvmem, \vmem, \iota;c) 
 \Stepsto{o}  (\context, \nvmem', \vmem', c) 
}
\and
\inferrule*[right=CP-If-T]{N, V \vdash e \Downarrow_{r} \m{true}}{
  (\context, \nvmem, \vmem, \m{if}\ e\ \m{then}\ \cmd_1\ \m{else}\ \cmd_2) 
 \Stepsto{[r]}  (\context, \nvmem, \vmem, \cmd_1) 
}
\and
\inferrule*[right=CP-If-F]{N, V \vdash e \Downarrow_{r} \m{false}}{
  (\context, \nvmem, \vmem, \m{if}\ e\ \m{then}\ \cmd_1\ \m{else}\ \cmd_2) 
 \Stepsto{[r]}  (\context, \nvmem, \vmem, \cmd_2) 
}
\end{mathpar}

\section{Checkpointed Data (WAR)}
\label{sec:cp-data}

\subsection{Algorithm for Checking Checkpointed Data (WAR Variables)}~\\
\label{sec:war-checking}

Rules of the form $N; W; R \Vdash \iota: W';R'$ check each memory
access is well-formed w.r.t. the checkpointed variables in $N$,
written set $W$, and read set $R$, and returns newly written set $W'$
and newly read set $R'$.

\noindent\framebox{$N; W; R \Vdash_\war \iota: W';R'$}
\begin{mathpar}
\inferrule*[right=WAR-Skip]{ }{
  N; W;R \Vdash_\war \m{skip}: W; R
}
\and
\inferrule*[right=WAR-NoRd]{ R' = R \cup rd(e) \\ x \notin R'}{
  N; W;R \Vdash_\war x:= e : W \cup x; R' 
}
\and
\inferrule*[right=WAR-Checkpointed]{ R' = R \cup rd(e) 
\\ x \in R' 
\\ x \notin W
\\ x \in N
}{
 N; W;R \Vdash_\war x:= e : W \cup x; R' 
}
\and
\inferrule*[right=WAR-Wt]{ R' = R \cup rd(e) \\ x \in R' \\ x \in W}{
  N;W;R \Vdash_\war x:= e : W; R' 
}
\and
\inferrule*[right=WAR-NoRd-Arr]{ R' = R \cup rd(e) \cup rd(e') \\ a \notin R'}{
  N; W;R \Vdash_\war a[e'] := e : W \cup a; R' 
}
\and
\inferrule*[right=WAR-Checkpointed-Arr]{ R' = R \cup rd(e) \cup rd(e') 
\\  a \in R' \\ a \in N }{
  N; W;R \Vdash_\war a[e']:= e : W \cup a; R' 
}
\end{mathpar}

Judgment $N;W;R \Vdash_\war \cmd:\m{ok}$ means that all of $\cmd$'s WAR
variables are in $N$, given $R$ is the set of variables that are read
from the most recent checkpoint, $W$ is the set of written variables, 
and $N$ is the set of checkpointed
variables in the most recent checkpoint instruction.

\noindent\framebox{$N;W;R \Vdash_\war \cmd:\m{ok}$}
\begin{mathpar}
\inferrule*[right=WAR-I]{   N;W;R \Vdash_\war \iota: W'; R'}{
  N;W;R \Vdash_\war \iota: \m{ok}
}
\and
\inferrule*[right=WAR-Cp]{\omega;\emptyset;\emptyset \Vdash_\war \cmd: \m{ok}}{
  N;W; R \Vdash_\war \m{checkpoint} (\omega);\cmd : \m{ok}
}
\and
\inferrule*[right=WAR-Seq]{N;W; R \Vdash_\war \iota: W';R' \\ N;W'; R' \Vdash_\war \cmd: \m{ok}}{
  N;W;R \Vdash_\war \iota;\cmd : \m{ok}
}

\inferrule*[right=WAR-If]{ R' = R \cup rd(e) \\ N, W, R' \vdash \cmd_1: \m{ok}
 \\ N, W, R' \vdash \cmd_2: \m{ok}}{
  N;W;R \Vdash_\war \m{if}\ e\ \m{then}\ \cmd_1\ \m{else}\ \cmd_2: \m{ok} 
}

\end{mathpar}

\subsection{Algorithms for Collecting Checkpointed Locations}
\label{sec:war-collection}
Existing systems typically implement an algorithm for identifying WAR
variables. We show two variants here and prove that both of the
algorithms produce programs that pass the WAR checking defined in the
previous section.

\paragraphb{Algorithm used by Dino}~\\

\noindent\framebox{$N;W;R \Vdash_\mt{DINO} \iota: N'; W', R' $}
\begin{mathpar}
\inferrule*[right=D-WAR-Skip]{ }
{N;W;R \Vdash_\mt{DINO} \m{skip}: N;W; R
}
\quad
\inferrule*[right=D-WAR-Written]{ R' = R \cup rd(e) \\ x \notin R'}{
  N;W;R \Vdash_\mt{DINO} x:= e : N;W \cup x; R' 
}
\and
\inferrule*[right=D-WAR-CP-Asgn]{ R' = R \cup rd(e) \\ x \in R' \\ x \notin W}{
  N;W;R \Vdash_\mt{DINO} x:= e : N \cup x; W \cup x; R' 
}
\and
\inferrule*[right=D-WAR-WtDom]{ R' = R \cup rd(e) \\ x \in R' \\ x \in W}{
  N;W;R \Vdash_\mt{DINO} x:= e : N;W; R' 
}
\and
\inferrule*[right=D-WAR-Wt-Arr]{ R' = R \cup rd(e) \cup rd(e') \\ a \notin R'}{
  N;W;R \Vdash_\mt{DINO} a[e'] := e : N; W \cup a; R' 
}
\and
\inferrule*[right=D-WAR-CP-Arr]{ R' = R \cup rd(e) \cup rd(e') \\ a \in R'}{
  N; W;R \Vdash_\mt{DINO} a[e']:= e : N \cup a; W \cup a; R' 
}
\end{mathpar}

\noindent\framebox{$N;W;R \Vdash_\mt{DINO} \cmd \longrightarrow \cmd' : N'$}
\begin{mathpar}
\inferrule*[right=D-WAR-Instr]{ N;W;R \Vdash_\mt{DINO} \iota : N';W';R' }{
  N;W;R \Vdash_\mt{DINO} \iota \longrightarrow \iota: N'
}
\and
\inferrule*[right=D-WAR-Seq]{ N;W;R \Vdash_\mt{DINO} \iota: N';W'; R'
 \\ N';W';R'  \Vdash_\mt{DINO} \cmd \longrightarrow \cmd' : N''}{
  N;W;R \Vdash_\mt{DINO} \iota;\cmd \SeqStepsto{ } \iota;\cmd' : N'' 
}
\end{mathpar}
\begin{mathpar}
%
\inferrule*[right=D-WAR-If]{ R' = R \cup rd(e) 
\\ N, W, R' \vdash \cmd_i  \longrightarrow \cmd_i': N_i  
 \\ i \in [1, 2]}{
  N;W;R \Vdash_\mt{DINO} \m{if}\ e\ \m{then}\ \cmd_1\ \m{else}\ \cmd_2
  \longrightarrow \m{if}\ e\ \m{then}\ \cmd_1'\ \m{else}\ \cmd_2': N_1 \cup N_2
}
\and
\inferrule*[right=D-WAR-CP]{\emptyset;\emptyset;\emptyset \Vdash_\mt{DINO} \cmd \SeqStepsto{} \cmd': N'}{
  N;W;R \Vdash_\mt{DINO} \m{checkpoint()};\cmd \SeqStepsto{} \m{checkpoint(N')};\cmd' : N
}
\and

\end{mathpar}

\begin{lem}\label{lem:war-dino-superset}
\begin{enumerate}
\item If $\ee::N;W;R \Vdash_\mt{DINO} \iota: N'; W', R' $ then $N'\supseteq N$
\item If $\ee::N;W;R \Vdash_\mt{DINO} \cmd \mapsto \cmd' : N'$
then $N'\supseteq N$. 
\end{enumerate}
\end{lem}
\begin{proofsketch} By induction over the structure of $\ee$. (2) uses (1).
\end{proofsketch}

\begin{lem}\label{lem:war-dino-iota}
If $\ee::N;W;R \Vdash_\mt{DINO}  \iota: N'; W', R' $ then 
$\forall$ $N_1\supseteq N'$, $N_1;W;R \Vdash_\war  \iota: W', R'$ 
\end{lem}
\begin{proofsketch} 
By induction over the structure of $\ee$. 
\end{proofsketch}

\begin{lem}[DINO WAR collection algorithm is correct]
\label{lem:war-dino-correct}
If $\ee::N;W;R \Vdash_\mt{DINO} \cmd \longrightarrow \cmd' : N'$
then  $\forall$ $N_1\supseteq N'$,   $N_1;W;R \Vdash_\war \cmd' : \m{ok}$.
\end{lem}
\begin{proofsketch} 
By induction over the structure of $\ee$. Uses
Lemma~\ref{lem:war-dino-superset} and Lemma~\ref{lem:war-dino-iota}. 
\end{proofsketch}

\section{Correctness of Basic Checkpointed System}
\label{sec:cp-correctness}
\subsection{Auxiliary definitions}
\label{sec:cp-correct-defs}
We first define relations between traces emitted by continuous
executions and traces emitted by intermittent executions. We write
$O_1 \oleq O_2$ to mean that $O_1$, emitted by an intermittent
execution, is a prefix of $O_2$, emitted by a continuous execution of
the same program .  Here, the intermittent execution has not
experienced a power failure nor a checkpoint. Note that
both $O_1$ and $O_2$ only contain read actions.
We write $O_1 \oleqm O_2$ to mean that $O_1$, emitted by an intermittent
execution, is a sequence of prefixes of
$O_2$, emitted by a continuous execution of the same program . 
Here, the intermittent execution may have experienced a number of
power failures and subsequent reboots, but have not passed a
checkpoint. Here, $O_1$ may contain have reboot and read actions and
$O_2$ only has reads.
We write $O_1 \oleqmc O_2$ to mean that $O_1$, emitted by an
intermittent execution, is a sequence of sequences of prefixes of
fragments of $O_2$, emitted by
a continuous execution of the same program.  Here, the intermittent
execution may have experienced a number of power failures and
subsequent reboots and a number of checkpoints. $O_1$ can have reads, reboots
and checkpoints, $O_2$ has (no-op) checkpoints and read observations.

\begin{mathpar}

\inferrule*[right=Rb-Base]{ }{
  O \oleqm O
}
\and 
\inferrule*[right=Rb-Ind]{O_1 \oleq O_2 \\ O_1' \oleqm O_2}{
  O_1, \m{reboot}, O_1' \oleqm O_2
}
\\
\inferrule*[right=Cp-Base]{ O_1\oleqm O_2}{
  O_1 \oleqmc O_2
}
\and
\inferrule*[right=Cp-Ind]{O_1 \oleqm O_2 \\ O_1' \oleqmc O_2'}{
  O_1, \m{checkpoint}, O_1' \oleqmc O_2, \m{checkpoint},  O_2'
}

\end{mathpar}

We write $\mt{CP}(\Sigma)$ to mean that the first instruction in
 $\Sigma$ is $\m{checkpoint}$. We write $\sigma\MSeqStepsto{} \m{CP}$ to denote the trace from $\sigma$
to the nearest checkpoint. We write $\Sigma\MStepsto{} \m{CP}$ to denote the trace from $\Sigma$
to the nearest checkpoint.

\begin{defn}[Relating memories at the same execution point]
\label{def:mem-same_exec} 
$N_0, V_0, c_0, c' \vdash \nint \sim \ncont$ iff 
\\ $\m{dom}(\nint)= \m{dom}(\ncont)$, and
let $T_1 = (N_0, V_0, c_0) \MSeqStepsto{}  (\nint, V', c')$, 
$T_2 = (\nint, V', c') \MSeqStepsto{} \m{CP}$, where $T_1$ does not contain checkpoints, 
$\forall \loc \in \nint$ s.t. $\nint(\loc) \neq \ncont(\loc)$,  
$\loc \in \mt{Wt}(T_2)\cap\mt{FstWt}(T_1\cdot T_2)$ and $\loc\notin \mt{WT}(T_1)$.
\end{defn} 

\begin{defn}[Relating memories between current and initial execution point]
\label{def:mem-initial}  
  $N_c, N_{\mt{rb}}, V_0, c_0 \vdash \nint \sim \ncont$ iff 
  $\m{dom}(\nint)= \m{dom}(\ncont)$, 
 $N_c\subseteq \ncont$, 
  $\forall \loc \in \nint$ s.t. $\nint(\loc) \neq \ncont(\loc)$,
  $\loc \in \mt{FstWt}(T)\cup\m{dom}(N_c)$, where $T_1 = (N_{\mt{rb}}, V_0, c_0) \MSeqStepsto{}  (\nint, V', c')$, 
  $T_2 = (\nint, V', c') \MSeqStepsto{} \m{CP}$, and $T= T_1 \cdot T_2$. 
\end{defn} 

\begin{defn}[Related configurations]
$N_{\mt{rb}} \vdash (\context, \nint, \vmem_1, \cmd_1) \sim (\ncont, \vmem_2, \cmd_2)$ 
iff $\context = (\nvmem_c, \vmem_0, \cmd_0)$, $\nvmem_c \subseteq N_{\mt{rb}}$, 
and $N_{\mt{rb}}, \vmem_0, \cmd_0, \cmd_1 \vdash \nint \sim \ncont$,
$\vmem_1=\vmem_2$, and $\cmd_1=\cmd_2$.
\end{defn} 

\begin{defn}[Erased configuration]
$\erase{(\context, \nvmem, \vmem, \cmd)} = ( \nvmem, \vmem, \cmd)$ 
\end{defn} 

\begin{lem}[Solid to Dash]\label{lem:mem-relation-relation}
If $N_c, N_0, V_0, c_0\vdash \nint \sim \ncont$, $N_c; \emptyset; \emptyset \Vdash 
c_0: \m{ok}$, $N_c\subseteq \nvmem_0$, $N_c\subseteq \nint$ and  $N_c \subseteq \ncont$ 
then $\nint, V_0, c_0, c_0 \vdash \nint \sim \ncont$. 
\end{lem}
\begin{proof} 
  By examining the two relations and Lemma~\ref{lem:wt-fst-version}.
\end{proof}

\begin{lem}[Dash to Solid]\label{lem:mem-same-to-init}
  If $\nint, V_0, c_0, c_0 \vdash \nint \sim \ncont$,  and $N_c; \emptyset; \emptyset \Vdash 
  c_0: \m{ok}$,  $N_c\subseteq \nint$ and  $N_c \subseteq \ncont$,
  then $N_c, \nint, V_0, c_0 \vdash \nint \sim \ncont$
\end{lem}
\begin{proof}
By examining the two relations.
\end{proof}

\begin{lem}[Solid to Solid after Reboot]\label{lem:solid-reboot}
  If $N_{\mt{ckpt}},N_{\mt{rb_0}}, V_0, c_0 \vdash \nint \sim \ncont$
  , $N_{\mt{rb}} = \nint \lhd N_{\mt{ckpt}}$,
  and $N_{\mt{ckpt}}; \emptyset; \emptyset \Vdash 
  c_0: \m{ok}$, 
  then $N_{\mt{ckpt}},N_{\mt{rb}}, V_0, c_0 \vdash N_{\mt{rb}} \sim \ncont$
\end{lem}
\begin{proof}
By examining the two relations and Lemma~\ref{lem:wt-fst-version}.
\end{proof}

\subsection{Correctness Proofs}
\label{sec:cp-correct-proofs}

The Correctness Theorem for intermittent systems without Inputs follows from the following Lemma.
\begin{lem}[Correctness]
\label{lem:correctness-basic}
If
  $T = (\context, \nvmem, \vmem, \cmd) \MStepsto{O_1} \Sigma$,
$\context = (\nvmem_0, \vmem, \cmd)$,  $\nvmem_0\subseteq \nvmem$,
$ \nrb = \m{nearestRb(T)}$
and  $\nvmem_0, \emptyset, \emptyset \Vdash \cmd: \m{ok}$ 
then $\exists O_2, \sigma$ s.t. 
\begin{enumerate}
\item   $(\nvmem, \vmem, \cmd) \MSeqStepsto{O_2}
  \sigma$ and $\nrb \vdash \Sigma \sim \sigma$
and
\item $\forall T' = \Sigma \MStepsto{O} \Sigma'$, $\mt{CP}(\Sigma')$, 
and $T'$ does not contain checkpoints or reboots
implies
$\sigma \MSeqStepsto{O} \erase{(\Sigma')}$
and $O_1, O\oleqmc O_2, O$.
\end{enumerate}
\end{lem}
\ifproofs
\begin{proof} By induction on the number of checkpoints in $O_1$.
\begin{description}
\item[Base case:] $O_1$ does not include any checkpoint, directly
  apply Lemma~\ref{lem:one-cp-m-f}.
\item[Inductive case:] $O_1$ contains $k+1$ checkpoints where $k\geq
  0$
\begin{tabbing}
By assumption, 
\\\qquad \= (1)~~ \= 
$T= (\context, \nvmem, \vmem, \cmd) \MStepsto{O_1}
(\context, \nvmem_1, \vmem_1, \m{checkpoint}(\omega);\cmd_1) 
\Stepsto{\m{checkpoint}} \Sigma' \MStepsto{O_2} \Sigma''$,
\\\>\> where $O_1$ does not contain checkpoints
\\By Lemma~\ref{lem:one-cp-m-f} 
\\\>(2)\>$(\nvmem, \vmem, \cmd) \MSeqStepsto{O'_1}
(\nvmem_1, \vmem_1, \m{checkpoint}(\omega);\cmd_1) $ and $O_1 \oleqm
O'_1$
\\By \rulename{CP-CheckPoint} and (1)
\\\>(3)\> $\Sigma' =(\context_1, \nvmem_1, \vmem_1, \cmd_1)$ where
$\context_1 =  (\proj{\nvmem_1}{\omega}, \vmem_1, \cmd_1)$
\\By \rulename{CheckPoint} and (2)
\\\>(4)\> $(\nvmem_1, \vmem_1, \m{checkpoint}(\omega);\cmd_1)
\SeqStepsto{\m{checkpoint}} (\nvmem_1, \vmem_1, \cmd_1)$ 
\\By  $\nvmem_0, \emptyset , \emptyset \Vdash \cmd: \m{ok}$ and
Lemma~\ref{lem:cmd-runtime-ok}
\\\>(5)\> exists $W'. R'$ s.t. $\nvmem_0, W', R' \vDash
\m{checkpoint}(\omega);\cmd_1:\m{ok}$
\\By inversion of (5)
\\\>(6)\>  $\omega, \emptyset, \emptyset  \Vdash \cmd_1: \m{ok}$
\\By I.H. on the tail of $T$, (3), and (6)
\\\>(7)\>  $(\nvmem_1, \vmem_1, \cmd_1) \MSeqStepsto{O'_2}
  \sigma''$ and $\nrb \vdash \Sigma'' \sim \sigma''$
and
\\\>(8)\> and 
 $\Sigma'' \MStepsto{O_3} (\context''', \nvmem''', \vmem''', \m{checkpoint}(\omega');\cmd''')$
implies
\\\>\>$\sigma'' \MSeqStepsto{O_3} (\nvmem''', \vmem''', \m{checkpoint}(\omega');\cmd''')$,
and $O_2, O_3\oleqmc O'_2, O_3$.
\\By (2) and (8) and \rulename{Cp-Ind}
\\\>(9)\> $O_1, \m{checkpoint}, O_2, O_3\oleqmc O'_1, \m{checkpoint}, O'_2, O_3$ 
\\By connection of the executions the conclusion holds
\end{tabbing}
\end{description}
\end{proof}
\fi

\begin{lem}\label{lem:cmd-runtime-ok}
  If $\nvmem_0, W, R \Vdash \cmd: \m{ok}$ and
  $(\nvmem, \vmem, \cmd)\MSeqStepsto{O} (\nvmem', \vmem', \cmd')$, and $O$ does not
  contain checkpoints, then $\exists$ $W'$, $R'$ s.t.
  $\nvmem_0, W', R' \Vdash \cmd': \m{ok}$.
\end{lem} 
\begin{proofsketch}
By induction over the derivation $\nvmem_0, W, R \Vdash \cmd:
\m{ok}$. 
\end{proofsketch}

\begin{lem}[One checkpoint, multiple failures]
\label{lem:one-cp-m-f}
If
  $T = (\context, \nvmem_1, \vmem, \cmd) \MStepsto{O_1} \Sigma$
  where $T$ does not execute checkpoint, $\nrb = \m{nearestRb(T)}$
 $\context = (\nvmem_0, \vmem, \cmd)$, $\nvmem_0\subseteq\nvmem_1$, $\nvmem_0\subseteq\nvmem_2$,
 and $\nvmem_0, \nvmem_1, \vmem, \cmd \vdash \nvmem_1 \sim \nvmem_2$ 
 and  $\nvmem_0, \emptyset, \emptyset \Vdash \cmd: \m{ok}$ 
then
\begin{itemize}
\item  $\exists O_2, \sigma$ s.t. $(\nvmem_2, \vmem, \cmd) \MSeqStepsto{O_2}
  \sigma$ and $\nrb \vdash \Sigma \sim \sigma$ and
\item 
$\forall T' = \Sigma \MStepsto{O} \Sigma'$ where $\Sigma' =
(\context', \nvmem', \vmem', \m{checkpoint}(\omega);\cmd')$ and $T'$
does not contain checkpoints or reboots
implies
$\sigma \MSeqStepsto{O} \erase{(\Sigma')}$,
and $O_1, O\oleqm O_2, O$.
\end{itemize}
\end{lem}
\ifproofs
\begin{proof}
By induction on the number of reboots in $O_1$.
\begin{description}
\item[Base case:] $O_1$ does not include any reboot, first apply
  Lemma~\ref{lem:mem-relation-relation} to show $N_1, V, c, c \vdash
  N_1 \sim N_2$; then directly
  apply Lemma~\ref{lem:one-f} and Lemma~\ref{lem:one-f-completion}.
\item[Inductive case:] $O_1$ contains $k+1$ reboots where $k\geq
  0$
\begin{tabbing}
By assumption
\\\qquad \= (1)\quad \= 
$T= (\context, \nvmem_1, \vmem, \cmd) \MStepsto{O'_1}
(\context, \nvmem_1', \vmem_1, \m{reboot}(\omega)) 
\Stepsto{\m{reboot}} \Sigma' \MStepsto{O''_1} \Sigma''$,
\\\>\> and $O'_1$ does not contain reboots or checkpoints
\\ Use Lemma~\ref{lem:partial-run} to bring the relation of
non-volatile memories to satisfy the precondition after reboot.
\\By assumption $\nvmem_0, \nvmem_1, \vmem, \cmd \vdash \nvmem_1 \sim \nvmem_2$
\\By Lemma~\ref{lem:partial-run} 
\\\>(2)\> $\nvmem_0, \nvmem_1, \vmem, \cmd \vdash \nvmem_1' \sim \nvmem_2$.  
It follows that  $\nvmem_0, \nvmem_1, \vmem, \cmd \vdash \nvmem_1' \lhd\nvmem_0 \sim \nvmem_2$
\\By (2) and Lemma~\ref{lem:solid-reboot}
\\\>(2b) $\nvmem_0, \nvmem_1' \lhd\nvmem_0, \vmem, \cmd \vdash \nvmem_1' \lhd\nvmem_0 \sim \nvmem_2$
\\By \rulename{CP-Reboot} and (1)
\\\>(3)\> $\Sigma' = (\context,\nvmem_1' \lhd \nvmem_0,\vmem,\cmd)$
\\By I.H. on the tail of T, (2) and (2b)
\\\> (4)\> exists $O_2$ and $\sigma$  s.t. $(\nvmem_2, \vmem, \cmd) \MSeqStepsto{O_2}
  \sigma$  and $\nrb \vdash \Sigma'' \sim \sigma $
\\\> (5)\> $\Sigma'' \MStepsto{O} (\context', \nvmem', \vmem', \m{checkpoint}(\omega);\cmd')$
\\\>\>implies
$\sigma \MSeqStepsto{O} (\nvmem', \vmem',
\m{checkpoint}(\omega);\cmd')$ and $ O''_1, O\oleqm O_2, O$
\\ Next, we reuse the proof of the base case to show $O'_1 \oleq O_2, O$
\\ By (1) and expanding the first part of $T$
\\\>(6)\> $T = (\context, \nvmem_1, \vmem, \cmd) \MStepsto{O'_1}
\Sigma_{f} \Stepsto{} (\context, \nvmem_1', \vmem_1, \m{reboot}()) $
\\By assumption that configurations can always make progress 
\\\>(7)\>  exists $O_f$ s.t. $\Sigma_f \MStepsto{O_f} 
  (\context'_f, \nvmem'_f, \vmem'_f, \m{checkpoint}(\omega);\cmd'_f)$
\\By the base case proof
\\\>(8)\> exists $O'_2$ s.t. $(\nvmem_2, \vmem, \cmd) \MSeqStepsto{O'_2}
(\nvmem'_f, \vmem'_f, \m{checkpoint}(\omega);\cmd'_f)$ and $O'_1, O_f
\oleq O'_2$
\\By the semantics are deterministic, (8) and (4), (5)
\\\>(9)\> $O'_2 = O_2, O$
\\By  (8), (9) and the definition of $\oleq$
\\\>(10)\> $O'_1 \oleq O_2, O$
\\By (5) and (10) and \rulename{Rb-Ind}
\\\>(11)\>$ O'_1, \m{reboot}, O''_1, O\oleqm O_2, O$
\\By (5), (11) the conclusion holds.
\end{tabbing}
\end{description}
\end{proof}
\fi

\begin{lem}[Partial run relates to initial state (Multi-steps)]
\label{lem:partial-run}
If
$T= (\context, \nvmem_1, \vmem, \cmd_0) \MStepsto{O} (\context,
\nvmem'_1, \vmem', \cmd)$,  $T$ does not
  execute checkpoint or reboot, $\nvmem_0; \emptyset; \emptyset\Vdash \cmd_0: \m{ok}$,
$\context = (\nvmem_0, \vmem, \cmd_0)$,  and $\nvmem_0, \nvmem_1, \vmem, \cmd_0 \vdash \nvmem_1 \sim \nvmem_2$ 
then 
 $\nvmem_0, \nvmem_1, \vmem,\cmd_0 \vdash \nvmem'_1 \sim \nvmem_2$. 
\end{lem}
\begin{proof}
By induction over the length of $T$. The base case is trivial. The
inductive case uses I.H. and Lemma~\ref{lem:partial-run-one-step}. 
\end{proof}

\begin{lem}[Relating syntactic and semantic read]~\label{lem:ok-read-same}
If 
$\ee:: \nvmem_0; W_0; R_0\Vdash \cmd_0: \m{ok}$, 
and $T = (N_1, V_0, \cmd_0) \MSeqStepsto{}  (N, V, c)$, $T$ does not
contain any checkpoints
then 
$\nvmem_0; Wt(T) \cup W_0; Rd(T) \cup R_0 \Vdash \cmd: \m{ok}$.
\end{lem}
\begin{proofsketch}
Induction over the derivation $\ee$. 
\end{proofsketch}

\begin{lem}[Writes are either First Writes or Checkpointed]
\label{lem:wt-fst-version}
If  $\nvmem_0; \emptyset; \emptyset\Vdash \cmd_0: \m{ok}$, 
let $T_c$ be a trace from $\cmd_0$ to the nearest checkpoint:
 $T_c = T_0\cdot (N, V,  x:=e; \cmd_1)  \MSeqStepsto{O}  (N', V', \m{checkpoint}(\omega);c')$ 
then $x\in \mt{FstWt}(T_c)\cup \nvmem_0$
\end{lem}
\begin{proof}
By induction over the number of writes to $x$ in $T_0$
\begin{description}
\item[Base case:]  $T_0$ does not contain any writes to $x$
\begin{tabbing}
\\We consider two subcases: (I) $x\in \mt{RD}(T_0)\cup \mt{rd}(e)$
  and (II) $x\notin \mt{RD}(T_0)\cup \mt{rd}(e)$ 
\\{\bf Subcase (I)} $x\in \mt{RD}(T_0)\cup \mt{rd}(e)$, 
\\By Lemma~\ref{lem:ok-read-same} 
\\\quad \= (1)\quad \= 
 $\ee':: \nvmem_0; Wt(T_0); Rd(T_0) \Vdash x:=e; \cmd_1: \m{ok}$,
 and 
\\By inversion of $\ee'$, assumption that $x \notin \m{Wt}(T_0)$ and (1), $x\in \nvmem_0$ 
\\{\bf Subcase (II)} $x\notin \mt{RD}(T_0)\cup \mt{rd}(e)$
\\By definition of $\mt{FstWt}$, $x\in \mt{FstWt}(T_c)$
\end{tabbing}
\item[Inductive case:] $T_0$ contains $n+1$ writes to $x$
\begin{tabbing}
By assumptions
\\\quad \= (1)\quad \=
  $T_0 = T_1\cdot (N, V,  x:=e; \cmd_2)  \MSeqStepsto{O_1} (N, V,
  x:=e; \cmd_1)$
\\By I.H. on $T_1$, the conclusion holds.
\end{tabbing}
\end{description}
\end{proof}

\begin{lem}[Partial run relates to initial state (One step)]
\label{lem:partial-run-one-step}
If  all of the following hold
\begin{itemize}
\item $T= (\context, \nvmem, \vmem, \cmd) \Stepsto{o} (\context,
\nvmem', \vmem', \cmd')$,  where  $\context = (\nvmem_0, \vmem_0, \cmd_0)$,  
\item $o$ is not checkpoint or reboot,  
\item  $\nvmem_0; \emptyset; \emptyset\Vdash \cmd_0: \m{ok}$, 
\item $T_0 = (N_1, V_0, \cmd_0) \MSeqStepsto{}  (N, V, \cmd)$, $T_0$ does not
contain any checkpoints, 
\item $\nvmem_0, \nvmem_1, \vmem_0, \cmd_0 \vdash \nvmem\sim \nvmem_2$ 
\end{itemize}
then 
 $\nvmem_0, \nvmem_1, \vmem_0, \cmd_0 \vdash \nvmem' \sim \nvmem_2$. 
\end{lem}
\ifproofs
\begin{proof} 
We case on  $T$. We only show the cases where
 non-volatile memory is updated; as the conclusion trivially hold in
 other cases. 
\begin{description}
\item [Case] $T$ ends in a variable assignment (when $x:=e$ is the
  last instruction, $\cmd_1 = \m{skip}$)
\begin{tabbing}
By assumption
\\\qquad \= (1)~~ \= $
  (\context, \nvmem, \vmem, x:=e; \cmd_1) 
 \Stepsto{[r]}  (\context, \nvmem[x\mapsto v], \vmem, \cmd_1) $
\\\>\> $x \in \m{dom}(N)$, and $N, V \vdash e \Downarrow_{r} v$
\\\>(2)\> $\nvmem_0; \emptyset; \emptyset\Vdash \cmd_0: \m{ok}$
\\ Let $T_c$ be a trace from $\cmd_0$ to the nearest checkpoint
\\\>(3)\>$T_c = (\nvmem_1, \vmem_0, \cmd_0)
  \MSeqStepsto{}  (N', V', \m{checkpoint}(\omega);c')$ 
\\By the semantic rules are deterministic
\\\>(4)\> $T_c = T_0\cdot (N, V,  x:=e; \cmd_1)  
 \SeqStepsto{[r]}  (\context, \nvmem[x\mapsto v], \vmem, \cmd_1)   
 \MSeqStepsto{}  (N', V', \m{checkpoint}(\omega);c')$ 
\\We only need to show that $x\in \mt{FstWt}(T_c)\cup \nvmem_0$. 
\\By Lemma~\ref{lem:wt-fst-version}, $x\in \mt{FstWt}(T_c)\cup \nvmem_0$. 
\end{tabbing}

\item[Case] $T$ ends in array assignment 
\begin{tabbing}
By assumption
\\\qquad \= (1)~~ \= $(\context, \nvmem, \vmem, a[e]:=e'; \cmd_1) 
 \Stepsto{[r,r']}  (\context, \nvmem[a[v]\mapsto v'], \vmem, \cmd_1) $
\\\>\> $N, V \vdash e \Downarrow_{r} v $ and $N, V \vdash e' \Downarrow_{r'} v'$
\\\>(2)\> $\nvmem_0; \emptyset; \emptyset\Vdash \cmd_0: \m{ok}$
\\ Let $T_c$ be a trace from $\cmd_0$ to the nearest checkpoint
\\\>(3)\>$T_c = (\nvmem_1, \vmem_0, \cmd_0)
  \MSeqStepsto{}  (N', V', \m{checkpoint}(\omega);c')$ 
\\By the semantic rules are deterministic
\\\>(4)\> $T_c = T_0\cdot (N, V,  a[e]:=e'; \cmd_1)  
 \SeqStepsto{[r,r']}  (\context, \nvmem[a[v]\mapsto v], \vmem, \cmd_1)   
 \MSeqStepsto{}  (N', V', \m{checkpoint}(\omega);c')$ 
\\We only need to show that $a[v]\in \mt{FstWt}(T_c)\cup \nvmem_0$.
\\We consider two subcases: (I) $a[v]\in \mt{RD}(T_0)\cup \mt{rd}(e) \cup \mt{rd}(e')$ 
and (II) $a[v]\notin \mt{RD}(T_0)\cup \mt{rd}(e) \cup \mt{rd}(e')$ 
\\{\bf Subcase (I)} $a[v]\in \mt{RD}(T_0)\cup \mt{rd}(e)\cup \mt{rd}(e')$
\\By Lemma~\ref{lem:ok-read-same} and (2) and (4)
\\\>(5)\> $\exists$ $W$ and $R$ s.t. $\ee':: \nvmem_0; W; R \Vdash a[e]:=e'; \cmd_1: \m{ok}$,
 and 
 \\\>(6)\>$\mt{RD}(T_0)  \subseteq R$
\\By inversion of $\ee'$ and (6), $a\in \nvmem_0$ 
\\{\bf Subcase (II)} $a[v]\notin \mt{RD}(T_0)\cup \mt{rd}(e)\cup \mt{rd}(e')$
\\By definition of $\mt{FstWt}$, $a[v]\in \mt{FstWt}(T_c)$
\end{tabbing}
\end{description}
\end{proof}
\fi

\begin{lem}[One failure]
\label{lem:one-f}
If
\begin{itemize}
\item $T= (\context, \nvmem_1, \vmem, \cmd) \MStepsto{O} (\context,
\nvmem'_1, \vmem', \cmd')$,  $T$ does not
  execute checkpoint or reboot,
\item $\context = (\nvmem_0, \vmem_0, \cmd_0)$,  and 
\item $T_0 = (N, V_0, \cmd_0) \MSeqStepsto{}  (N_1, V, \cmd)$, $T_0$ does not
contain any checkpoints, 
\item $N, \vmem_0, \cmd_0, \cmd \vdash \nvmem_1 \sim \nvmem_2$ 
\end{itemize}
then 
$ (\nvmem_2, \vmem, \cmd) \MSeqStepsto{O} (\nvmem'_2, \vmem', \cmd')$, 
 and $\nvmem, \vmem_0, \cmd_0, \cmd' \vdash \nvmem'_1 \sim \nvmem'_2$. 
\end{lem}
\begin{proof}
By induction over the length of $T$, apply Lemma~\ref{lem:related-to-related-one-step}.
\end{proof}

\begin{lem}[related NV step to equal NV by checkpoint]
\label{lem:one-f-completion}
  If   $T= (\context, \nvmem_1, \vmem, \cmd) \MStepsto{O} \Sigma'$, 
$\mt{CP}(\Sigma')$
and $T$ does not
  execute checkpoint or reboot,
 $\context = (\nvmem_0, \vmem_0, \cmd_0)$, 
$T_0 = (N, V_0, \cmd_0) \MSeqStepsto{}  (N_1, V, \cmd)$, $T_0$ does not
contain any checkpoints, 
 and $\nvmem, \vmem_0, \cmd_0, \cmd\vdash \nvmem_1 \sim \nvmem_2$, 
 then
  $ (\nvmem_2, \vmem, \cmd) \MSeqStepsto{O} \erase{\Sigma'}$.
\end{lem}
\ifproofs
\begin{proof}
By induction over the length of $T$.
\begin{description}
\item[Base case:] $|T| = 0$
\begin{tabbing}
By assumption
\\\qquad \= (1)~~ \= $\cmd= \m{checkpoint}(\omega);\cmd'$
\\ \>(2)\> $\forall \loc \in N_1$ s.t. $N_1(\loc) \neq N_2(\loc)$,  
$\loc \in \mt{Wt}(\emptyset)\cap\mt{FstWt}(T_0)$ and $\loc\notin \mt{WT}(T_0)$
\\By (2), there is no $\loc$ 
\\\>(3)\> $N_1 = N_2$
\\ The continuous powered execution also takes $0$ steps and the
conclusion holds
\end{tabbing}

\item[Inductive case:]
~
\begin{tabbing}
By assumption
\\\qquad \= (1)~~ \= 
$T= (\context, \nvmem_1, \vmem, \cmd) \Stepsto{o} (\context,
  \nvmem'_1, \vmem_1, \cmd_1) \MStepsto{O} (\context,
  \nvmem', \vmem', \m{checkpoint}(\omega);\cmd')$, 
\\By Lemma~\ref{lem:related-to-related-one-step}
\\\>(2)\>  $(\nvmem_2, \vmem, \cmd) \SeqStepsto{o} (\nvmem'_2,
\vmem_1, \cmd_1) $ and $N, \vmem_0, \cmd_0, \cmd_1 \vdash \nvmem'_1 \sim \nvmem'_2$.
\\By I.H. on the tail of $T$ 
\\\>(3)\>$(\nvmem'_2,\vmem_1, \cmd_1) \MSeqStepsto{O} (\nvmem',
\vmem', \m{checkpoint}(\omega);\cmd')$
\\By (2) and (3)
\\\> The conclusion holds
\end{tabbing}
\end{description}
\end{proof}
\fi

\begin{lem}[related NV step to related NV]
\label{lem:related-to-related-one-step}
  If   $T= (\context, \nvmem_1, \vmem, \cmd) \Stepsto{o} (\context,
  \nvmem'_1, \vmem', \cmd')$, and $T$ does not
  execute checkpoint or reboot,
 $\context = (\nvmem_0, \vmem_0, \cmd_0)$, 
$T_0 = (N, \vmem_0, \cmd_0) \MSeqStepsto{}  (N_1, \vmem, \cmd)$, $T_0$ does not
contain any checkpoints, 
 and $\nvmem, \vmem_0, \cmd_0, \cmd\vdash \nvmem_1 \sim \nvmem_2$, 
 then
  $ (\nvmem_2, \vmem, \cmd) \SeqStepsto{o} (\nvmem'_2, \vmem', \cmd')$
and $\nvmem, \vmem_0, \cmd_0, \cmd' \vdash \nvmem'_1 \sim \nvmem'_2$.
\end{lem}
\ifproofs
\begin{proof} 
By examining the structure of $T$. 
\begin{description}
\item[Case:] $T$ ends in \rulename{CP-Skip} rule. 
\begin{tabbing}
By assumption
\\\qquad \= (1)~~ \= 
$T= (\context, \nvmem_1, \vmem, \m{skip};\cmd) 
 \Stepsto{}  (\context, \nvmem_1, \vmem, \cmd) $ and
\\\>(2)\> $\context = (\nvmem_0, \vmem_0, \cmd_0)$,
  and $\vmem_0, \cmd_0, (\m{skip};\cmd)\vdash \nvmem_1 \sim \nvmem_2$
\\ By  \rulename{Skip} rule
\\\>(3)\> $ (\nvmem_2, \vmem, \m{skip};\cmd)  \SeqStepsto{}
(\nvmem_2, \vmem, \cmd)$
\\By (2) and $\m{skip}$ does not write to store
\\\>(4)\>  $N, \vmem_0, \cmd_0, \cmd\vdash \nvmem_1 \sim \nvmem_2$
\end{tabbing}

\item[Case:] $T$ ends in \rulename{CP-Seq} rule and uses
  \rulename{CP-NV-Assign} or ends in \rulename{CP-NV-Assign}
  rule. The proof for both of these two cases are the same except that
  in the latter, the resulting command is $\m{skip}$.
\begin{tabbing}
By assumption
\\\qquad \= (1)~~ \= 
$T= (\context, \nvmem_1, \vmem, x:=e;\cmd) 
 \Stepsto{[r_1]}  (\context, \nvmem_1[x \mapsto v_1], \vmem, \cmd) $ and
\\\>(2)\> $\nvmem_1; \vmem \vdash e \Downarrow_{r_1} v_1$ and 
\\\>(3)\> $\context = (\nvmem_0, \vmem_0, \cmd_0)$ and
  $\vmem_0, \cmd_0, (x:=e;\cmd)\vdash \nvmem_1 \sim \nvmem_2$
\\ By  \rulename{CP-NV-Assign} rule
\\\>(4)\> $ (\nvmem_2, \vmem, x:=e;\cmd)  \SeqStepsto{[r_2]}
(\nvmem_2[x \mapsto v_2], \vmem, \cmd)$ and
\\\>(5)\> $\nvmem_2; \vmem \vdash e \Downarrow_{r_2} v_2$
\\ By Lemma~\ref{lem:related-n-eq-var} and Lemma~\ref{lem:eq-var-eq-eval}
\\\>(6)\> $v_1 = v_2$ and $r_1 = r_2$
\\By (3) and (6)
\\\>(7)\>  $N, \vmem_0, \cmd_0, \cmd \vdash \nvmem_1[x \mapsto v_1] \sim \nvmem_2[x \mapsto v_2]$
\end{tabbing}

\item[Case:] $T$ ends in  \rulename{CP-Seq} rule and uses
  \rulename{CP-Arr-Assign} rule or ends in \rulename{CP-Arr-Assign} rule. 
\begin{tabbing}
By assumption
\\\qquad \= (1)~~ \= 
$T= (\context, \nvmem_1, \vmem, a[e']:=e;\cmd) 
 \Stepsto{[r_1',r_1]}  (\context, \nvmem_1[ (a[v_1']) \mapsto v_1], \vmem, \cmd) $ and
\\\>(2)\> $\nvmem_1; \vmem \vdash e \Downarrow_{r_1} v_1$ 
and $\nvmem_1; \vmem \vdash e' \Downarrow_{r_1'} v_1'$ 
\\\>(3)\> $\context = (\nvmem_0, \vmem_0, \cmd_0)$ and
  $\vmem_0, \cmd_0, (a[e']:=e;\cmd)\vdash \nvmem_1 \sim \nvmem_2$
\\ By  \rulename{CP-Ary-Assign} rule
\\\>(4)\> $ (\nvmem_2, \vmem, a[e']:=e;\cmd)  \SeqStepsto{[r_2',r_2]}
(\nvmem_2[(a[v_{2}']) \mapsto v_2], \vmem, \cmd)$ and
\\\>(5)\> $\nvmem_2; \vmem \vdash e' \Downarrow_{r_2'} v_{2}'$ and$\nvmem_2; \vmem \vdash e \Downarrow_{r_2} v_2$
\\ By Lemma~\ref{lem:related-n-eq-var} and Lemma~\ref{lem:eq-var-eq-eval}
\\\>(6)\> $v_1 = v_2$ and $v_1' = v_2'$ and $r_1 = r_2$ and $r_1' = r_2'$
\\By (3) and (6)
\\\>(7)\>  $N, \vmem_0, \cmd_0, \cmd \vdash \nvmem_1[a[v_1'] \mapsto v_1] \sim \nvmem_2[a[v_2'] \mapsto v_2]$
\end{tabbing}

\item[Case:] $T$ ends in  \rulename{CP-Seq} rule and uses
  \rulename{CP-V-Assign} rule or ends in \rulename{CP-V-Assign} rule.  
\begin{tabbing}
By assumption 
\\\qquad \= (1)~~ \= 
$T= (\context, \nvmem_1, \vmem, x:=e;\cmd) 
 \Stepsto{[r_1]}  (\context, \nvmem_1, \vmem[x\mapsto v_1], \cmd) $
 and $\nvmem_1; \vmem \vdash e \Downarrow_{r_1} v_1$
\\\>(2)\> $\context = (\nvmem_0, \vmem_0, \cmd_0)$,
and  $\vmem_0, \cmd_0, (x:=e;\cmd)\vdash \nvmem_1 \sim \nvmem_2$
\\ By  \rulename{CP-V-Assign} rule
\\\>(3)\> $ (\nvmem_2, \vmem, x:=e;\cmd)  \SeqStepsto{[r_2]}
(\nvmem_2, \vmem[x\mapsto v_2] , \cmd)$ and  $\nvmem_2; \vmem \vdash e
\Downarrow_{r_2} v_2$
\\ By Lemma~\ref{lem:related-n-eq-var} and Lemma~\ref{lem:eq-var-eq-eval}
\\\>(4)\> $v_1 = v_2$ and $r_1 = r_2$
\\By (4), $\vmem[x\mapsto v_1] =\vmem[x\mapsto v_2]$ 
\\By (2) and \rulename{CP-V-Assign} does not write to non-volatile store
\\\>(4)\>  $N, \vmem_0, \cmd_0, \cmd \vdash \nvmem_1 \sim \nvmem_2$
\end{tabbing}

\item[Case:] $T$ ends in \rulename{CP-If-T} rule. 
\begin{tabbing}
By assumption
\\\qquad \= (1)~~ \= 
$T= (\context, \nvmem_1, \vmem, \m{if}~e~\m{then}~c_1~\m{else}~c_2) 
 \Stepsto{[r_1]}  (\context, \nvmem_1, \vmem, \cmd_1) $ and $\nvmem_1,
 \vmem  \vdash e \Downarrow_{r_1} \m{true}$ and
\\\>(2)\> $\context = (\nvmem_0, \vmem_0, \cmd_0)$, and 
  $\vmem_0, \cmd_0, \m{if}~e~\m{then}~c_1~\m{else}~c_2\vdash \nvmem_1 \sim \nvmem_2$
\\Let's assume
\\\> (3)\>  $\nvmem_2, \vmem  \vdash e \Downarrow_{r_2} v$
\\By Lemma~\ref{lem:related-n-eq-var} and Lemma~\ref{lem:eq-var-eq-eval}
\\\>(4)\> $v = \m{true}$ and $r_1 = r_2$
\\ By  \rulename{CP-If-T} rule and (4)
\\\>(5)\> $ (\nvmem_2, \vmem, \m{if}~e~\m{then}~c_1~\m{else}~c_2)  \SeqStepsto{[r_2]}
(\nvmem_2, \vmem, \cmd_1)$
\\ By (5) and \rulename{CP-If-T} doesn't write to storage
\\\>(6)\>  $N, \vmem_0, \cmd_0, \cmd_1\vdash \nvmem_1 \sim \nvmem_2$
\end{tabbing}

\item[Case:] $T$ ends in   \rulename{CP-If-F} rule. Similar to the
  previous case. 
\end{description}
\end{proof}
\fi

\begin{lem} \label{lem:related-n-eq-var}
If 
$T= (\nvmem, \vmem_0, \cmd_0) \MSeqStepsto{} (\nvmem_1, \vmem', \iota;\cmd)$,  $T$ does not
  execute checkpoint,
$\nvmem, \vmem_0, \cmd_0, (\iota;\cmd)\vdash \nvmem_1 \sim \nvmem_2$ where
$\iota\neq \m{checkpoint}$, 
and $\loc\in \mt{rd}(\iota)\cap \nvmem_1$,
then $\nvmem_1(\loc) = \nvmem_2(\loc)$.
\end{lem}
\ifproofs
\begin{proof}
We assume that  $\nvmem_1(\loc) \neq \nvmem_2(\loc)$ then derive a
contradiction.
\begin{tabbing}
By assumption, 
\\\qquad \= (1)~~ \= $\nvmem_1(\loc) \neq \nvmem_2(\loc)$  and
$\vmem_0, \cmd_0, (\iota;\cmd)\vdash \nvmem_1 \sim \nvmem_2$.
\\ By definition~\ref{def:mem-same_exec}
\\\>(2)\> 
$T_2 = (N_1, V', \iota;\cmd) \MSeqStepsto{}  (N'', V'', \m{checkpoint}(\omega);
c'')$, where $T_2$ do not contain checkpoints, 
\\\>\>$\forall x \in N_1$ s.t. $N_1(x) \neq N_2(x)$,  
$x \in \mt{Wt}(T_2)\cap\mt{FstWt}(T\cdot T_2)$ and $x\notin \mt{WT}(T)$.
\\By (1) and (2)
\\\>(3)\> $\loc \in \mt{Wt}(T_2)\cap\mt{FstWt}(T\cdot T_2)$ and $\loc
\notin \mt{WT}(T)$.
\\By $\loc\in \mt{rd}(\iota)$, $\loc \notin \mt{WT}(T)$, and the definition of $\mt{FstWt}$
\\\>(4)\> $\loc\notin \mt{FstWt}(T\cdot T_2)$
\\(3) and (4) derive a contradiction, so the conclusion holds.
\end{tabbing}
\end{proof}
\fi

\begin{lem} \label{lem:eq-var-eq-eval}
If $\forall \loc\in \mt{rd}(e)$, $\nvmem_1\cup\vmem_1(\loc) =  \nvmem_2\cup
\vmem_2(\loc)$, and $\nvmem_1,\vmem_1 \vdash e \Downarrow_{r_1} v_1$ 
$\nvmem_2, \vmem_2 \vdash e \Downarrow_{r_2} v_2$ 
then $r_1 = r_2$ and $v_1 = v_2$.
\end{lem}
\begin{proof}
By induction over the structure of $e$.
\end{proof}

\section{Checkpointing Variants}
\label{sec:cp-variants}

\subsection{Idempotent Regions}
\label{sec:variants-ratchet}

An alternative to logging is to place checkpoints  
so that all operations in an inter-checkpoint region safely re-execute. One
strategy previously used for intermittent system Ratchet is to break
WAR dependences by inserting a checkpoint before the WAR's write. 
The idea is to rewrite a program without checkpoints into one with
a checkpoint instruction that has an empty $\omega$ (i.e., $\m{checkpoint}(\,)$) at a subset of program locations. 
We define the rewriting rules of the form: $W;R \Vdash \cmd
\SeqStepsto{}: \cmd' $ , where $W$ and $R$ is the variables written
and respectively read from the beginning of the program to $\cmd$. 
$\cmd$ is rewritten into $\cmd'$. 

~\\
\noindent\framebox{$W;R \Vdash \cmd \longrightarrow \cmd' $}

\begin{mathpar}
\inferrule*[right=R-WAR-Written]{ R' = R \cup rd(e) \\ x \notin R'
\\   W \cup x;R' \Vdash\cmd \SeqStepsto{} \cmd'   }{
  W;R \Vdash x:= e;\cmd \SeqStepsto{}  x:= e;\cmd'
}
\and
\inferrule*[right=R-WAR-Cp-Simple]{ 
\\ x \in  R
 \\ x \notin W \\x \notin \mt{rd}(e)
\\x;  \mt{rd}(e)\Vdash\cmd \SeqStepsto{} \cmd' }{
  W;R \Vdash x:= e;\cmd  \SeqStepsto{} \m{checkpoint}();x:=e;\cmd'
}
\and
\inferrule*[right=R-WAR-Cp-Split]{ R' = R \cup rd(e)
\\ x \in R \\x \notin W \\x \in rd(e) \\\mt{fresh}(x')
\\x; x' \Vdash\cmd \SeqStepsto{} \cmd' }{
  W;R \Vdash x:= e;\cmd  \SeqStepsto{} x' :=e;\m{checkpoint}();x:=x';\cmd' 
}
\and
\inferrule*[right=R-WAR-WriteDom]{ R' = R \cup rd(e) 
\\ x \in R' \\ x \in W
\\   W ;R' \Vdash\cmd \SeqStepsto{} \cmd'   }{
  W\cup\{x\};R \Vdash x:= e;\cmd \SeqStepsto{} x:=e; \cmd'
}
\and
\inferrule*[right=R-WAR-Written-Arr]{ R' = R \cup rd(e) \cup rd(e') 
 \\ a \notin R'
\\ W \cup a; R' \Vdash\cmd \SeqStepsto{} \cmd' 
 }{
  W;R \Vdash a[e'] := e;\cmd \SeqStepsto{} a[e'] := e; \cmd'
}
\and
\inferrule*[right=R-WAR-Cp-Arr-Simple]{ 
 R' =  rd(e) \cup rd(e') 
 \\ a \in R \\a \notin R'
\\a; R' \Vdash\cmd \SeqStepsto{} \cmd' }{
  W;R \Vdash a[e']:= e;\cmd \SeqStepsto{}  \m{checkpoint}();a[e']:= e;\cmd' 
}
\end{mathpar}
\begin{mathpar}
\inferrule*[right=R-WAR-Cp-Arr-Split]{ 
R' = rd(e)  \cup rd(e') 
 \\ a \in R' \\ \mt{fresh}(x,y)
\\ a; \{x,y\}  \Vdash\cmd \SeqStepsto{} \cmd'  }{
  W;R \Vdash a[e']:= e;\cmd \SeqStepsto{}  x:=e; y:= e'; \m{checkpoint}();
  a[y]:= x; \cmd' 
}
\and
\inferrule*[right=WAR-If]{ R' = R \cup rd(e) \\ W;R' \Vdash \cmd_1
  \SeqStepsto{} \cmd_1' 
 \\ W;R' \Vdash \cmd_2 \SeqStepsto{} \cmd_2' }{
  W;R \Vdash \m{if}\ e\ \m{then}\ \cmd_1\ \m{else}\ \cmd_2
  \SeqStepsto{} 
 \m{if}\ e\ \m{then}\ \cmd_1'\ \m{else}\ \cmd_2' 
}

\end{mathpar}

\begin{lem}
If $W; R \Vdash c \SeqStepsto{ } c'$ then $\emptyset; W; R \Vdash_\war c': \m{ok}$
\end{lem}
\begin{proofsketch}
By induction over the structure of the rewriting derivation. 
\end{proofsketch}

\subsection{JIT Checkpointing}
\label{sec:variants-jit}
Just-In-Time (JIT) checkpointing system relies on the runtime to
detect low power and then the system checkpoints the volatile memory
and program to be executed on the spot.  We model such a system based
on Hibernas. The tricky part is to deal with situations where the
checkpointing action fails due to powerfailure. Hibernas implements a
flag, indicating whether the checkpoint has succeeded. We model it
using the checkpoint context $\kappa$ as follows.
\[
\begin{array}{llcl}
\textit{context} & \context & \bnfdef & \m{fail} \bnfalt \m{success}(\vmem, \cmd)
\end{array}
\]
The context can either be $\m{fail}$, which means that the checkpoint
did not complete; or $\m{success}(\vmem,\cmd)$, meaning the checkpoint
is successful and the volatile memory and the command at the
checkpoint is $\vmem$ and $\cmd$ respectively.

The powerfail, reboot, and checkpoint rules are defined below. 
~\\\\
\noindent\framebox{$(\timestamp, \context, \nvmem, \vmem, \cmd) 
 \stepsto (\timestamp', \context', \nvmem', \vmem', \cmd')$}

\begin{mathpar}

\inferrule*[right=JIT-LowPower]{ \mt{PowerLow}}{
  (\timestamp, \context, \nvmem, \vmem, \cmd) 
 \stepsto  (\timestamp+1, \context, \nvmem, \vmem, \m{checkpoint}();\cmd) 
}

\and
\inferrule*[right=JIT-CP-Success]{ \m{pick}(n)}{
  (\timestamp, \context, \nvmem, \vmem, \m{checkpoint}();\cmd) 
 \stepsto  (\timestamp+1, \m{success}(\vmem, \cmd), \nvmem, \vmem, \m{reboot}(n)) 
}
\and
\inferrule*[right=JIT-CP-Fail]{ \m{pick}(n)}{
  (\timestamp, \context, \nvmem, \vmem, \m{checkpoint}();\cmd) 
 \stepsto  (\timestamp+1, \m{fail}, \nvmem, \vmem, \m{reboot}(n)) 
}
\and
\inferrule*[right=JIT-Restore]{ }{
  (\timestamp, \m{success}(\vmem, \cmd), \nvmem, \vmem', \m{reboot}(n)) 
 \stepsto  ( \timestamp+n, \m{success}(\vmem, \cmd), \nvmem, \vmem, \cmd) 
}
\and
\inferrule*[right=JIT-Restart]{ }{
  (\timestamp, \m{fail}, \nvmem, \vmem', \m{reboot}(n)) 
 \stepsto  (\timestamp+n, \m{fail}, \nvmem_0, \resetm(\vmem), \cmd_0) 
}
\end{mathpar}
When power is low (\rulename{JIT-LowPower}), a checkpoint is
inserted. Rule \rulename{JIT-CP-Success} applies when the checkpoint
instruction completes. Rule \rulename{JIT-CP-Fail} applies when the checkpoint
instruction fails due to power failure. Upon reboot, if the checkpoint
was a success, the volatile memory and command are restored
(\rulename{JIT-Restore}); otherwise, the system starts from its
initial state. 

The correctness proof of JIT is much simpler as the JIT checkpoints
allow the system to continue at the exact point of power failure (no
roll back at all). The only somewhat interesting case is the
restart. The correctness requires the system start from the initial
state. For embedded systems, this requires the initial memory $N_0$ to
be rewritten, which was not mentioned in the Hibernas. 

\subsection{Undo Logging}
\label{sec:variants-undo}

An undo logging system stores a potentially inconsistent location's value into a log immediately before
the first write to that location. 
Upon reboot after a power failure, the
log is applied to the non-volatile memory, and thus the system rolls back
(undoes) the effects of the failed intermittent execution. 
Our basic checkpoint-based model (i.e., DINO) is a conservative, static
form of undo logging, which stores all initial values to the context
$\context$ at checkpoint time, rather than on-demand when a write
is encountered. 

\paragraphb{Syntax and Operational Semantics}
The undo logging context, denoted $\ulctx$, additionally includes a
log $\ulog$, residing in non-volatile memory, a list of
already-logged locations $\loggedl$, and the list of locations that will need
to be logged $\omega$. 
\[
\begin{array}{llcl}
\textit{log} & \ulog & : & M
\\
\textit{logged locations} & \loggedl & \bnfdef & \cdot\bnfalt
                                                 \loggedl, \loc
\\
\textit{context} & \ulctx & \bnfdef & (\ulog, \vmem, \cmd, \omega, \loggedl)
\end{array}
\]
We summarize the semantics for undo-logging below. 
~\\\\
\noindent\framebox{$(\ulctx, \nvmem, \vmem, \cmd) 
 \ulStepsto{} (\ulctx', \nvmem', \vmem', \cmd')$}
\begin{mathpar}
\inferrule*[right=UL-PowerFail]{ }{
  (\ulctx, \nvmem, \vmem, \cmd) 
 \ulStepsto{}  (\ulctx, \nvmem, \resetm(\vmem), \m{reboot}()) 
}
\and
\inferrule*[right=UL-Commit]{ }{
  ((\ulog', \vmem', \cmd', \omega', \loggedl), \nvmem, \vmem, \m{checkpoint}(\omega);\cmd) 
  \ulStepsto{\m{checkpoint}}  ((\emptyset, \vmem, \cmd, \omega, \emptyset), \nvmem, \vmem, \cmd) 
}
\and
\inferrule*[right=UL-Reboot]{ }{
  ((\ulog, \vmem, \cmd, \omega, \loggedl), \nvmem, \vmem', \m{reboot}()) 
  \ulStepsto{\m{reboot}}  ((\emptyset, \vmem, \cmd, \omega, \emptyset), \nvmem\lhd\ulog, \vmem, \cmd) 
}
\and
\inferrule*[right=UL-NV-Log]{ x\in\omega \\ x\notin \loggedl 
\\\nvmem, \vmem \vdash e \Downarrow_r v}{
  ((\ulog, \vmem_c, \cmd_c, \omega, \loggedl), \nvmem, \vmem, x := e) 
 \\\\\ulStepsto{[r]}  ((\ulog[x\mapsto\nvmem(x)], \vmem_c, \cmd_c, \omega, \loggedl \cup \{x\}), 
  \nvmem[x\mapsto v], \vmem, \m{skip}) 
}

\and
\inferrule*[right=UL-NV-Assign]{ x \in \m{dom}(N) \\ (x\notin\omega \lor x \in \loggedl)
\\\nvmem, \vmem \vdash e \Downarrow_r v}{
  (\ulctx, \nvmem, \vmem, x := e) 
  \ulStepsto{[r]}  (\ulctx, 
  \nvmem[x\mapsto v], \vmem, \m{skip}) 
}
\and
\inferrule*[right=UL-V-Assign]{ x \in \m{dom}(V) 
\\\nvmem, \vmem \vdash e \Downarrow_r v}{
  (\ulctx, \nvmem, \vmem, x := e) 
  \ulStepsto{[r]}  (\ulctx, 
  \nvmem, \vmem[x\mapsto v], \m{skip}) 
}
\and
\inferrule*[right=UL-Arr-Log]{ a \in\omega\\ a[v] \notin \loggedl
\\\nvmem, \vmem \vdash e \Downarrow_r v \\\nvmem, \vmem \vdash e' \Downarrow_{r'} v'}{
  ((\ulog, \vmem_c, \cmd_c, \omega, \loggedl), \nvmem, \vmem, a[e] := e') 
  \ulStepsto{[r,r']}  
\\((\ulog[a[v] \mapsto\nvmem(a[v])], \vmem_c, \cmd_c, \omega, \loggedl \cup \{a[v]\}), 
  \nvmem[a[v]\mapsto v'], \vmem, \m{skip}) 
}

\and
\inferrule*[right=UL-Arr-Assign]{ (a\notin\omega \lor a[v] \in \loggedl)
\\\nvmem, \vmem \vdash e \Downarrow_r v \\ \nvmem, \vmem \vdash e' \Downarrow_{r'} v'}{
  (\ulctx, \nvmem, \vmem, a[e] := e') 
  \ulStepsto{[r,r']}  (\ulctx, 
  \nvmem[a[v]\mapsto v'], \vmem, \m{skip}) 
}
\end{mathpar}
\begin{mathpar}
\and
\inferrule*[right=UL-Seq]{  (\ulctx, \nvmem, \vmem, \iota) 
 \ulStepsto{o}  (\ulctx, \nvmem', \vmem', \m{skip})}{
  (\ulctx, \nvmem, \vmem, \iota;c) 
 \ulStepsto{o}  (\ulctx, \nvmem', \vmem', c) 
}
\and
\inferrule*[right=UL-If-T]{N, V \vdash e \Downarrow_{r} \m{true}}{
  (\ulctx, \nvmem, \vmem, \m{if}\ e\ \m{then}\ \cmd_1\ \m{else}\ \cmd_2) 
 \ulStepsto{[r]}  (\ulctx, \nvmem, \vmem, \cmd_1) 
}
\and
\inferrule*[right=UL-If-F]{N, V \vdash e \Downarrow_{r} \m{false}}{
  (\ulctx, \nvmem, \vmem, \m{if}\ e\ \m{then}\ \cmd_1\ \m{else}\ \cmd_2) 
 \ulStepsto{[r]}  (\ulctx, \nvmem, \vmem, \cmd_2) 
}
\end{mathpar}

The \rulename{UL-Checkpoint} rule 
is analogous to the \rulename{CP-Checkpoint} rule.  The rule 
updates $\omega$ to $\omega'$ in the context, and clears
$\ulog$ and $\loggedl$.  The execution fills the log with values to roll back
on reboot.  
The \rulename{UL-NV-Log} rule applies to the first assignment to a
variable in $\omega$ since the last checkpoint. The current value of
the assigned location is written to the log, the location is added to
the logged list, and non-volatile memory is updated with the new
value. 
On rebooting, the \rulename{UL-Reboot} rule applies
the log to non-volatile memory, reverting changes to a subset of variables in $\omega$.
The log and logged list are again reset, and the program starts to re-execute
the command stored in the context.

In contrast, the \rulename{UL-NV-Assign} rule simply updates
non-volatile memory, applying to subsequent (i.e., non-first)
assignments of variables in $\omega$ and to assignments of variables
not in $\omega$.
Variables already in the 
log don't need any extra tracking because on reboot their values
will correctly revert to whatever they were at the last checkpoint.

Assignments to arrays follow the same pattern as assignments to
non-volatile memory.  
Unlike the basic checkpointing system, which adds the entire array to
$\omega$ if any array element could be involved in WAR dependence, undo
logging can more precisely handle arrays by adding a single array element to
the log dynamically {\em at the write} operation.
avoiding checkpointing the entire array at potentially high checkpointing time cost.  

\paragraphb{Equivalence to basic checkpoint system}
We define a binary relation between basic checkpoint and undo-logging
state, written $\dstate \rel \ustate$, as follows. We then prove it's
a bisimulation relation.
\[
\inferrule*{
\dstate = (\context, \nvmem, \vmem, \cmd) 
\\ \ustate = (\ulctx, \nvmem_u, \vmem_u, \cmd_u)
\\ \dcon = (\nvmem_c, \vmem_c, \cmd_c)
\\ \ucon = (\ulog, \vmem_c, \cmd_c, \omega, \loggedl)
\\ \cmd= \cmd_u
\\ \nvmem = \nvmem_u
\\ \vmem = \vmem_u
\\ \omega = \m{dom}(N_c)
\\ \loggedl = \m{dom}(\ulog)
\\ \forall \loc\in \loggedl, \ulog(\loc) = \nvmem_c(\loc)
\\ \forall \loc\in \omega\backslash\loggedl, \nvmem_u(\loc) = \nvmem_c(\loc)
}{
\dstate \rel \ustate
}
\]
The key idea is that the log $\ulog$ is a subset of the checkpointed
data $N_c$, and, for any location in $\omega$ that is not yet logged,
its checkpointed value ($N_c(\loc)$) is the same as its value in
non-volatile memory of the undo log system ($N_u(\loc)$).  The last
property holds because the location is unwritten if it is not yet
logged.

The equivalence proofs need to show that
after a reboot, the non-volatile memories 
are equal, even though a reboot on the basic system results
$\nvmem \lhd \nvmem_c$ and the undo-logging system results in
$\nvmem_u \lhd \ulog$.  
The log is always a subset of the checkpointed
memory so all locations in $\nvmem_c \cap \ulog$ will match after the
application.  Any checkpointed locations not in the log have not been
updated yet, so their value in $\nvmem_u$ matches their value in the
checkpointed set, and so the two memories are equal after a reboot.

\begin{lem}[Undo logging simulates basic check point system]\label{lem:sim-undo-dino}
$\dstate \rel \ustate$ and $\dstate \Stepsto{o_1} \dstate' $
then $\exists \ustate'$ s.t. $\ustate \ulStepsto{o_2}\ustate'$ and
$\dstate' \rel \ustate'$ and $o_1 = o_2$.
\end{lem}

\begin{proof}

By induction over $\ee::\dstate \Stepsto{o_1} \dstate' $

\begin{description}
\item[Cases:] $\ee$ ends in \rulename{CP-Skip}, \rulename{CP-V-Assign}, \rulename{CP-If-T}
  \rulename{CP-If-F} rule. The non-volatile memory and checkpointed
  data is not altered. The undo logging system
  takes corresponding \rulename{UL-Skip}, \rulename{UL-V-Assign}, \rulename{UL-If-T}
  \rulename{UL-If-F} to reach a related configuration.

\item[Case:] $\ee$ ends in \rulename{CP-Seq} rule. We apply
  I.H. directly. 

\item[Case:] $\ee$ ends in \rulename{CP-NV-Assign} rule  
\begin{tabbing}
By assumption
\\\qquad \= (1)\quad \= $\dstate \rel \ustate$
\\ By  \rulename{CP-NV-Assign} rule
\\\>(2)\> $((\nvmem_c, \vmem_c, \cmd_c), \nvmem, \vmem, x:=e) 
\Stepsto{[r]}  ((\nvmem_c, \vmem_c, \cmd_c), \nvmem[x\mapsto v], \vmem, \m{skip}) $
\\\> (3)\> $\nvmem, \vmem \vdash e \Downarrow_r v$
\\We consider two subcases: (I) $x \notin \omega$ or $x \in \ulog$ 
 and (II) $x \in \omega$ and $x \notin \ulog$ 
\\ {\bf Subcase (I)} $x \notin \omega$ or $x\in\ulog$
\\ By  \rulename{UL-Assign} rule
\\\>(I1)\> $(\ucon, \nvmem_u, \vmem_u, x:=e) 
\ulStepsto{[r_u]}  (\ucon, \nvmem_u[x\mapsto v_u], \vmem_u,  \m{skip}) $
\\\> (I2)\> $\nvmem_u, \vmem_u \vdash e \Downarrow_{r'} v_u$
\\By assumption
\\\> (I3)\> $\nvmem_u = \nvmem$ and $\vmem_u = \vmem$
\\ By expression evaluation is deterministic
\\\>(I4)\> $v= v_u$, $r=r_u$, 
\\ By (I3) and (I4)
\\\> (I5)\>  $\nvmem[x\mapsto v] = \nvmem_u[x\mapsto v_u]$
\\ By (I5) and Volatile memory and log do not change, $\dstate' \rel \ustate'$
\\ {\bf Subcase (II)} $x \in \omega$ and $x \notin \ulog$
\\ By  \rulename{UL-NV-Log} rule
\\\>(II1)\> $((\ulog, \vmem_c,\cmd_c, \omega, \ulog), \nvmem_u, \vmem_u, x:=e) $
\\\>\> $\ulStepsto{[r_u]}  ((\ulog[x \mapsto \nvmem_u(x)], \vmem_c,\cmd_c, \omega, \ulog \cup x), 
\nvmem_u[x\mapsto v_u], \vmem_u,  \m{skip}) $
\\\>(II2)\> $\nvmem_u, \vmem_u \vdash e \Downarrow_{r_u} v_u$
\\ By assumption that $\dstate \rel \ustate$
\\\>(II3)\>$\nvmem = \nvmem_u$ and $\vmem_u = \vmem$
\\By (II3)
\\\>(II4)\> $v= v_u$, $r=r_u$, 
\\By (II3) and (II4)
\\\> (II5)\> $\nvmem[x\mapsto v] = \nvmem_u[x\mapsto v_u]$
\\By assumption that $\dstate \rel \ustate$
\\\>(II6)\> $\nvmem_u(x) = \nvmem_c(x)$
\\By (II6)
\\\>(II7)\> $(\ulog[x \mapsto \nvmem_u(x)])(x) = \nvmem_c(x)$
\\ By (II5) and (II) and other parts of the log and volatile memories do not change, $\dstate' \rel \ustate'$
\end{tabbing}

\item[Case:] $\ee$ ends in \rulename{CP-Assign-Arr} rule
\begin{tabbing}
By assumption
\\\qquad \= (1)\quad \= $\dstate \rel \ustate$
\\ By  \rulename{CP-Assign-Arr} rule
\\\>(2)\> $((\nvmem_c, \vmem_c, \cmd_c), \nvmem, \vmem, a[e]:=e') 
    \Stepsto{[r, r']}  ((\nvmem_c, \vmem_c, \cmd_c), \nvmem[a[v]\mapsto v'], \vmem,  \m{skip}) $
\\\> (3)\> $\nvmem, \vmem \vdash e \Downarrow_r v$
\\\> (4)\> $\nvmem, \vmem \vdash e' \Downarrow_{r'} v'$
\\We consider two subcases: and (I) $a \notin \omega$ or $a[v] \in \loggedl$ or (II) $a \in \omega$  
    
\\ {\bf Subcase (I)} $a \notin \omega$ or $a[v] \in \loggedl$
\\ By  \rulename{UL-Arr-Assign} rule
\\\>(I1)\> $(\ucon, \nvmem_u, \vmem_u, a[e]:=e') 
    \ulStepsto{[r_u,r_u']}  (\ucon, \nvmem_u[a[v_u]\mapsto v_u'], \vmem_u,  \m{skip}) $
\\\> (I2)\> $\nvmem_u, \vmem_u \vdash e \Downarrow_{r_u} v_u$
\\\> (I3)\> $\nvmem_u, \vmem_u \vdash e' \Downarrow_{r_u'} v_u'$
\\By assumption that $\dstate \rel \ustate$
\\\> (I4)\> $\nvmem_u = \nvmem$ and $\vmem_u = \vmem$
\\ By expression evaluation is deterministic
\\\>(I5)\> $v_u = v$ and $v_u' = v'$ and $r_u = r$ and $r_u' = r'$
\\By (I4) and (I5)
\\\> (I6)\> $\nvmem[a[v]\mapsto v'] = \nvmem_u[a[v_u]\mapsto v_u']$
\\ By (I6) and volatile memories and log do not change, $\dstate' \rel \ustate'$

\\ {\bf Subcase (II)} $a \in \omega$ and $a[v]\notin \loggedl$
\\\>(II1)\>\ $\nvmem_u, \vmem_u \vdash e \Downarrow_{r_u} v_u$
\\\>(II2)\>\ $\nvmem_u, \vmem_u \vdash e' \Downarrow_{r_u'} v_u'$
\\ By assumption that $\dstate \rel \ustate$, 
\\\>(II3)\>\ $\nvmem = \nvmem_u$ and $\vmem = \vmem_u$
\\By expression evaluation is deterministic
\\\> (II4)\>\ $v_u = v$ and $v_u' = v'$ and $r_u = r$ and $r_u' = r'$
\\ By  \rulename{UL-Arr-Log} rule
\\\>(II5)\>\ $((\ulog, \vmem_c,\cmd_c, \omega, \loggedl), \nvmem_u, \vmem_u, a[e]:=e') $
\\\>\>
    $\ulStepsto{[r_u,r_u']}  ((\ulog[a[v_u] \mapsto \nvmem_u(a[v_u])], \vmem_c,\cmd_c, \omega, \loggedl \cup a[v_u]), 
    \nvmem_u[a[v_u]\mapsto v_u'], \vmem_u, \m{skip}) $
\\By (II3) and (II4) 
\\\> (II6)\> $\nvmem[a[v]\mapsto v'] = \nvmem_u[a[v_u]\mapsto v_u']$
\\By assumption that $\dstate \rel \ustate$
\\\>(II7)\> $\nvmem_u(a[v_u]) = \nvmem_c(a[v_u])$ 
\\By (II7)
\\\>(II8) $\ulog[a[v_u] \mapsto \nvmem_u(a[v_u])](a[v_u]) = \nvmem_c(a[v_u])$
\\By (II6) and (II8) and other parts of the log and volatile memories do not change, 
$\dstate' \rel \ustate'$
\end{tabbing}

\item[Case:] $\ee$ ends in \rulename{CP-Checkpoint} rule  
\begin{tabbing}
By assumption
\\\qquad \= (1)~~ \= $\dstate \rel \ustate$
\\ By  \rulename{CP-Checkpoint} rule
\\\>(2)\> $((\nvmem_c, \vmem_c, \cmd_c), \nvmem, \vmem, \m{checkpoint}(\omega;\cmd)) 
\Stepsto{\m{checkpoint}}  ((\proj{\nvmem}{\omega}, \vmem, \cmd), \nvmem, \vmem, \cmd)$
\\ By \rulename{UL-Commit} rule
\\\>(3) \> $((\ulog, \vmem_c, \cmd_c, \omega', \loggedl), \nvmem_u, \vmem_u, \m{checkpoint}(\omega);\cmd) 
\ulStepsto{\m{checkpoint}}  ((\emptyset, \vmem_u, \cmd, \omega, \emptyset), \nvmem_u, \vmem_u, \cmd)$
\\By assumption and memory doesn't change
\\\>(4)\> $\nvmem = \nvmem_u$ and $\vmem = \vmem_u$
\\By (2), (3) and (4)  
\\\>(5)\> $\forall \loc \in \omega, \nvmem_u(\loc) = \proj{\nvmem}{\omega}(\loc)$
\\By (5), memories do not change, and observation is the same, $\dstate' \rel \ustate'$
\end{tabbing}

\item[Case:] $\ee$ ends in \rulename{CP-reboot} rule 
\begin{tabbing}
By assumption
\\\qquad \= (1)~~ \= $\dstate \rel \ustate$
\\ By  \rulename{CP-reboot} rule
\\\>(2)\> $((\nvmem_c, \vmem_c, \cmd_c), \nvmem, \vmem, \m{reboot}()) 
\Stepsto{\m{reboot}}  ((\nvmem_c, \vmem_c, \cmd_c), \nvmem\lhd\nvmem_c, \vmem_c, \cmd_c)$
\\By \rulename{UL-reboot} rule
\\\>(3)\> $((\ulog, \vmem_c, \cmd_c, \omega, \loggedl), \nvmem_u, \vmem_u, \m{reboot}()) 
\ulStepsto{\m{reboot}}  ((\emptyset, \vmem_c, \cmd_c, \omega, \emptyset), \nvmem_u\lhd\ulog, \vmem_c, \cmd_c)$
\\By assumption that $\dstate \rel \ustate$
\\\>(4)\> $\nvmem_u = \nvmem$ and $\vmem_u = \vmem$
\\\>(5)\> $\forall \loc\in \loggedl, \ulog(\loc) = \nvmem_c(\loc)$ and
$LL= \m{dom}(\ulog)$
\\\>(6)\> $\forall \loc\in \omega\backslash\loggedl, \nvmem_u(\loc) = \nvmem_c(\loc)$
\\T.S. that $\nvmem_u\lhd\ulog = \nvmem\lhd\nvmem_c$
\\By (5) 
\\\>(7)\> $\forall \loc : \loc \in \ulog \land \loc \in \omega 
\Rightarrow \nvmem_u\lhd\ulog(\ulog) = \nvmem\lhd\nvmem_c(\loc)$
\\By (6)
\\\>(8)\> $\forall \loc : \loc \notin \ulog \land \loc \in \omega 
\Rightarrow \nvmem_u(\loc) = \nvmem_c(\loc)$
\\By (7) and (8)
\\\>(9)\> $\nvmem_u\lhd\ulog = \nvmem\lhd\nvmem_c$
\\By (9), $\omega \in \mt{dom}(\nvmem_c)$ and volatile memories do not change
\\\>(11) $\dstate' \rel \ustate'$

\end{tabbing}
\end{description}
\end{proof}

\begin{lem}[Basic check point system simulates undo logging]\label{lem:sim-dino-undo}
$\dstate \rel \ustate$ and $\ustate \ulStepsto{o_1} \ustate' $
then $\exists \dstate'$ s.t. $\dstate \Stepsto{o_2}\dstate'$ and
$\dstate' \rel \ustate'$ and $o_1 = o_2$.
\end{lem}
    
    \begin{proof}
By induction over $\ee::\ustate \ulStepsto{o_1} \ustate' $
\begin{description}
\item[Cases:] $\ee$ ends in \rulename{UL-Skip}, \rulename{UL-V-Assign}, \rulename{UL-If-T}
  \rulename{UL-If-F} rule. The non-volatile memory and checkpointed
  data is not altered. The related basic check point system
  takes the corresponding \rulename{CP-Skip}, \rulename{CP-V-Assign}, \rulename{CP-If-T}
  \rulename{CP-If-F} to reach a related configuration.

\item[Case:] $\ee$ ends in \rulename{UL-Seq} rule. We apply
  I.H. directly. 
    \item[Case:] $\ee$ ends with \rulename{UL-NV-Assign}   
    \begin{tabbing}
    By assumption
    \\\qquad \= (1)~~ \= $\dstate \rel \ustate$
    \\ By  \rulename{UL-NV-Assign} rule
    \\\>(2)\> $(\ucon, \nvmem_u, \vmem_u, x:=e) 
    \ulStepsto{[r_u]}  (\ucon, \nvmem_u[x\mapsto v_u], \vmem_u, \m{skip}) $
    \\\> (3)\> $\nvmem_u, \vmem_u \vdash e \Downarrow_{r_u} v_u$
    \\ By  \rulename{CP-NV-Assign} rule
    \\\>(4)\> $((\nvmem_c, \vmem_c, \cmd_c), \nvmem, \vmem, x:=e ) 
    \Stepsto{[r]}  ((\nvmem_c, \vmem_c, \cmd_c), \nvmem[x\mapsto v], \vmem, \m{skip}) $
    \\\> (5)\> $\nvmem, \vmem \vdash e \Downarrow_r v$
    \\By assumption that $\dstate \rel \ustate$
    \\\> (6) $\vmem = \vmem_u$ and $\nvmem = \nvmem_u$
    \\ By expression evaluation is deterministic
    \\\> (7)\> $v_u = v$ and $r_u = r$
    \\By (7) 
    \\\> (8)\> $\nvmem[x\mapsto v] = \nvmem_u[x\mapsto v_u]$
    \\ By (8) and volatile memory and log do not change, $\dstate' \rel \ustate'$
    \end{tabbing}

    \item[Case:] $\ee$ ends in \rulename{UL-NV-Log}   
    \begin{tabbing}
    By assumption
    \\\qquad \= (1)~~ \= $\dstate \rel \ustate$
    \\ By  \rulename{UL-NV-Log} rule
    \\\>(2)\> $((\ulog, \vmem_c,\cmd_c, \omega, \loggedl), \nvmem_u, \vmem_u, x:=e) $
\\\>\> $   \ulStepsto{[r_u]}  ((\ulog[x \mapsto \nvmem_u(x)], \vmem_c,\cmd_c, \omega, \loggedl \cup {x}), 
    \nvmem_u[x\mapsto v_u], \vmem_u, \m{skip}) $
    \\\> \> $\nvmem_u, \vmem_u \vdash e \Downarrow_{r_u} v_u$ and 
    \\\> (3) \>  $x \in \omega$ and $x \notin \loggedl$ 
    \\ By  \rulename{CP-NV-Assign} rule
    \\\>(4)\> $((\nvmem_c, \vmem_c, \cmd_c), \nvmem, \vmem, x:=e) 
    \Stepsto{[r]}  ((\nvmem_c, \vmem_c, \cmd_c), \nvmem[x\mapsto v], \vmem \m{skip}) $
    \\\>  $\nvmem, \vmem \vdash e \Downarrow_r v$
    \\By assumption that $\dstate \rel \ustate$
    \\\> (5)\> $\vmem = \vmem_u$ and $\nvmem = \nvmem_u$
    \\ By expression evaluation is deterministic
    \\\> (6)\> $v_u = v$ and $r_u = r$
    \\By (6) 
    \\\> (7)\>$\nvmem[x\mapsto v] = \nvmem_u[x\mapsto v_u]$
   \\By (3) and assumption that $\forall \loc\in \omega\backslash\loggedl, \nvmem_u(\loc) = \nvmem_c(\loc)$
   \\\> (8)\> $\ulog[x \mapsto \nvmem_u(x)](x) = \nvmem_c(x)$
    \\ By (4), (8), and volatile memories do not change, $\dstate' \rel \ustate'$
    \end{tabbing}
    
    \item[Case:] $\ee$ ends in \rulename{UL-Arr-Assign} rule  
    \begin{tabbing}
    By assumption
    \\\qquad \= (1)~~ \= $\dstate \rel \ustate$
    \\ By  \rulename{UL-Arr-Assign} rule
    \\\>(2)\> $(\ucon, \nvmem_u, \vmem_u, a[e]:=e') 
        \ulStepsto{[r_u,r_u']}  (\ucon, \nvmem_u[a[v_u]\mapsto v_u'], \vmem_u, \m{skip}) $
    \\\> $\nvmem_u, \vmem_u \vdash e \Downarrow_{r_u} v_u$ and 
    $\nvmem_u, \vmem_u \vdash e' \Downarrow_{r_u'} v_u'$
    \\ By  \rulename{CP-Arr-Assign} rule
    \\\>(3)\> $((\nvmem_c, \vmem_c, \cmd_c), \nvmem, \vmem, a[e]:=e') 
        \Stepsto{[r, r']}  ((\nvmem_c, \vmem_c, \cmd_c), \nvmem[a[v]\mapsto v'], \vmem, \m{skip}) $
    \\\> $\nvmem, \vmem \vdash e \Downarrow_r v$ and $\nvmem, \vmem \vdash e' \Downarrow_{r'} v'$
    \\By assumption that $\dstate \rel \ustate$
    \\\> (4)\> $\vmem = \vmem_u$ and $\nvmem = \nvmem_u$
    \\ By expression evaluation is deterministic
    \\\> (5)\> $v_u = v$ and $v_u' = v'$ and $r_u = r$ and $r_u' = r'$
    \\By (5)
    \\\> (6)\> $\nvmem[a[v]\mapsto v'] = \nvmem_u[a[v_u]\mapsto v_u']$
    \\ By (6) and volatile memories and log do not change, $\dstate' \rel \ustate'$
    \end{tabbing}

    \item[Case:] $\ee$ ends in \rulename{UL-Arr-Log} rule  
    \begin{tabbing}
    By assumption
    \\\qquad \= (1)~~ \= $\dstate \rel \ustate$

    \\ By  \rulename{UL-Arr-Log} rule
    \\\>(2)\> $((\ulog, \vmem_c,\cmd_c, \omega, \loggedl), \nvmem_u, \vmem_u, a[e]:=e')$ 
       \\\>\>$ \ulStepsto{[r_u,r_u']}  ((\ulog[a[v_u] \mapsto \nvmem_u(a[v_u])], 
        \vmem_c,\cmd_c, \omega, \loggedl \cup a[v_u]), 
        \nvmem_u[a[v_u]\mapsto v_u'], \vmem_u, \m{skip}) $
    \\\>\>$\nvmem_u, \vmem_u \vdash e \Downarrow_{r_u} v_u$ 
    and $\nvmem_u, \vmem_u \vdash e' \Downarrow_{r_u'} v_u'$
    \\\>(3)\> $a \in \omega$ and $a[v_u] \notin \loggedl$ 
    \\ By  \rulename{CP-Arr-Assign} rule
    \\\>(3)\> $((\nvmem_c, \vmem_c, \cmd_c), \nvmem, \vmem, a[e]:=e') 
        \Stepsto{[r, r']}  ((\nvmem_c, \vmem_c, \cmd_c), \nvmem[a[v]\mapsto v'], \vmem, \m{skip}) $
    \\\> $\nvmem, \vmem \vdash e \Downarrow_r v$ and $\nvmem, \vmem \vdash e' \Downarrow_{r'} v'$
    \\By assumption that $\dstate \rel \ustate$
    \\\> (4) \> $\vmem = \vmem_u$ and $\nvmem = \nvmem_u$
   \\By (3) and assumption that 
    $\forall \loc\in \omega\backslash\loggedl, \nvmem_u(\loc) = \nvmem_c(\loc)$
    \\\>(6) \>$\nvmem_u(a[v_u]) = \nvmem_c(a[v_u])$
    \\ By expression evaluation is deterministic
    \\\> (7)\> $v_u = v$ and$v_u' = v'$ and $r_u = r$ and $r_u' = r'$
    \\By (7)
    \\\> (8)\> $\nvmem[a[v]\mapsto v'] = \nvmem_u[a[v_u]\mapsto v_u']$
    \\By (6)
    \\\>(9) $\ulog[a[v_u] \mapsto \nvmem_u(a[v_u])](a[v_u]) = \nvmem_c(a[v_u])$
    \\By (8), (9), and rest of log and volatile memories do not change, $\dstate' \rel \ustate'$
    \end{tabbing}
    
    \item[Case:] $\ee$ ends in \rulename{UL-Commit} rule  
    \begin{tabbing}
    By assumption
    \\\qquad \= (1)~~ \= $\dstate \rel \ustate$
    \\ By \rulename{UL-Commit} rule
    \\\>(2) \> $((\ulog, \vmem_c, \cmd_c, \omega', \loggedl, \nvmem_u, \vmem_u, \m{checkpoint}(\omega);\cmd) 
    \ulStepsto{\m{checkpoint}}  ((\emptyset, \vmem_u, \cmd, \omega, \emptyset), \nvmem_u, \vmem_u, \cmd)$
    \\ By  \rulename{CP-Checkpoint} rule
    \\\>(3)\> $((\nvmem_c, \vmem_c, \cmd_c), \nvmem, \vmem, \m{checkpoint(\omega)};\cmd) 
    \Stepsto{\m{checkpoint}}  ((\proj{\nvmem}{\omega}, \vmem, \cmd), \nvmem, \vmem, \cmd)$
    \\By assumption and memory doesn't change
    \\\>(4)\> $\nvmem = \nvmem_u$
    \\\>(5)\> $\vmem = \vmem_u$
    \\By (2), (3), and (4) 
    \\\>(6)\> $\nvmem_u(\omega) = \proj{\nvmem}{\omega}$
    \\By (6), observation is the same, and memories do not change, $\dstate' \rel \ustate'$
    \end{tabbing}

    \item[Case:] $\ee$ ends in \rulename{UL-Reboot} rule 
    \begin{tabbing}
    By assumption
    \\\qquad \= (1)~~ \= $\dstate \rel \ustate$
    \\By \rulename{UL-reboot}
    \\\>(2)\> $((\ulog, \vmem_c, \cmd_c, \omega, \loggedl), \nvmem_u, \vmem_u', \m{reboot}()) 
    \ulStepsto{\m{reboot}}  ((\emptyset, \vmem_c, \cmd_c, \omega, \emptyset), \nvmem_u\lhd\ulog, \vmem_c, \cmd_c)$
    \\ By  \rulename{CP-reboot} rule
    \\\>(3)\> $((\nvmem_c, \vmem_c, \cmd_c), \nvmem, \vmem, \m{reboot}) 
    \Stepsto{\m{reboot}()}  ((\nvmem_c, \vmem_c, \cmd_c), \nvmem\lhd\nvmem_c, \vmem_c, \cmd_c)$
    \\By assumption
    \\\>(4)\> $\nvmem_u = \nvmem$
    \\\>(5)\> $\forall \loc\in \loggedl, \ulog(\loc) = \nvmem_c(\loc)$
    and $\loggedl=\m{dom}{\ulog}$
    \\\>(6)\> $\forall \loc\in \omega\backslash\loggedl, \nvmem_u(\loc) = \nvmem_c(\loc)$
    \\T.S. that $\nvmem_u\lhd\ulog = \nvmem\lhd\nvmem_c$
    \\By (5), (6) and $\omega = \mt{dom}(\nvmem_c)$
    \\\>(7)\> $\forall \loc \in \omega: \loc \in \ulog 
    \Rightarrow \nvmem_u\lhd\ulog(log) = \nvmem\lhd\nvmem_c(\loc)$
    \\\>(8)\> $\forall \loc \in \omega : \loc \notin \ulog
    \Rightarrow \nvmem_u(\loc) = \nvmem_c(\loc)$
    \\By (7) and (8)
\\\>(9)\> $\nvmem_u\lhd\ulog = \nvmem\lhd\nvmem_c$
\\By (9), $\omega = \mt{dom}(\nvmem_c)$, observation is the same, and volatile memories do not change
\\\>(11) $\dstate' \rel \ustate'$
    
    \end{tabbing}

    \end{description}
    \end{proof}

\begin{cor}[Correctness of Undo Logging]
 If $(\emptyset, \nvmem, \vmem, \cmd) \ulMStepsto{O_1} \Sigma$,
  $\mt{CP}(\Sigma)$ 
    and  $\Vdash \cmd: \m{ok}$ 
    then $\exists O_2, \sigma$ s.t. 
  $(\nvmem, \vmem, \cmd) \MSeqStepsto{O_2}\sigma$, 
  $\erase{\Sigma} = \sigma$
    and  $O_1\oleqmc O_2$.
\end{cor}
\begin{proofsketch}
  By the simulation relation, for any trace in undo logging, there is
  a trace in basic checkpoint system with the same observations and
  erased configurations, and the correctness follows from that.
\end{proofsketch}

\subsection{Redo Logging}
\label{sec:variants-redo}
\paragraphb{Syntax and Operational Semantics}
The checkpointed context $\rlctx$ is defined as follows. 
\[
\begin{array}{llcl}
\textit{context} & \rlctx & \bnfdef & (\ulog, \vmem, \cmd, \omega)
\end{array}
\]
We summarize the semantics for redo-logging below. 

\noindent\framebox{$(\rlctx, \nvmem, \vmem, \cmd) 
\rlStepsto{} (\rlctx', \nvmem', \vmem', \cmd')$}
\begin{mathpar}
\inferrule*[right=RL-PowerFail]{ }{
  (\rlctx, \nvmem, \vmem, \cmd) 
 \rlStepsto{}  (\rlctx, \nvmem, \resetm(\vmem), \m{reboot}()) 
}
\and
\inferrule*[right=RL-CheckPoint]{ }{
  ((\ulog, \vmem', \cmd', \omega'), \nvmem, \vmem, \m{checkpoint}(\omega);\cmd) 
 \rlStepsto{\m{checkpoint}}  ((\emptyset, \vmem, \cmd, \omega), \nvmem \lhd \ulog, \vmem, \cmd) 
}
\and
\inferrule*[right=RL-Reboot]{ }{
  ((\ulog, \vmem, \cmd, \omega), \nvmem, \vmem', \m{reboot}()) 
  \rlStepsto{\m{reboot}}  ((\emptyset, \vmem, \cmd, \omega), \nvmem, \vmem, \cmd) 
}
\and
\inferrule*[right=RL-NV-Log]{ x\in\m{dom}(N)\\ x\in\omega\\\nvmem \lhd \ulog, \vmem \vdash e \Downarrow_r v}{
  ((\ulog, \vmem_c, \cmd_c, \omega), \nvmem, \vmem, x := e) 
 \rlStepsto{[r]}  ((\ulog[x\mapsto v], \vmem_c, \cmd_c, \omega), 
  \nvmem, \vmem, \m{skip}) 
}

\and
\inferrule*[right=RL-NV-Assign]{  x\in\m{dom}(N)\\ x\notin\omega \\\nvmem\lhd \ulog, \vmem \vdash e \Downarrow_r v}{
  ((\ulog, \vmem_c, \cmd_c, \omega), \nvmem, \vmem, x := e) 
 \rlStepsto{[r]}  ((\ulog, \vmem_c, \cmd_c, \omega), 
  \nvmem[x\mapsto v], \vmem, \m{skip}) 
}

\and
\inferrule*[right=RL-Arr-Log]{ a\in\omega
\\\nvmem \lhd \ulog, \vmem \vdash e \Downarrow_r v \\\nvmem\lhd \ulog, \vmem  \vdash e' \Downarrow_{r'} v'}{
  ((\ulog, \vmem_c, \cmd_c, \omega), \nvmem, \vmem, a[e] := e') 
 \rlStepsto{[r,r']}  ((\ulog[a[v]\mapsto v'], \vmem_c, \cmd_c, \omega), 
  \nvmem, \vmem, \m{skip}) 
}
\end{mathpar}
\begin{mathpar}
\inferrule*[right=RL-Arr-Assign]{  a\in\m{dom}(N)\\ a\notin\omega 
\\\nvmem \lhd \ulog, \vmem \vdash e \Downarrow_r v \\\nvmem \lhd \ulog, \vmem \vdash e' \Downarrow_{r'} v'}{
  ((\ulog, \vmem_c, \cmd_c, \omega), \nvmem, \vmem, a[e] := e';\cmd) 
 \rlStepsto{[r,r']}  ((\ulog, \vmem_c, \cmd_c, \omega), 
  \nvmem[a[v]\mapsto v'], \vmem, \cmd) 
}
\end{mathpar}

\begin{mathpar}
\inferrule*[right=RL-Seq]{  (\rlctx, \nvmem, \vmem, \iota) 
 \rlStepsto{o}  (\rlctx, \nvmem', \vmem', \m{skip})}{
  (\rlctx, \nvmem, \vmem, \iota;c) 
 \rlStepsto{o}  (\rlctx, \nvmem', \vmem', c) 
}
\and
\inferrule*[right=RL-If-T]{ \rlctx = (\ulog, \vmem, \cmd, \omega)\\N \lhd \ulog, V \vdash e \Downarrow_{r} \m{true}}{
  (\rlctx, \nvmem, \vmem, \m{if}\ e\ \m{then}\ \cmd_1\ \m{else}\ \cmd_2) 
 \rlStepsto{[r]}  ((\ulog, \vmem, \cmd, \omega), \nvmem, \vmem, \cmd_1) 
}
\and
\inferrule*[right=RL-If-F]{\rlctx = (\ulog, \vmem, \cmd, \omega)\\N \lhd \ulog, V \vdash e \Downarrow_{r} \m{false}}{
  (\rlctx, \nvmem, \vmem, \m{if}\ e\ \m{then}\ \cmd_1\ \m{else}\ \cmd_2) 
 \rlStepsto{[r]}  (\rlctx, \nvmem, \vmem, \cmd_2) 
}
\end{mathpar}

\paragraphb{Equivalence to basic checkpoint system}
We define a binary relation between basic checkpoint and redo-logging
state, written $\dstate \rel \rstate$, as follows.  We prove that it
is a bisimulation relation.
\[
\inferrule*{
\dstate = (\context, \nvmem, \vmem, \cmd) 
\\ \rstate = (\rcon, \nvmem_r, \vmem_r, \cmd_r)
\\ \dcon = (\nvmem_c, \vmem_c, \cmd_c)
\\ \rcon = (\ulog, \vmem_c, \cmd_c, \omega)
\\ \cmd= \cmd_r
\\ \vmem = \vmem_r
\\ \omega = \m{dom}(N_c)
\\ \omega = \m{dom}(\ulog)
\\ \forall \loc. \nvmem_r(\loc) \neq \nvmem(\loc), \loc \in \ulog
\\ \forall \loc\in \ulog, \ulog(\loc) = \nvmem(\loc)
\\ \forall \loc\in \omega, \nvmem_r(\loc) = \nvmem_c(\loc)
}{
\dstate \rel \rstate
}
\]

\begin{lem}
$\dstate \rel \rstate$ and $\dstate \Stepsto{o_1} \dstate' $
then $\exists \rstate'$ s.t. $\rstate \rlStepsto{o_2}\rstate'$ and
$\dstate' \rel \rstate'$ and $o_1 = o_2$.
\end{lem}

\begin{proof}

By induction over $\ee::\dstate \Stepsto{o_1} \dstate' $

\begin{description}
\item[Cases:] $\ee$ ends in \rulename{CP-Skip}, \rulename{CP-V-Assign}, \rulename{CP-If-T}
  \rulename{CP-If-F} rule. The non-volatile memory and checkpointed
  data is not altered. The undo logging system
  takes corresponding \rulename{RL-Skip}, \rulename{RL-V-Assign}, \rulename{RL-If-T}
  \rulename{RL-If-F} to reach a related configuration.

\item[Case:] $\ee$ ends in \rulename{CP-Seq} rule. We apply
  I.H. directly. 

\item[Case:]$\ee$ ends in  \rulename{CP-NV-Assign} rule
    \begin{tabbing}
    By assumption
    \\\qquad \= (1)~~ \= $\dstate \rel \rstate$
    \\ By  \rulename{CP-NV-Assign} rule
    \\\>(2)\> $((\nvmem_c, \vmem_c, \cmd_c), \nvmem, \vmem, x:=e) 
    \Stepsto{[r]}  ((\nvmem_c, \vmem_c, \cmd_c), \nvmem[x\mapsto v], \vmem,\m{skip}) $
    \\\> (3)\> $\nvmem, \vmem \vdash e \Downarrow_r v$
    \\We consider two subcases: (I) $x \in \omega$ 
    and (II) $x \notin \omega$
    \\ {\bf Subcase (I)} $x \notin \omega$
    \\ By  \rulename{RL-NV-Assign} rule
    \\\>(I1)\> $(\rcon, \nvmem_r, \vmem_r, x:=e) 
    \rlStepsto{[r_r]}  (\rcon, \nvmem_r[x\mapsto v_r], \vmem_r,\m{skip}) $
    \\\> (I2)\> $\nvmem_r \lhd \ulog, \vmem_r \vdash e \Downarrow_{r_r} v_r$
    \\ By assumption that $\dstate \rel \rstate$
    \\\> (I3) $\vmem_r = \vmem$
    \\ By (I3) and expression evaluation is deterministic
    \\\> (I4)\> $v_r = v$ and $r = r_r$ 
    \\By (I4) 
    \\\>(I5)$\nvmem[x\mapsto v] = \nvmem_r[x\mapsto v_r]$
    \\By (I5), volatile memory and log do not change, $\dstate' = \rstate'$
    \\ {\bf Subcase (II)} $x \in \omega$
    \\ By  \rulename{RL-NV-Log} rule
    \\\>(II1)\>\ $((\ulog, \vmem_c,\cmd_c, \omega), \nvmem_r, \vmem_r, x:=e) 
    \rlStepsto{[r_r]}  ((\ulog[x \mapsto v_r], \vmem_c,\cmd_c, \omega), 
    \nvmem_r, \vmem_r,\m{skip}) $
    \\\>(II2)\>\ $\nvmem_r \lhd \ulog, \vmem_r \vdash e \Downarrow_{r_r} v_r$ 
    \\ By assumption that $\dstate \rel \rstate$
    \\\> (II3) $\vmem_r = \vmem$
    \\ By (II3) and expression evaluation is deterministic
    \\\> (II4)\>\ $v_r = v$ and $r = r_r$ 
    \\By (II4)
    \\\>(II5)\>~$\nvmem[x\mapsto v] = \ulog[x\mapsto v_r]$
    \\By assumption and $\nvmem_r$ does not change
    \\\>(II6)\> $\forall \loc \in \omega : \nvmem_r[\loc] = \nvmem_c[\loc]$
    \\By volatile memory does not change, $\nvmem_r$ does not change,
     $\dstate' = \rstate'$
    \end{tabbing}
    
    \item[Case:] $\ee$ ends in \rulename{CP-Assign-Arr} rule
    \begin{tabbing}
        By assumption
        \\\qquad \= (1)\quad \= $\dstate \rel \ustate$
        \\ By  \rulename{CP-Assign-Arr} rule
        \\\>(2)\> $((\nvmem_c, \vmem_c, \cmd_c), \nvmem, \vmem, a[e]:=e') 
        \Stepsto{[r, r']}  ((\nvmem_c, \vmem_c, \cmd_c), \nvmem[a[v]\mapsto v'], \vmem, \m{skip}) $
        \\\> (3)\> $\nvmem, \vmem \vdash e \Downarrow_r v$
        \\\> (4)\> $\nvmem, \vmem \vdash e' \Downarrow_{r'} v'$
        \\We consider two subcases: (I) $a \in \omega$ 
        and (II) $a \notin \omega$
        \\ {\bf Subcase (I)} $a \notin \omega$
        \\ By  \rulename{RL-Arr-Assign} rule
        \\\>(I1)\> $(\rcon, \nvmem_r, \vmem_r, a[e]:=e') 
        \rlStepsto{[r_r,r_r']}  (\rcon, \nvmem_r[a[v_r]\mapsto v_r'], \vmem_r, \m{skip}) $
        \\\> (I2)\> $\nvmem_r \lhd \ulog, \vmem_r \vdash e \Downarrow_{r_r} v_r$
        \\\> (I3)\> $\nvmem_r \lhd \ulog, \vmem_r \vdash e' \Downarrow_{r_r'} v_r'$
        \\ By assumption
        \\\> (I4) $\vmem = \vmem_r$ and $\ulog \in \nvmem$
        \\By (I4),  expression evaluation is deterministic
        \\\>(I5)\>~ $v_r = v$ and $v_r' = v'$ and $r_r = r$ and $r_r' = r'$
        \\By (I5)
        \\\> (I6)\> $\nvmem[a[v]\mapsto v'] = \nvmem_r[a[v_r]\mapsto v_r']$
        \\ By (I6), volatile memory and log do not change, $\dstate' \rel \rstate'$.
    
        \\ {\bf Subcase (II)} $a \in \omega$
        \\ By  \rulename{RL-Arr-Update} rule
        \\\>(II1)\> $((\ulog, \vmem_c,\cmd_c, \omega), \nvmem_r, \vmem_r, a[e]:=e') 
        \rlStepsto{[r,r']}  ((\ulog[a[v_r] \mapsto v_r'], \vmem_c,\cmd_c, \omega), 
        \nvmem_r, \vmem_r, \m{skip}) $
        \\\>(II2)\> $\nvmem_r \lhd \ulog, \vmem_r\vdash e \Downarrow_r v_r$
        \\\>(II3)\> $\nvmem_r \lhd \ulog, \vmem_r\vdash e' \Downarrow_{r'} v_r'$
        \\ By assumption
        \\\> (II4) $\vmem = \vmem_r$ and $\ulog \in \nvmem$
        \\By (II4), expression evaluation is deterministic
        \\\>(II5)\> $v_r = v$ and $v_r' = v'$ and $r_r = r$ and $r_r' = r'$
        \\By (II5)
        \\\>(II6)\> $\nvmem[a[v]\mapsto v'] = \ulog[a[v_r]\mapsto v_r']$
        \\By volatile memory does not change, $\nvmem_r$ does not change,
     $\dstate' = \rstate'$
    \end{tabbing}
    
    \item[Case:] $\ee$ ends in \rulename{CP-Checkpoint} rule
    \begin{tabbing}
    By assumption
    \\\qquad \= (1)~~ \= $\dstate \rel \rstate$
    \\ By  \rulename{CP-Checkpoint} rule
    \\\>(2)\> $((\nvmem_c, \vmem_c, \cmd_c), \nvmem, \vmem, \m{checkpoint}(\omega);\cmd) 
    \Stepsto{\m{checkpoint}}  ((\proj{\nvmem}{\omega}, \vmem, \cmd), \nvmem, \vmem, \cmd)$
    \\ By \rulename{RL-Commit} rule
    \\\>(3)\> $((\ulog, \vmem_c, \cmd_c, \omega'), \nvmem_r, \vmem_r, \m{checkpoint}(\omega);\cmd) 
    \rlStepsto{\m{checkpoint}}  ((\emptyset, \vmem_r, \cmd, \omega), 
    \nvmem_r \lhd \ulog, \vmem_r, \cmd)$
    \\By assumption that $\ulog \subseteq \nvmem$ and 
    $\forall \loc$ s.t. $\nvmem_r(\loc) \neq \nvmem(\loc) : \loc \in \ulog$
    \\\> (4) $\nvmem_r \lhd \ulog = \nvmem$
    \\By (4), volatile memory does not change, and observation is the same, $\dstate' = \rstate'$.
    \end{tabbing}

    \item[Case:] $\ee$ ends in \rulename{CP-reboot} rule  
    \begin{tabbing}
    By assumption
    \\\qquad \= (1)~~ \= $\dstate \rel \ustate$
    \\ By  \rulename{CP-reboot} rule
    \\\>(2)\> $((\nvmem_c, \vmem_c, \cmd_c), \nvmem, \vmem, \m{reboot}) 
    \Stepsto{\m{reboot}()}  ((\nvmem_c, \vmem_c, \cmd_c), \nvmem\lhd\nvmem_c, \vmem_c, \cmd_c)$
    \\By \rulename{RL-reboot} rule
    \\\>(3)\> $((\ulog, \vmem_c, \cmd_c, \omega), \nvmem_r, \vmem_r, \m{reboot}()) 
    \rlStepsto{\m{reboot}}  ((\emptyset, \vmem_c, \cmd_c, \omega), \nvmem_r, \vmem_c, \cmd_c)$
    \\By assumption
    \\\>(4) $\forall \loc $s.t. $\nvmem_r(\loc) \neq \nvmem(\loc), \loc \in \ulog$
    \\\>(5) $\forall \loc \in \omega, \nvmem_r(\loc) = \nvmem_c(\loc)$
    \\By (4) and semantics are deterministic 
    \\\>(6) $\forall \loc$ s.t. $\nvmem_r(\loc) \neq \nvmem(\loc), \loc \in \omega$
    \\By (5) and (6)
    \\\>(7) $\forall \loc$ s.t. $\nvmem_r(\loc) \neq \nvmem(\loc), \nvmem_r(\loc) = \nvmem_c(\loc)$
    \\By (7)
    \\\> (8) $\nvmem\lhd\nvmem_c = \nvmem_r$
    \\By (8), volatile memories are the same, and observation is the same, $\dstate' \rel \rstate'$
    
    \end{tabbing}

    \end{description}
    
\end{proof}

\begin{lem}
  $\dstate \rel \rstate$ and $\rstate \rlStepsto{o_1} \rstate' $
  then $\exists \dstate'$ s.t. $\dstate \Stepsto{o_2}\dstate'$ and
  $\dstate' \rel \rstate'$ and $o_1 = o_2$.
  \end{lem}

  \begin{proof}
  
  By induction over $\ee::\rstate \Stepsto{o_1} \rstate' $
  
  \begin{description}
  \item[Cases:] $\ee$ ends in \rulename{RL-Skip}, \rulename{RL-V-Assign}, \rulename{RL-If-T},
    \rulename{RL-If-F} rule. The non-volatile memory and checkpointed
    data is not altered. The checkpoint system
    takes corresponding \rulename{CP-Skip}, \rulename{CP-V-Assign}, \rulename{CP-If-T},
    \rulename{CP-If-F} to reach a related configuration.

  \item[Case:] $\ee$ ends in \rulename{RL-Seq} rule. We apply
    I.H. directly. 
  
  \item[Case:]$\ee$ ends in  \rulename{RL-NV-Assign} rule
      \begin{tabbing}
      By assumption
      \\\qquad \= (1)~~ \= $\dstate \rel \rstate$
      \\ By  \rulename{RL-NV-Assign} rule
      \\\>(2)\> $(\rcon, \nvmem_r, \vmem_r, x:=e) 
      \rlStepsto{[r_r]}  (\rcon, \nvmem_r[x\mapsto v_r], \vmem_r, \m{skip}) $
      \\\> (3)\> $\nvmem_r \lhd \ulog, \vmem_r \vdash e \Downarrow_{r_r} v_r$
      \\ By  \rulename{CP-NV-Assign} rule
      \\\>(4)\> $((\nvmem_c, \vmem_c, \cmd_c), \nvmem, \vmem, x:=e) 
      \Stepsto{[r]}  ((\nvmem_c, \vmem_c, \cmd_c), \nvmem[x\mapsto v], \vmem, \m{skip}) $
      \\\> (5)\> $\nvmem, \vmem \vdash e \Downarrow_r v$
      \\ By assumption that $\dstate \rel \rstate$
      \\\> (6\>) $\vmem_r = \vmem$
      \\ By (6), and expression evaluation is deterministic
      \\\> (7)\> $v_r = v$ and $r = r_r$ 
      \\By (7) 
      \\\>(8)\> $\nvmem_r[x\mapsto v_r] = \nvmem[x\mapsto v]$
      \\By (8), volatile memory and log do not change, $\dstate' = \rstate'$
      \end{tabbing}

  \item[Case:]$\ee$ ends in  \rulename{RL-NV-Log} rule
      \begin{tabbing}
      By assumption
      \\\qquad \= (1)~~ \= $\dstate \rel \rstate$
      \\ By  \rulename{RL-Update} rule
      \\\>(2)\>\ $((\ulog, \vmem_c,\cmd_c, \omega), \nvmem_r, \vmem_r, x:=e) 
      \rlStepsto{[r_r]}  ((\ulog[x \mapsto v_r], \vmem_c,\cmd_c, \omega), 
      \nvmem_r, \vmem_r, \m{skip}) $
      \\\>(3)\>\ $\nvmem_r \lhd \ulog, \vmem_r \vdash e \Downarrow_{r_r} v_r$ 
      \\ By  \rulename{CP-NV-Assign} rule
      \\\>(4)\> $((\nvmem_c, \vmem_c, \cmd_c), \nvmem, \vmem, x:=e) 
      \Stepsto{[r]}  ((\nvmem_c, \vmem_c, \cmd_c), \nvmem[x\mapsto v], \vmem, \m{skip}) $
      \\\> (5)\> $\nvmem, \vmem \vdash e \Downarrow_r v$
      \\ By assumption that $\dstate \rel \rstate$
      \\\> (6)\> $\vmem_r = \vmem$ 
      \\ By (6), and expression evaluation is deterministic
      \\\> (7)\>\ $v_r = v$ and $r = r_r$ 
      \\By (7)
      \\\>(8)\>$\nvmem[x\mapsto v] = \ulog[x\mapsto v_r]$
      \\By volatile memory does not change, $\nvmem_r$ does not change,
       $\dstate' = \rstate'$
      \end{tabbing}
      
  \item[Case:] $\ee$ ends in \rulename{RL-Assign-Arr} rule
    \begin{tabbing}
        By assumption
        \\\qquad \= (1)\quad \= $\dstate \rel \ustate$
        \\ By  \rulename{RL-Arr-Assign} rule
        \\\>(2)\> $(\rcon, \nvmem_r, \vmem_r, a[e]:=e') 
          \rlStepsto{[r_r,r_r']}  (\rcon, \nvmem_r[a[v_r]\mapsto v_r'], \vmem_r, \m{skip}) $
        \\\> (3)\> $\nvmem_r \lhd \ulog, \vmem_r\vdash e \Downarrow_{r_r} v_r$
        \\\> (4)\> $\nvmem_r \lhd \ulog, \vmem_r\vdash e' \Downarrow_{r_r'} v_r'$
        \\ By  \rulename{CP-Assign-Arr} rule
        \\\>(5)\> $((\nvmem_c, \vmem_c, \cmd_c), \nvmem, \vmem, a[e]:=e') 
          \Stepsto{[r, r']}  ((\nvmem_c, \vmem_c, \cmd_c), \nvmem[a[v]\mapsto v'], \vmem, \m{skip}) $
        \\\> (6)\> $\nvmem, \vmem \vdash e \Downarrow_r v$
        \\\> (7)\> $\nvmem, \vmem \vdash e' \Downarrow_{r'} v'$
         \\ By assumption
        \\\> (8)\> $\vmem = \vmem_r$
        \\By (8), and expression evaluation is deterministic
        \\\>(9)\> $v_r = v$ and $v_r' = v'$ and $r_r = r$ and $r_r' = r'$
        \\By (9)
        \\\> (10)\> $\nvmem[a[v]\mapsto v'] = \nvmem_r[a[v_r]\mapsto v_r']$
        \\ By (10), volatile memory and log do not change, $\dstate' \rel \rstate'$.
    \end{tabbing}
      
    \item[Case:]$\ee$ ends in  \rulename{RL-Arr-Update} rule
      \begin{tabbing} 
        By assumption
        \\\qquad \= (1)\quad \= $\dstate \rel \ustate$
        \\ By  \rulename{RL-Arr-Update} rule
        \\\>(2)\> $((\ulog, \vmem_c,\cmd_c, \omega), \nvmem_r, \vmem_r, a[e]:=e') 
          \rlStepsto{[r,r']}  ((\ulog[a[v_r] \mapsto v_r'], \vmem_c,\cmd_c, \omega), 
          \nvmem_r, \vmem_r, \m{skip}) $
        \\\>(3)\> $\nvmem_r \lhd \ulog, \vmem_r\vdash e \Downarrow_r v_r$
        \\\>(4)\> $\nvmem_r \lhd \ulog, \vmem_r\vdash e' \Downarrow_{r'} v_r'$
        \\ By  \rulename{CP-Assign-Arr} rule
        \\\>(5)\> $((\nvmem_c, \vmem_c, \cmd_c), \nvmem, \vmem, a[e]:=e') 
          \Stepsto{[r, r']}  ((\nvmem_c, \vmem_c, \cmd_c), \nvmem[a[v]\mapsto v'], \vmem, \m{skip}) $
        \\\> (6)\> $\nvmem, \vmem \vdash e \Downarrow_r v$
        \\\> (7)\> $\nvmem, \vmem \vdash e' \Downarrow_{r'} v'$
        \\ By assumption
        \\\> (8) $\vmem = \vmem_r$
        \\By (9), and expression evaluation is deterministic
        \\\>(10)\> $v_r = v$ and $v_r' = v'$ and $r_r = r$ and $r_r' = r'$
        \\By (10)
        \\\>(11)\> $\nvmem[a[v]\mapsto v'] = \ulog[a[v_r]\mapsto v_r']$
        \\By (11), volatile memory does not change, 
        $\nvmem_r$ does not change,
      $\dstate' = \rstate'$
      \end{tabbing}
      
      \item[Case:] $\ee$ ends in \rulename{RL-Commit} rule
      \begin{tabbing}
      By assumption
      \\\qquad \= (1)\quad  \= $\dstate \rel \rstate$
      \\ By \rulename{RL-Commit} rule
      \\\>(2)\> $((\ulog, \vmem_c, \cmd_c, \omega'), \nvmem_r, \vmem_r,
       \m{checkpoint}(\omega);\cmd) 
       \rlStepsto{\m{checkpoint}}  ((\emptyset, \vmem_r, \cmd, \omega), 
      \nvmem_r \lhd \ulog, \vmem_r, \cmd)$
      \\ By  \rulename{CP-Checkpoint} rule
      \\\>(3)\> $((\nvmem_c, \vmem_c, \cmd_c), \nvmem, \vmem, \m{checkpoint}(\omega);\cmd) 
      \Stepsto{\m{checkpoint}}  ((\proj{\nvmem}{\omega}, \vmem, \cmd), \nvmem, \vmem, \cmd)$
      \\By assumption that $\ulog \subseteq \nvmem$ and $\forall \loc$ 
      s.t. $\nvmem_r(\loc) \neq \nvmem(\loc), \loc \in \ulog$ 
      \\\> (4) $\nvmem_r \lhd \ulog = \nvmem$
      \\By (4), volatile memory does not change, and observation is the same, $\dstate' = \rstate'$.
      \end{tabbing}
      
      \item[Case:] $\ee$ ends in \rulename{RL-reboot} rule  
      \begin{tabbing}
      By assumption
      \\\qquad \= (1)\quad \= $\dstate \rel \ustate$
      \\By \rulename{RL-reboot} rule
      \\\>(2)\> $((\ulog, \vmem_c, \cmd_c, \omega), \nvmem_r, \vmem_r, \m{reboot}()) 
      \rlStepsto{\m{reboot}}  ((\emptyset, \vmem_c, \cmd_c, \omega), \nvmem_r, \vmem_c, \cmd_c)$
      \\ By  \rulename{CP-reboot} rule
      \\\>(3)\> $((\nvmem_c, \vmem_c, \cmd_c), \nvmem, \vmem, \m{reboot}) 
      \Stepsto{\m{reboot}()}  ((\nvmem_c, \vmem_c, \cmd_c), \nvmem\lhd\nvmem_c, \vmem_c, \cmd_c)$
      
      \\By assumption
      \\\>(4) $\forall \loc$ s.t. $\nvmem_r(\loc) \neq \nvmem(\loc), \loc \in \ulog$
      \\\>(5) $\forall \loc \in \omega, \nvmem_r(\loc) = \nvmem_c(\loc)$
      \\By (4) and semantics are deterministic
      \\\>(6) $\forall \loc$ s.t. $\nvmem_r(\loc) \neq \nvmem(\loc), \loc \in \omega$
      \\By (5) and (6)
      \\\>(7) $\forall \loc$ s.t. $\nvmem_r(\loc) \neq \nvmem(\loc), \nvmem_r(\loc) = \nvmem_c(\loc)$
      \\By (7)
      \\\> (8) $\nvmem\lhd\nvmem_c = \nvmem_r$
      \\By (8), volatile memories are the same, and observation is the same, $\dstate' \rel \rstate'$
      
      \end{tabbing}
  
      \end{description}
      
  \end{proof}

\begin{cor}[Correctness of Redo Logging]
 If $(\emptyset, \nvmem, \vmem, \cmd) \rlMStepsto{O_1} \Sigma$,
  $\mt{CP}(\Sigma)$ 
    and  $\Vdash \cmd: \m{ok}$ 
    then $\exists O_2, \sigma$ s.t. 
  $(\nvmem, \vmem, \cmd) \MSeqStepsto{O_2}\sigma$, 
  $\erase{\Sigma} = \sigma$
    and  $O_1\oleqmc O_2$.
\end{cor}
\begin{proofsketch}
By the bi-simulation relation, for any trace in undo logging, there is
  a trace in basic checkpoint system with the same observations and
  erased configurations, and the correctness follows from that. 
\end{proofsketch}

\section{Tasks}

\subsection{Syntax and Operational Semantics}
We define the syntax and semantics of a task-based system using redo
logging. 
\[
\begin{array}{llcl}
\textit{Task-shared memory} & \tshared & : & M
\\
\textit{Privatized memory} & \tpriv & : & M
\\
\textit{Task-local memory} & \tlocal & : & M
\\
\textit{Task IDs} & i & :  & \m{Int} 
\\
\textit{Instr.} & \iota & \bnfdef & \cdots\bnfalt \m{toTask}(i)
\\
\textit{Task Map} & T & \bnfdef & \cdot\bnfalt T, i\mapsto (\omega, \cmd)
\\
\textit{context} & \tskctx & \bnfdef & (T, i) 
\end{array}
\]

We first define a function to reset the volatile part of a memory
region. The $\mt{resetVol}$ function, given memory $M$, resets any volatile 
portion of $M$ and persists any non-volatile portion. 
\begin{mathpar}
  \inferrule*[]{ M = M_n, M_v \\\mt{persistent}(M_n)  \\ \mt{volatile}(M_v)}{
    \mt{resetVol}(M) = M_n, \mt{reset}(M_v)
  }
\end{mathpar}

\noindent\framebox{$(\context, \tshared, \tpriv, \tlocal, \cmd) 
 \tskStepsto{O} (\context', \tshared', \tpriv', \tlocal', \cmd')$}
\begin{mathpar}

\inferrule*[right=TSK-PowerFail]{ \tlocal' = \mt{resetVol}(\tlocal) }{
  (\tskctx, \tshared, \tpriv, \tlocal, \cmd) 
 \tskStepsto{}  (\tskctx, \tshared, \tpriv,
 \tlocal', \m{reboot})
}

\and

\and
\inferrule*[right=TSK-Reboot]{ 
\tskctx = (T, i)
\\ T(i) = (\omega, \cmd)
}{
  (\tskctx, \tshared, \tpriv, \tlocal, \m{reboot}) 
 \tskStepsto{}  (\tskctx, \tshared, \emptyset, \tlocal, \cmd) 
}
\and

\inferrule*[right=TSK-Trans]{ 
\tskctx = (T, i)
\\ T(j) = (\omega, \cmd)
}{
  (\tskctx, \tshared, \tpriv, \tlocal, \m{toTask}(j)) 
 \tskStepsto{}  ((T,j)), \tshared \lhd \tpriv, \emptyset, \tlocal, \cmd) 
}
\end{mathpar}
\begin{mathpar}
\inferrule*[right=TSK-Update-S]{\tskctx = (T, i)
\\ T(i) = (\omega, \cmd) \\\\ x\in\m{dom}(\tshared) 
\\ x \notin \omega \\ \tshared \lhd \tpriv, \tlocal \vdash e \Downarrow_{r_t} v_t}{
  (\tskctx, \tshared, \tpriv, \tlocal, x := e) 
 \tskStepsto{[r_t]}  (\tskctx, \tshared[x\mapsto v_t], \tpriv, \tlocal, \m{skip}) 
}

\and

\inferrule*[right=TSK-Update-S-Log]{\tskctx = (T, i)
\\ T(i) = (\omega, \cmd) \\\\ x\in\m{dom}(\tshared) 
\\x \in \omega \\ \tshared \lhd \tpriv, \tlocal \vdash e \Downarrow_{r_t} v_t}{
  (\tskctx, \tshared, \tpriv, \tlocal, x := e) 
 \tskStepsto{[r_t]}  (\tskctx, \tshared, \tpriv[x\mapsto v_t], \tlocal, \m{skip}) 
}

\and

\inferrule*[right=TSK-Update-L]{ x\in\m{dom}(\tlocal)  \\ \tshared \lhd \tpriv, \tlocal \vdash e \Downarrow_{r_t} v_t}{
  (\tskctx, \tshared, \tpriv, \tlocal, x := e) 
  \tskStepsto{[r_t]}  (\tskctx, \tshared, \tpriv, \tlocal[x\mapsto v_t], \m{skip}) 
}

\and

\inferrule*[right=TSK-Arr-S]{\tskctx = (T, i)
\\ T(i) = (\omega, \cmd) \\ a\in\m{dom}(\tshared) \\ a \notin \omega 
\\\\\tshared \lhd \tpriv, \tlocal \vdash e \Downarrow_{r_t} v_t 
\\\tshared \lhd \tpriv, \tlocal \vdash e' \Downarrow_{r_t} v_t'}{
  (\tskctx, \tshared, \tpriv, \tlocal, a[e] := e') 
  \tskStepsto{[r_t, r_t']}  (\tskctx, \tshared[a[v_t]\mapsto v_t'], \tpriv, \tlocal,\m{skip}) 
}
\and
\inferrule*[right=TSK-Arr-S-Log]{ \tskctx = (T, i)
\\ T(i) = (\omega, \cmd) \\a\in\m{dom}(\tshared) \\a \in \omega 
\\\\\tshared \lhd \tpriv, \tlocal \vdash e \Downarrow_{r_t} v_t 
\\\tshared \lhd \tpriv, \tlocal \vdash e' \Downarrow_{r_t} v_t'}{
  (\tskctx, \tshared, \tpriv, \tlocal, a[e] := e') 
  \tskStepsto{[r_t, r_t']}  (\tskctx, \tshared, \tpriv[a[v_t]\mapsto v_t'], \tlocal,\m{skip}) 
}
\and

\inferrule*[right=TSK-Arr-L]{ a\in\m{dom}(\tlocal) \\ \tshared \lhd \tpriv, \tlocal \vdash e \Downarrow_{r_t} v_t 
\\\tshared \lhd \tpriv, \tlocal \vdash e' \Downarrow_{r_t'} v_t'}{
  (\tskctx, \tshared, \tpriv, \tlocal, a[e] := e') 
  \tskStepsto{[r_t, r_t']}  (\tskctx, \tshared, \tpriv, \tlocal[a[v_t]\mapsto v_t'], \m{skip}) 
}

\end{mathpar}

\subsection{Well-formedness Checking}
We create a top-level well-formedness judgment for tasks $\tlocal \Vdash_\war T :\m{ok}$. 
A task-based program is well-formed if every task in the program is well-formed. We assume the 
$\tlocal$ set is given. 
\noindent\framebox{$\tlocal \Vdash_\war T:\m{ok}$}
\begin{mathpar}

\inferrule*[right=T-WAR-TSK]{
 \forall i \in \m{dom}(T), T(i) = (\omega, \cmd) 
 \\ \tlocal;\omega;\emptyset;\emptyset \Vdash_\war \cmd: \m{ok}}{
  \tlocal \Vdash_\war T : \m{ok}
}
\end{mathpar}

A task program is only well-formed if the first access to any location
in $\tlocal$ is a write, in addition to the usual checking of WAR
variables, which are task-shared locations. 
In other words, $\tlocal$ variables are not WAR variables as it is
never read before written to. 
We show the changed rule for the WAR variable checking judgment
below. The corresponding array access rule is changed analogously. 

\noindent\framebox{$\tlocal; N; W; R \Vdash_\war \iota: W';R'$}
\begin{mathpar}

\inferrule*[right=WAR-Checkpointed]{ R' = R \cup rd(e) 
\\ x \in R' 
\\ x \notin W
\\ x \in N
\\ x \notin \tlocal
}{
  \tlocal;N; W;R \Vdash_\war x:= e : W \cup x; R' 
}

\end{mathpar}
The rule ensure additionally that it will never be the 
case that a task-local location (e.g., $x$ in this rule) is read
before it's written to. It's trivial to prove the following lemma:
\begin{lem}\label{lem:tlocal-checkpoint}
If $\tlocal;\omega;\emptyset;\emptyset \Vdash_\war \cmd: \m{ok}$, then
$\m{dom}(\tlocal) \cap \omega = \emptyset$.
\end{lem}

\subsection{Translation from Task-based Systems to Checkpoint Systems}

To capture the behavior of of a task transition, we augment the redo-log 
semantics with a $\m{goto}$ command and code context.
\[
\begin{array}{llcl}
\textit{Commands} & \cmd & \bnfdef & \cdots \bnfalt \m{goto}~\ell
\\
\textit{Code context} & \Psi & \bnfdef & \Psi \cdot \bnfalt \ell:\m{checkpoint(\omega);\cmd}
\end{array}
\]
The small step rule for $\m{goto}$ is as follows. The
continuously-powered rule is similar. 
\begin{mathpar}
  \inferrule*[right=RL-Goto]{ \Psi(\ell) = \cmd}{
    \Psi, (\rlctx,\nvmem , \vmem, \m{goto}~\ell;\cmd') 
    \rlStepsto{}  \Psi, (\rlctx,\nvmem , \vmem, \cmd) 
  } 
\end{mathpar}  

We give a translation relation from a task to a redo log program, written $T \trel \Psi$. 

\begin{mathpar}
  \inferrule*[ ]{ }{
    \cdot \trel \cdot
  } 
\and
  \inferrule*[ ]{ T \trel \Psi \\ \llbracket \cmd_t \rrbracket = \cmd_r 
  }{
    T, i \mapsto (\omega, \cmd_t) \leadsto \Psi, \ell:\m{checkpoint}(\omega);\cmd_r
  } 

  \and
  \inferrule*[ ]{  }{
    \llbracket \m{toTask}(i) \rrbracket = \m{goto}~\ell_i
  } 
\and
  \inferrule*[ ]{ \iota \neq \llbracket \m{toTask}(i) \rrbracket }{
    \llbracket \iota \rrbracket = \iota
  } 
\and
  \inferrule*[ ]{  }{
    \llbracket \iota;\cmd \rrbracket = \llbracket \iota\rrbracket; \llbracket\cmd \rrbracket
  }   
\and
  \inferrule*[ ]{ }{
    \llbracket \ifthen{e}{\cmd_1}{\cmd_2} \rrbracket = \ifthen{e}{\llbracket \cmd_1 \rrbracket}{\llbracket \cmd_2 \rrbracket}
  } 
\end{mathpar}  

Finally, we extend the behavior of assignments to $\tlocal$ 
to capture whether or not a variable has been initialized. Each write to task-local memory writes a pair $(b,v)$, where 
$v$ is the value, as before, and $b$ is a bit $\in \{0,1\}$ indicating if the value has been written to. Assignments 
flip the bit of a location to 1, and resetting volatile memory on
power failure sets all bits to 0; that is every volatile location
$\loc$ becomes $\loc \mapsto (0, \mt{reset}())$ after reset.
\begin{mathpar} 
\inferrule*[right=TSK-Update-L]{ x\in\m{dom}(\tlocal)  \\ \tshared \lhd \tpriv, \tlocal \vdash e \Downarrow_{r_t} v_t}{
  (\tskctx, \tshared, \tpriv, \tlocal, x := e) 
  \tskStepsto{[r_t]}  (\tskctx, \tshared, \tpriv, \tlocal[x\mapsto (1,v_t)], \m{skip}) 
}
\end{mathpar}

We use this to define a relation between $\tlocalv$ and $\vmem$, the memory of each 
system that resides in the volatile memory, $\tlocalv \approx \vmem$. The relation holds if as long as a task-local variable has been initialized, it will equal the value in the 
redo-log system's volatile memory.

\begin{mathpar}
\inferrule*[]{ }{
  \tlocalv \approx \vmem
}
\and
\inferrule*[]{ \tlocalv \approx \vmem}{
  \tlocalv, \loc \mapsto (0, v_t) \approx \vmem, \loc \mapsto v_r
}
\and
\inferrule*[]{ \tlocalv \approx \vmem \\ v_t = v_r }{
  \tlocalv, \loc \mapsto (1, v_t) \approx \vmem, \loc \mapsto v_r
}
\end{mathpar}

We are now ready to define a binary relation between a task program and a redo-log program. 
written $\tstate \rel \Psi, \rstate$
\[
\inferrule*{
\rstate = (\rlctx, \nvmem_r, \vmem_r, \cmd_r)
\\ \tstate = (\tskctx, \tshared, \tpriv, \tlocal, \cmd_t) 
\\\\ \rcon = (\ulog, \vmem_c, \cmd_c, \omega_r)
\\ \tcon = (T, i) \\ T(i) = (\omega_t, \cmd_\mt{tt})
\\\\ T \trel \Psi
\\ \llbracket \cmd_t \rrbracket = \cmd_r
\\ \llbracket \cmd_\mt{tt} \rrbracket = \cmd_c
\\ \omega_t = \omega_r
\\ \tpriv = \ulog
\\ \tlocal = \tlocalv, \tlocaln
\\ \nvmem_r = \tshared \cup \tlocaln
\\ \vmem_r \approx \tlocalv
\\ \m{dom}(\vmem_r) \subseteq \dom(\tlocalv) 
}{
\tstate \rel \Psi, \rstate
}
\]
To aid in proving bi-simulation, we first prove a helper lemma that states that expression 
evaluation is equivalent between a related task and redo-log system.
\begin{lem}[Related undo-log and task expression evaluation are the same]\label{lem:eval-equiv}
  Given $\rstate=(\rlctx, \nvmem_r, \vmem_r, \cmd_r)$, $\tstate=(\tskctx, \tshared, \tpriv, \tlocal, \cmd_t) $,  and 
  an expression $e$ s.t. $\tstate \rel \Psi, \rstate$, $\forall\loc\in\mt{rd}(e)\cap
  \m{dom}(\tlocal)$, $\tlocal(\loc) = (1, \_)$ 
and $\tshared\lhd\tpriv, \tlocal \vdash e \Downarrow_{r_t} v_t$, and 
  $\nvmem_r\lhd\ulog, \vmem_r \vdash e \Downarrow_{r_r} v_r$, then $r_r = r_t$ and $v_r = v_t$  
\end{lem}
\begin{proof}

  By induction over the structure of $e$. 
  \begin{description} 
  \item[Case:] $e$ = $x$  
  \begin{tabbing}
  \\  We examine two subcases: (I) $x \in \m{dom}(\nvmem_r)$ and (II) $x \in \m{dom}(\vmem_r)$ 
  \\{\bf Subcase (I)}: $x \in \m{dom}(\nvmem_r)$
  \\ By assumption, $\nvmem_r = \tshared \cup \tlocaln$ and $\tpriv=\ulog$
  \\\qquad \=  (I1) ~~\=  $\nvmem_r\lhd\ulog, \vmem_r
  =\tshared\lhd\tpriv, \tlocal $
  \\ By (I1)
 \\\> (I2)\> $r_t = r_r = \m{rd}\ x\ v_t$ and $v_t = v_r = (\nvmem_r\lhd\ulog, \vmem_r)(x)$
  
  \\{\bf Subcase (II)}: $x \in \m{dom}(\vmem_r)$
  \\By assumption
  \\\>(II1) $x \in \dom(\tlocalv)$
  \\By assumption
     \\\>(II2)   $\tlocalv(x) = (1, v_t)$
  \\By (II2) and assumption that $\vmem_r \approx \tlocalv$
  \\\>(II3)\> $\nvmem_r(x) = v_r = v_t$
  \\By (II3) and rule \rulename{Rd-Var} 
  \\\> (II4)\> $r_t = r_r = \m{rd}\ x\ v_t$
  \end{tabbing}

  \item[Case:] $e$ = $a[e']$
  \\We apply the I.H. to $e'$ 
  \\This is similar to the previous case. Neither $a$ or variables in
  $e'$ can differ in $\tlocal$ and $\vmem_r$.
  
  \item[Case:] $e$ = $e_1 \m{bop} e_2$ 
  \\We apply the I.H. to $e_1$ and $e_2$ 
  \\By rule \rulename{BinOp} and expression evaluation is deterministic, $v_t = v_r$ and $r_t = r_r$

\end{description}
\end{proof}  

\begin{lem}[Redo logging simulates task-based
  system]\label{lem:sim-undo-task}
$\vdash_\war \tstate: \m{ok}$, 
$\tstate \rel \Psi, \rstate $ and $\tstate \tskStepsto{o_1} \tstate' $
then $\exists \rstate'$ s.t. $\rstate \rlMStepsto{o_2}\rstate'$ and
$\tstate' \rel \Psi, \rstate' $and $o_1 = o_2$.
\end{lem}

\begin{proof}

  By induction over $\ee::\tstate \tskStepsto{o_1} \tstate' $
  
  \begin{description}
  \item[Cases:] $\ee$ ends in \rulename{TSK-Skip}, \rulename{TSK-If-T}
    \rulename{TSK-If-F} rule. The non-volatile memory and task-shared data is not altered
    data is not altered. The redo logging system
    takes corresponding \rulename{RL-Skip}, \rulename{RL-If-T}
    \rulename{RL-If-F} to reach a related configuration.

  \item[Case:] $\ee$ ends in \rulename{TSK-Seq} rule. We apply
    I.H. directly. 
  
  \item[Case:] $\ee$ ends in \rulename{TSK-Update-S} rule  
  \begin{tabbing}
  By assumption
  \\\qquad \= (1)\quad \= $\tstate \rel \Psi, \ustate$
  \\ By  \rulename{TSK-Update-S} rule
  \\\>(2)\> $(\tskctx, \tshared, \tpriv, \tlocal, x := e) 
  \tskStepsto{[r_t]}  (\tskctx, \tshared[x\mapsto v_t], \tpriv,
  \tlocal, \m{skip})$
 \\\> (3)\> $T(i) = (\omega_t, \cmd_t) $, $x\notin \omega_t$, $x\in\m{dom}(\tshared)$
  \\\> (4)\> $\tshared\lhd\tpriv, \tlocal \vdash e \Downarrow_r v_t$
\\ By (1)
\\\>(5)\> $\rcon = (\ulog, \vmem_c, \cmd_c, \omega_r)$,
$\omega_r=\omega_t$, $\nvmem_r= \tshared, \tlocaln$
\\By (3), (5)
    \\\>(6)\>   $x\notin \omega_r$, $x\in\m{dom}(\nvmem_r)$
  \\ By  \rulename{RL-NV-Assign} rule
    \\\>(7)\> $\Psi, (\rcon, \nvmem_r, \vmem_r, x:=e) 
    \rlStepsto{[r_r]}  \Psi, (\rcon, \nvmem_r[x\mapsto v_r], \vmem_r,\m{skip}) $
    \\\> (8)\> $\nvmem_r \lhd \ulog, \vmem_r \vdash e \Downarrow_{r_r}
    v_r$
\\ By $\vdash_\war \tstate: \m{ok}$, 
\\\>(9)\> $\forall\loc\in\mt{rd}(e)\cap \tlocal$, $\tlocal(\loc) =(1,
\_)$
    \\By Lemma~\ref{lem:eval-equiv} and (9)
    \\\> (10)\> $v_t = v_r$ and $r_t = r_r$ 
    \\By (10) 
    \\\>(11)
  $\tshared[x\mapsto v_t] = \nvmem_r[x\mapsto v_r]$
    \\By (11) and  volatile memory and log do not change, $\tstate' \rel \Psi, \rstate'$
  \end{tabbing}

  \item[Case:] $\ee$ ends in \rulename{TSK-Update-S-Log} rule  
\begin{tabbing}
  By assumption
  \\\qquad \= (1)\quad \= $\tstate \rel \Psi, \ustate$
  \\ By  \rulename{TSK-Update-S-Log} rule
  \\\>(2)\> $(\tskctx, \tshared, \tpriv, \tlocal, x := e) 
  \tskStepsto{[r_t]}  (\tskctx, \tshared, \tpriv[x\mapsto v_t],
  \tlocal, \m{skip})$
 \\\> (3)\> $T(i) = (\omega_t, \cmd_t) $, $x\in \omega_t$, $x\in\m{dom}(\tshared)$
  \\\> (4)\> $\tshared\lhd\tpriv, \tlocal \vdash e \Downarrow_r v_t$
\\ By (1)
\\\>(5)\> $\rcon = (\ulog, \vmem_c, \cmd_c, \omega_r)$,
$\omega_r=\omega_t$, $\nvmem_r= \tshared, \tlocaln$
\\By (3), (5)
    \\\>(6)\>   $x\in \omega_r$, $x\in\m{dom}(\nvmem_r)$
  \\ By  \rulename{RL-NV-Log} rule
    \\\>(7)\> $\Psi, (\rcon, \nvmem_r, \vmem_r, x:=e) 
    \rlStepsto{[r_r]} \Psi, (\rcon', \nvmem_r, \vmem_r,\m{skip}) $ where 
  $\rcon' = (\ulog [x\mapsto v_r], \vmem_c, \cmd_c, \omega_r)$
    \\\> (8)\> $\nvmem_r \lhd \ulog, \vmem_r \vdash e \Downarrow_{r_r}
    v_r$
\\ By $\vdash_\war \tstate: \m{ok}$, 
\\\>(9)\> $\forall\loc\in\mt{rd}(e)\cap \tlocal$, $\tlocal(\loc) =(1,
\_)$
    \\By Lemma~\ref{lem:eval-equiv} 
    \\\> (10)\> $v_t = v_r$ and $r_t = r_r$ 
    \\By (10) 
    \\\>(11)
  $\tpriv[x\mapsto v_t] = \ulog[x\mapsto v_r]$
    \\By (11) and other parts of the state don't change, $\tstate' \rel \Psi, \rstate'$
  \end{tabbing}

  \item[Case:] $\ee$ ends in \rulename{TSK-Update-L} rule  
  \begin{tabbing}
  By assumption
  \\\qquad \= (1)\quad \= $\tstate \rel \Psi, \rstate$
  \\ By  \rulename{TSK-Update-L} rule
  \\\>(2)\> $(\tskctx, \tshared, \tpriv, \tlocal, x := e) 
  \tskStepsto{[r_t]}  (\tskctx, \tshared, \tpriv, \tlocal[x\mapsto v_t], \m{skip})$
  \\\> (3)\> $\tshared \lhd \tpriv, \tlocal \vdash e \Downarrow_{r_t}
  v_t$
\\ By (1)
\\\>(4)\> $\rcon = (\ulog, \vmem_c, \cmd_c, \omega_r)$,
 $\tpriv=\ulog$, $\tlocalv=\vmem_r$
$\omega_r=\omega_t$, $\nvmem_r= \tshared, \tlocaln$
  \\We consider two subcases: (I) $x \in \tlocaln$ or (II) $x \in \tlocalv$  
  \\ {\bf Subcase (I)} $x \in \tlocaln$
  \\ By Lemma~\ref{lem:tlocal-checkpoint} and $\vdash_\war \tstate: \m{ok}$, 
  \\\>(I1)\> $x \notin \omega_t$ 
\\By (4) and (I1)
\\\>(I2)\>  $x \notin \omega_r$ 
  \\ By  \rulename{RL-NV-Assign} rule
  \\\>(I3)\> $\Psi, (\rcon, \nvmem_r, \vmem_r, x:=e) 
  \rlStepsto{[r_r]} \Psi, (\rcon, \nvmem_r[x\mapsto v_r], \vmem_r,\m{skip}) $
  \\\> (I4)\> $\nvmem_r \lhd \ulog, \vmem_r \vdash e \Downarrow_{r_r}
  v_r$
\\ By $\vdash_\war \tstate: \m{ok}$, 
\\\>(I5)\> $\forall\loc\in\mt{rd}(e)\cap \tlocal$, $\tlocal(\loc) =(1,
\_)$
  \\By Lemma~\ref{lem:eval-equiv} 
  and expression evaluation is deterministic
  \\\> (I6)\> $v_t = v_r$ and $r_t = r_r$ 
  \\By (I6) 
  \\\>(I7)\>$\tlocaln[x\mapsto v_t] = \nvmem_r[x\mapsto v_r]$
  \\By (I7), volatile memory and log do not change, $\tstate' = \rstate'$
  \\ {\bf Subcase (II)} $x \in \tlocalv$
  \\ By  \rulename{RL-Assign-V} rule
  \\\>(II1)\> $\Psi, (\rcon, \nvmem_r, \vmem_r, x:=e) 
  \rlStepsto{[r_r]} \Psi, (\ucon, \nvmem_r, \vmem_r[x\mapsto v_r],  \m{skip}) $
  \\\> (II2)\> $\nvmem_r \lhd \ulog, \vmem_r \vdash e \Downarrow_{r_r}
  v_r$
\\ By $\vdash_\war \tstate: \m{ok}$, 
\\\>(II3)\> $\forall\loc\in\mt{rd}(e)\cap \tlocal$, $\tlocal(\loc) =(1,
\_)$
  \\By Lemma~\ref{lem:eval-equiv} and (II3)
  and expression evaluation is deterministic
  \\\>(II4)\> $v_t= v_r$, $r_t=r_r$, 
  \\ By (II4) 
  \\\> (II5)\>  $\tlocalv[x\mapsto (1,v_t)] = \vmem_r[x\mapsto v_r]$
  \\ By (II5) and task-shared, non-volatile memories, and log do not change, $\tstate' \rel \Psi, \rstate'$

  \end{tabbing}

  \item[Case:] $\ee$ ends in \rulename{TSK-Update-Arr-S} or
    \rulename{TSK-Update-Arr-S-Log}, or
    \rulename{TSK-Update-Arr-L}. These proofs are similar the previous three
    cases.   

  \item[Case:] $\ee$ ends in \rulename{TSK-Trans} rule
  \begin{tabbing}
  By assumption
  \\\qquad \= (1)~~ \= $\tstate \rel \Psi, \rstate$
  \\ By  \rulename{TSK-Trans} rule
  \\\>(2)\> $(\tskctx, \tshared, \tpriv, \tlocal, \m{toTask}(j)) 
  \tskStepsto{\m{checkpoint}}  ((T, \m{task}(j)), \tshared \lhd \tpriv, \emptyset, \tlocal, \cmd_t) $
  \\\>(3)\> $T(j)= (\omega, \cmd_t)$
  \\ By assumption, 
  \\\>(4)\> $\llbracket\m{toTask}(j)\rrbracket = \m{goto}~\ell_j$ 
  \\\>(5)\> 
  $\Psi(\ell_j)=\m{checkpoint}(\omega);\cmd_r$ and $\llbracket\cmd_t\rrbracket=\cmd_r$
  \\ By (4) and \rulename{RL-Goto} rule
  \\\>(6)\> $\Psi, (\rlctx, \nvmem_r, \vmem_r, \m{goto}~\ell_j) 
  \rlStepsto{} \Psi, (\rlctx, \nvmem_r, \vmem_r, \m{checkpoint}(\omega);\cmd_r)$
  \\ By (6), (5), and \rulename{RL-Commit} rule
  \\\>(7)\> $\Psi, ((\ulog, \vmem_c, \cmd_c, \omega'), \nvmem_r, \vmem_r, \m{checkpoint}(\omega);\cmd_r) 
  \rlStepsto{\m{checkpoint}} \Psi, ((\emptyset, \vmem_r, \cmd_r, \omega), \nvmem_r \lhd \ulog, \vmem_r, \cmd_r)$
  \\By assumption that $\ulog = \tpriv$ and $\nvmem_r = \tshared \cup \tlocaln$
  \\\> (8) $\tshared \lhd \tpriv \cup \tlocaln = \nvmem_r \lhd \ulog$
  \\By (8) volatile memory does not change, and the initial contexts
  relate to each other
\\\> (9)\> the observation is equivalent and $\tstate' \rel \Psi, \rstate'$.
  \end{tabbing}
  
  \item[Case:] $\ee$ ends in \rulename{TSK-reboot} rule  
  \begin{tabbing}
  By assumption
  \\\qquad \= (1)~~ \= $\tstate \rel \Psi, \rstate$
  \\ By  \rulename{TSK-reboot} rule
  \\\>(2)\> $(\tskctx, \tshared, \tpriv, \tlocal, \m{reboot}) 
  \tskStepsto{\m{reboot}}  (\tskctx, \tshared, \emptyset, \tlocal,
  \cmd_t) $
\\\>(3)\> $\tskctx=(T, i)$ and $T(i)= (\omega_t, \cmd_t)$
  \\By \rulename{RL-reboot} rule
  \\\>(4)\> $\Psi, ((\ulog, \vmem_c, \cmd_c, \omega), \nvmem_r, \vmem_r, \m{reboot}()) 
  \rlStepsto{\m{reboot}}  \Psi, ((\emptyset, \vmem_c, \cmd_c, \omega_c),
  \nvmem_r, \vmem_c, \cmd_c)$
  \\By (1)
\\\>(5)\> $\cmd_c=\llbracket\cmd_t\rrbracket$, $\omega_c=\omega_t$,
and $\tshared,\tlocaln=\nvmem_r$
  \\By power failure must precede reboot, and the definition of $\mt{resetVol}$
  \\\> (6)\> $\forall \loc \in \tlocalv, \tlocalv[\loc \mapsto (0, v_t)]$
  \\By (6)
  \\\>(7) $\vmem_c \approx \tlocalv$ 
  \\By (5), (6), and (7) and the new log and the $\tpriv$ are both empty, $\tstate' \rel \Psi, \rstate'$
  
  \end{tabbing}
  
  \end{description}
  \end{proof}

Here, we write $\Psi,\rstate \rlStepsto{o_1}^\dagger \Psi,\rstate' $ to include
all one-step transitions, except \rulename{RL-Commit} and include
instead \rulename{RL-Goto} followed by \rulename{Checkpoint} as one
atomic step. This is reasonable as these correspond to an atomic step in the task
setting. Further, there are no checkpoints in commands other than
those translated from $\m{toTask}$. 
  \begin{lem}[Task-based system simulates translated continuous
    program]\label{lem:sim-redo-task}
$\vdash_\war \tstate: \m{ok}$, 
    $\tstate \rel \Psi, \rstate $ and $\Psi,\rstate \rlStepsto{o_1}^\dagger \Psi,\rstate' $
    then $\exists \tstate'$ s.t. $\tstate \tskStepsto{o_2}\tstate'$ and
    $\tstate' \rel \Psi, \rstate' $and $o_1 = o_2$.
    \end{lem}
    
    \begin{proof}
    
      By induction over $\ee::\rstate \tskStepsto{o_1} \rstate' $
      
      \begin{description}
      \item[Cases:] $\ee$ ends in \rulename{RL-Skip}, \rulename{RL-If-T}
        \rulename{RL-If-F} rule. The memory
        is not altered. The task system
        takes corresponding \rulename{TSK-Skip}, \rulename{TSK-If-T}
        \rulename{TSK-If-F} to reach a related configuration.

      \item[Case:] $\ee$ ends in \rulename{RL-Seq} rule. We apply
        I.H. directly. 
      
      \item[Case:] $\ee$ ends in \rulename{RL-NV-Assign} rule  
      \begin{tabbing}
      By assumption
      \\\qquad \= (1)\quad \= $\tstate \rel \Psi, \rstate$
      \\ By  \rulename{RL-NV-Assign} rule
      \\\>(2)\> $\Psi, (\rcon, \nvmem_r, \vmem_r, x:=e) 
        \rlStepsto{[r_r]} \Psi, (\rcon, \nvmem_r[x\mapsto v_r], \vmem_r,\m{skip}) $
      \\\> (3)\> $\nvmem_r \lhd \ulog, \vmem_r \vdash e
      \Downarrow_{r_r} v_r$ and  $x\notin \omega_r$,
      $x\in\m{dom}(\nvmem_r)$
      \\By (1)
      \\\>(4)\> $\nvmem_r= \tshared, \tlocaln$, $\omega_r = \omega_t$,
      $\tpriv=\ulog$, $\vmem_r=\tlocalv$
      \\By (4)
      \\\> (5)\> $\tshared \lhd \tpriv, \tlocal \vdash e
      \Downarrow_{r_t} v_t$
      \\ By $\vdash_\war \tstate: \m{ok}$, 
      \\\>(6)\> $\forall\loc\in\mt{rd}(e)\cap \tlocal$, $\tlocal(\loc) =(1,
\_)$
      \\By assumption,(4), (5), (6)  and Lemma~\ref{lem:eval-equiv} 
      \\\> (7)\> $v_t = v_r$ and $r_t = r_r$ 

      \\We consider two subcases: (I) $x \in \m{dom}(\tshared)$ or (II) $x \in \m{dom}(\tlocaln)$
      \\ {\bf Subcase (I)} $x \in \m{dom}(\tshared)$
      \\ By  \rulename{TSK-Update-S} rule
      \\\>(I1)\> $(\tskctx, \tshared, \tpriv, \tlocal, x := e) 
      \tskStepsto{[r_t]}  (\tskctx, \tshared[x\mapsto v_t], \tpriv, \tlocal, \m{skip})$
      \\By (7) and (4) 
      \\\>(I2)$\tshared[x\mapsto v_t] = \nvmem_r[x\mapsto v_r]$
      \\By (I3), volatile memory and log do not change, $\tstate' = \rstate'$
      \\ {\bf Subcase (II)} $x \in \m{dom}(\tlocaln)$
      \\ By  \rulename{TSK-Update-L} rule
      \\\>(II1)\> $(\tskctx, \tshared, \tpriv, \tlocal, x := e) 
      \tskStepsto{[r_t]}  (\tskctx, \tshared, \tpriv, \tlocal[x\mapsto v_t], \m{skip})$
      \\By (7) 
      \\\>(II2)$\tlocaln[x\mapsto v_t] = \nvmem_r[x\mapsto v_r]$
      \\By (II2), volatile memories and log do not change, $\tstate' \rel \Psi, \rstate'$
      \end{tabbing}
    
      \item[Case:] $\ee$ ends in \rulename{RL-NV-Log} rule  
      \begin{tabbing}
      By assumption
      \\\qquad \= (1)\quad \= $\tstate \rel \Psi, \rstate$
      \\ By  \rulename{RL-NV-Log} rule
      \\\>(2)\>\ $\Psi, ((\ulog, \vmem_c,\cmd_c, \omega_r), \nvmem_r, \vmem_r, x:=e) 
      \rlStepsto{[r_r]}  \Psi, ((\ulog[x \mapsto v_r], \vmem_c,\cmd_c, \omega), 
        \nvmem_r, \vmem_r,\m{skip}) $
      \\\>(3)\>\ $\nvmem_r \lhd \ulog, \vmem_r \vdash e
      \Downarrow_{r_r} v_r$, $x\in\omega_r$, $x\in\m{dom}(\nvmem_r)$
      \\By (1)
      \\\>(4)\> $\nvmem_r= \tshared, \tlocaln$, $\omega_r = \omega_t$,
      $\tpriv=\ulog$, $\vmem_r=\tlocalv$
      \\By (3) and (4)
      \\\>(5)\> $x\in\omega_t$ and $x\in\tshared,\tlocaln$
      \\By Lemma~\ref{lem:tlocal-checkpoint} and $\vdash_\war \tstate: \m{ok}$, 
      \\\>(6)\> $x \notin \tlocaln$ 
      \\ By  \rulename{TSK-Update-S-Log} rule
      \\\>(7)\> $(\tskctx, \tshared, \tpriv, \tlocal, x := e) 
      \tskStepsto{[r_t]}  (\tskctx, \tshared, \tpriv[x\mapsto v_t], \tlocal, \m{skip}) $
        \\\> (8) $\tshared \lhd \tpriv, \tlocal \vdash e
        \Downarrow_{r_t} v_t$
        \\ By $\vdash_\war \tstate: \m{ok}$, 
        \\\>(9)\> $\forall\loc\in\mt{rd}(e)\cap \tlocal$, $\tlocal(\loc) =(1,\_)$
        \\ By assumption, (4), (8), (9) and Lemma~\ref{lem:eval-equiv}
        \\\> (10)\>\ $v_t = v_r$ and $r_t = r_r$ 
        \\By (4) and (10)
        \\\>(11)\>~$\tpriv[x\mapsto v_t] = \ulog[x\mapsto v_r]$
        \\By (11), $\tshared, \tlocal$ and $\nvmem_r, \vmem_r$ does not change and 
        observation is the same,
         $\tstate' \rel  \Psi, \rstate'$
      \end{tabbing}

      \item[Case:] $\ee$ ends in \rulename{RL-V-Assign} rule. The case
        is similar to the previous cases. Here, the corresponding
        \rulename{TSK-Update-L} rule is applied.

      \item[Case:] $\ee$ ends in \rulename{RL-Assign-Arr} or
        \rulename{RL-Arr-Log} or \rulename{RL-V-Assign-Arr} rule, the proofs are similar to the
        previous three cases.
    
      \item[Case:] $\ee$ ends in \rulename{RL-Goto} 
 followed by \rulename{RL-Commit} rule
      \begin{tabbing} 
      By assumption
      \\\qquad \= (1)~~ \= $\tstate \rel \Psi, \rstate$
      \\ By \rulename{RL-Goto} rule
      \\\>(2)\> $\Psi, (\rlctx, \nvmem_r, \vmem_r, \m{goto}~\ell_i) 
      \rlStepsto{} \Psi, (\rlctx, \nvmem_r, \vmem_r,
      \m{checkpoint}(\omega_r);\cmd_r)$
\\\>(3)\> $\Psi(\ell_i) = \m{checkpoint}(\omega_r);\cmd_r$

      \\ By (2) and \rulename{RL-Commit} rule
      \\\>(4)\> $\Psi, ((\ulog, \vmem_c, \cmd_c, \omega'), \nvmem_r, \vmem_r, \m{checkpoint}(\omega_r);\cmd_r) 
      \rlStepsto{\m{checkpoint}} \Psi, ((\emptyset, \vmem_r, \cmd_r,
      \omega_r), \nvmem_r \lhd \ulog, \vmem_r, \cmd_i)$
\\By (1)
\\\>(5)\> $\llbracket\m{toTask}(i)\rrbracket = \m{goto}~\ell_i$ and
$T(i)= (\omega_t, \cmd_t)$, $\omega_r=\omega_t$, and $\llbracket\cmd_t\rrbracket=\cmd_r$ 
      \\ By  \rulename{TSK-Trans} rule
      \\\>(6)\> $(\tskctx, \tshared, \tpriv, \tlocal, \m{toTask}(i)) 
      \tskStepsto{\m{checkpoint}}  ((T, i), \tshared \lhd \tpriv,
\emptyset, \tlocal, \cmd_t) $ 
      \\ By (4) and (5)
      \\\> (7)\> $\ulog' = \emptyset = \tpriv'$ 
      \\By assumption that $\ulog = \tpriv$ and $\tshared \cup \tlocaln = \nvmem_r$
      \\\> (8) $\tshared \lhd \tpriv \cup \tlocaln = \nvmem_r \lhd \ulog$
      \\By (7), (8), volatile memory does not change, and observation is equivalent, $\tstate' \rel \Psi, \rstate'$.
      \end{tabbing}

      \item[Case:] $\ee$ ends in \rulename{RL-reboot} rule  
  \begin{tabbing}
  By assumption
  \\\qquad \= (1)~~ \= $\tstate \rel \Psi, \rstate$
  \\By \rulename{RL-reboot} rule
  \\\>(2)\> $\Psi, ((\ulog, \vmem_c, \cmd_c, \omega_c), \nvmem_r, \vmem_r, \m{reboot}()) 
  \rlStepsto{\m{reboot}}  \Psi, ((\emptyset, \vmem_c, \cmd_c, \omega_c),
  \nvmem_r, \vmem_c, \cmd_c)$
\\\>(3)\> $\tskctx=(T, i)$ and $T(i)= (\omega_t, \cmd_t)$
  \\ By  \rulename{TSK-reboot} rule
  \\\>(4)\> $(\tskctx, \tshared, \tpriv, \tlocal, \m{reboot}) 
  \tskStepsto{\m{reboot}}  (\tskctx, \tshared, \emptyset,
  \tlocal, \cmd_t) $
\\By (1)
\\\>(5)\> $\cmd_c=\llbracket\cmd_t\rrbracket$, $\omega_c=\omega_t$,
and $\tshared,\tlocaln=\nvmem_r$
  \\ By power failure must precede reboot, and $\mt{resetVol}$
  \\\> (6)\> $\forall \loc \in \tlocalv, \tlocalv[\loc \mapsto (0, v_t)]$
  \\By (6)
  \\\>(7) $\vmem_c \approx \tlocalv$
  \\By (5), (6), and (7) and the new log and the $\tpriv$ are both empty, $\tstate' \rel \Psi, \rstate'$ 
  \end{tabbing}
      \end{description}
      \end{proof}

\section{Intermittent Computing with Inputs}

\subsection{Summary of Operational Semantics}

\noindent\framebox{$(\timestamp, \context, \nvmem, \vmem, \cmd) 
 \Stepsto{O} (\timestamp',\context', \nvmem', \vmem', \cmd')$}

\begin{mathpar}

\inferrule*[right=I/O-CP-PowerFail]{ \m{pick}(n)}{
  (\timestamp, \context, \nvmem, \vmem, \cmd) 
 \Stepsto{}  (\timestamp + 1, \context, \nvmem, \resetm(\vmem), \m{reboot}(n)) 
}
\and
\inferrule*[right=I/O-CP-CheckPoint]{ }{
  (\timestamp, \context, \nvmem, \vmem, \m{checkpoint}(\omega);\cmd) 
 \Stepsto{\m{checkpoint}}  (\timestamp +1, (\proj{\nvmem}{\omega}, \vmem, \cmd), \nvmem, \vmem, \cmd) 
}
\and
\inferrule*[right=I/O-CP-Reboot]{ \context=(\nvmem, \vmem, \cmd)}{
  (\timestamp, \context, \nvmem', \vmem', \m{reboot}(n)) 
 \Stepsto{\m{reboot}}  (\timestamp + n, \context, \nvmem'\lhd\nvmem, \vmem, \cmd) 
}
\and
\inferrule*[right=I/O-CP-V-Assign-I/O]{x \in \m{dom}(V)}{
  (\timestamp, \context, \nvmem, \vmem, x:=\m{IN}()) 
 \Stepsto{\m{in}(\timestamp)}  (\timestamp +1, \context, \nvmem, \vmem[x\mapsto \m{in}(\timestamp)], \m{skip}) 
}
\and
\inferrule*[right=I/O-NV-Assign]{x \in \m{dom}(N)\\N, V \vdash e \Downarrow_{r} v }{
  (\timestamp, \context, \nvmem, \vmem, x:=e) 
 \Stepsto{[r]}  (\timestamp + 1, \context, \nvmem[x\mapsto v], \vmem, \m{skip}) 
}
\and
\inferrule*[right=I/O-CP-Assign-Arr]{N, V \vdash e \Downarrow_{r} v \\N, V \vdash e' \Downarrow_{r'} v' }{
  (\timestamp, \context, \nvmem, \vmem, a[e]:=e') 
 \Stepsto{[r, r']}  (\timestamp +1, \context, \nvmem[a[v]\mapsto v'], \vmem, \m{skip}) 
}

\end{mathpar}
\begin{mathpar}
\inferrule*[right=I/O-CP-V-Assign]{x \in \m{dom}(V)\\N, V \vdash e \Downarrow_{r} v }{
  (\timestamp, \context, \nvmem, \vmem, x:=e) 
 \Stepsto{[r]}  (\timestamp +1, \context, \nvmem, \vmem[x\mapsto v], \m{skip}) 
}

\and
\inferrule*[right=I/O-CP-NV-Assign-In]{x \in \m{dom}(N) }{
  (\timestamp, \context, \nvmem, \vmem, x:=\m{IN}()) 
 \Stepsto{\m{in}(\timestamp)}  (\timestamp + 1, \context, \nvmem[x\mapsto \m{in}(\timestamp)], \vmem, \m{skip}) 
}
\and
\inferrule*[right=I/O-CP-Assign-Arr-In]{N, V \vdash e \Downarrow_{r} v }{
  (\timestamp, \context, \nvmem, \vmem, a[e]:=\m{IN}()) 
 \Stepsto{[r], \m{in}(\timestamp)}  (\timestamp +1, \context, \nvmem[a[v]\mapsto \m{in}(\timestamp)], \vmem, \m{skip}) 
}
\and
\inferrule*[right=I/O-CP-Skip]{ }{
  (\timestamp, \context, \nvmem, \vmem, \m{skip};\cmd) 
 \Stepsto{}  (\timestamp, \context, \nvmem, \vmem, \cmd) 
}
\and
\inferrule*[right=I/O-CP-Seq]{  (\timestamp, \context, \nvmem, \vmem, i) 
 \Stepsto{o}  (\timestamp + 1, \context, \nvmem', \vmem', \m{skip})}{
  (\timestamp, \context, \nvmem, \vmem, i;c) 
 \Stepsto{o}  (\timestamp + 1, \context, \nvmem', \vmem', c) 
}
\and
\inferrule*[right=I/O-CP-If-T]{N, V \vdash e \Downarrow_{r} \m{true}}{
  (\timestamp, \context, \nvmem, \vmem, \m{if}\ e\ \m{then}\ \cmd_1\ \m{else}\ \cmd_2) 
 \Stepsto{[r]}  (\timestamp + 1, \context, \nvmem, \vmem, \cmd_1) 
}
\and
\inferrule*[right=I/O-CP-If-F]{N, V \vdash e \Downarrow_{r} \m{false}}{
  (\timestamp, \context, \nvmem, \vmem, \m{if}\ e\ \m{then}\ \cmd_1\ \m{else}\ \cmd_2) 
 \Stepsto{[r]}  (\timestamp+1, \context, \nvmem, \vmem, \cmd_2) 
}
\end{mathpar}

\subsection{Checking RIO Variables}
\label{app:rio-checking}
We write $N;I;M \Vdash_\rio \iota: I';M'$ to mean 
that given the version set $N$, the input-dependent set $I$, write-set $M$,
 and control is not tainted, any
input-dependent variables in $\iota$ are in $I'$ and any variables written 
are in $M$. We write $N;M \Vtaint \iota:\m{ok}$ to mean 
that given the version set $N$, the must-write set $M$, and control is tainted, 
any variable updated in $\iota$ that is not in $M$ must be in $N$. Any array updated must be in $N$.
input-dependent variables in $\iota$ are in $I'$ 
~\\\\
\noindent\framebox{$N;I; M \Vdash_\rio \iota: I', M'$}

\begin{mathpar}
  \inferrule*[right=RIO-dep]{ I \cap \mt{rd}(e) \neq \emptyset }{
    N;I;M \Vdash_\rio x:= e : I \cup x; M \cup x
  }
  \and
  \inferrule*[right=RIO-dep-clear]{ I \cap \mt{rd}(e) = \emptyset \\ x \in I}{
    N;I;M \Vdash_\rio x:= e : I \setminus x;M \cup x
  }
  \and

  \inferrule*[right=RIO-NDep]{ I \cap \mt{rd}(e) = \emptyset }{
    N;I;M \Vdash_\rio x:= e : I; M \cup x
  }
  \and
  \inferrule*[right=RIO-Arr-NDep]{ I \cap (\mt{rd}(e)\cup \mt{rd}(e')) =\emptyset}{
    N;I;M \Vdash_\rio a[e] := e' : I ;M
  }
  \and
  \inferrule*[right=RIO-Arr-dep]{ I \cap \mt{rd}(e') \neq \emptyset}{
    N;I;M \Vdash_\rio a[e] := e' : I \cup a;M
  }
  \and
  \inferrule*[right=RIO-Get]{ }{
    N;I;M \Vdash_\rio x:= \m{IN}() :I \cup x;M \cup x 
  }
  \and
  \inferrule*[right=RIO-Arr-loc]{ I \cap \mt{rd}(e) \neq \emptyset \\ a \in N}{
    N;I;M \Vdash_\rio a[e]:= e' : I \cup a;M 
  }
 
\end{mathpar}

\noindent\framebox{$N;M \Vtaint \iota:\m{ok}$}
\begin{mathpar}
  \inferrule*[right=RIO-Assign-tainted]{ x \in (M \cup N)}{
    N;M \Vtaint x:= e : \m{ok} 
  }
  \and
  \inferrule*[right=RIO-Arr-tainted]{ a \in N}{
    N;M \Vtaint a[e]:= e' : \m{ok} 
  }
  \and
  \inferrule*[right=RIO-Get-tainted]{x \in (M \cup N)}{
    N;M \Vtaint x:= \m{IN}() : \m{ok} 
  }

\end{mathpar}

Judgment $N;I;M \Vdash_\rio \cmd:\m{ok}$ means that all of $\cmd$'s exclusive may-write
variables are in $N$, given $I$ is the set of input-dependent variables,
$N$ is the set of versioned variables from the most recent checkpoint,$M$ is 
the set of variables that must be written from the most recent checkpoint,  
and control is not tainted.
Judgement $N;M \Vtaint \cmd:\m{ok}$ means that all of $\cmd$'s exclusive may-write
variables are in $N$, given $M$ is the set of must-write variables,
$N$ is the set of versioned variables from the most recent checkpoint, and control is tainted. 
We say $\Vdash^{\mt{MstWt}} \cmd : M$ to denote the must write set $M$  
returned by the collection algorithm on the command $C$.
~\\\\
\noindent\framebox{$N;I; M \Vdash_\rio \cmd:\m{ok}$}

\begin{mathpar}
\inferrule*[right=RIO-$\iota$]{   N;I;M \Vdash_\rio \iota: I', M'}{
  N;I;M \Vdash_\rio \iota: \m{ok}
}
\and
\inferrule*[right=RIO-Cp]{\omega;\emptyset;\emptyset \Vdash_\rio \cmd: \m{ok}}{
  N;I;M \Vdash_\rio \m{checkpoint} (\omega);\cmd : \m{ok}
}
\and
\inferrule*[right=RIO-Seq]{N;I;M \Vdash_\rio \iota: I';M' \\ N;I';M' \Vdash_\rio \cmd: \m{ok}}{
  N;I;M \Vdash_\rio \iota;\cmd : \m{ok}
}
\and
\inferrule*[right=RIO-If-Dep]{I \cap \mt{rd}(e) \neq \emptyset \\
  M \Vdash^{\mt{MstWt}} \m{if}\ e\ \m{then}\ \cmd_1\ \m{else}\ \cmd_2:
  M' 
\\  \forall i \in [1, 2] , N;M' \Vtaint \cmd_i: \m{ok} }{
    N;I;M \Vdash_\rio \m{if}\ e\ \m{then}\ \cmd_1\ \m{else}\ \cmd_2: \m{ok} 
}
\and
 \inferrule*[right=RIO-If-NDep]{I \cap \mt{rd}(e) = \emptyset 
\\ 
N;I;M \Vdash_\rio \cmd_i: \m{ok}   \\ i \in [1, 2] }{
 N;I;M \Vdash_\rio \m{if}\ e\ \m{then}\ \cmd_1 \m{else}\ \cmd_2: \m{ok} 
}
\end{mathpar}

\noindent\framebox{$N;M \Vtaint \cmd:\m{ok}$}
\begin{mathpar}

\inferrule*[right=RIO-Cp-tainted]{\omega;\emptyset \Vdash_\rio \cmd: \m{ok}}{
  N;M \Vtaint \m{checkpoint} (\omega);\cmd : \m{ok}
}
\and

\inferrule*[right=RIO-Seq-tainted]{N;M \Vtaint \iota: \m{ok} \\ N;M \Vtaint \cmd: \m{ok}}{
  N;M \Vtaint \iota;\cmd : \m{ok}
}
\and

\inferrule*[right=RIO-If-tainted]{N;M \Vtaint \cmd_i: \m{ok}   \\ i \in [1, 2] }{
    N;M \Vtaint \m{if}\ e\ \m{then}\ \cmd_1\ \m{else}\ \cmd_2: \m{ok} 
}
\end{mathpar}

We write $M \Vdash^\mt{MstWt} \iota: M'$ to mean that given must-write set $M$ 
variables that must be written to in $\iota$ are in $M'$. 

\begin{mathpar}

\inferrule*[right=Must-NV-Assign]{ x~\mbox{is stored on non-volatile memory} }{
  M \Vdash^\mt{MstWt} x:= e : M \cup x 
}
\and
\inferrule*[right=Must-Assign-In]{ }{
  M \Vdash^\mt{MstWt} x:= \mt{IN}() : M \cup x 
}
\and
\inferrule*[right=Must-Assign-Arr]{ }{
  M \Vdash^\mt{MstWt} a[e]:= e' : M
}
\and

\inferrule*[right=Must-Seq]{ M \Vdash^\mt{MstWt} \iota : M' \\ 
M'\Vdash^\mt{MstWt} \cmd : M''}{
  M \Vdash^\mt{MstWt} \iota;\cmd : M''
}
\and
\inferrule*[right=Must-If]{M\vdash \cmd_i : M_i  
\\ i \in [1, 2]}{
  M \Vdash^\mt{MstWt} \m{if}\ e\ \m{then}\ \cmd_1\ \m{else}\ \cmd_2 : M_1 \cap M_2
}
\and
\inferrule*[right=Must-CP]{ }{
  M \Vdash^\mt{MstWt} \m{checkpoint()};\cmd : M
}

\end{mathpar}

\subsection{RIO Variable Collection Algorithm}

\noindent{\framebox{$X;M;I \Vdash_\rio \iota: X'; M'; I'$}}
\begin{mathpar}
  \inferrule*[right=I/O-skip]{ }{
    X;M;I \Vdash_\rio \m{skip} : X;M;I 
  }
  \and
  \inferrule*[right=I/O-Get]{ }{
    X;M;I \Vdash_\rio x:= \m{IN}() : X;M\cup x;I \cup x
  }
  \and
\inferrule*[right=I/O-Assign-Dep]{ I \cap \mt{rd}(e) \neq \emptyset}{
  X;M;I \Vdash_\rio x:= e : X;M \cup x;I \cup x
}
\and
\inferrule*[right=I/O-Assign-NDep]{ I \cap \mt{rd}(e) = \emptyset }{
  X;M;I \Vdash_\rio x:= e : X;M \cup x;I
}
\and
\inferrule*[right=I/O-dep-clear]{ I \cap \mt{rd}(e) = \emptyset \\ x \in I}{
  X;M;I \Vdash_\rio x:= e : X;M \cup x;I \setminus x
}
\and

\inferrule*[right=I/O-Arr-dep]{ I \cap \mt{rd}(e') \neq \emptyset}{
  X;M;I \Vdash_\rio a[e] := e' : X;M;I \cup a
}
\and
\inferrule*[right=I/O-Arr-loc]{ I \cap \mt{rd}(e) \neq \emptyset}{
  X;M;I \Vdash_\rio a[e]:= e' : X \cup a;M,I \cup a 
}
\and

\inferrule*[right=I/O-Arr-nodep]{ I \cap \mt{rd}(e') = \emptyset \\ I \cap \mt{rd}(e) = \emptyset}{
  X;M;I \Vdash_\rio a[e] := e' : X;M;I
}
\end{mathpar}

\noindent{\framebox{$X;M;I \Vdash_\rio \cmd \SeqStepsto{} \cmd': X'$}}
\begin{mathpar}
\inferrule*[right=I/O-If-NDep]{I \cap \mt{rd}(e) =\emptyset 
\\  X;M;I \Vdash_\rio \cmd_i \SeqStepsto{} \cmd_i': X_i
\\ i \in [1, 2]}{
X;M;I \Vdash_\rio \m{if}\ e\ \m{then}\ \cmd_1\ \m{else}\ \cmd_2 \SeqStepsto{ } 
\m{if}\ e\ \m{then}\ \cmd_1'\ \m{else}\ \cmd_2': 
X_1 \cup X_2
}
\and
\inferrule*[right=I/O-If-Dep]{I \cap \mt{rd}(e) \neq\emptyset
\\  X;M \Vtaint \cmd_i \SeqStepsto{} \cmd_i': X_i;M_i  \\ i \in [1, 2]}{
X;M;I \Vdash_\rio \m{if}\ e\ \m{then}\ \cmd_1\ \m{else}\ \cmd_2 \SeqStepsto{ } \\
\m{if}\ e\ \m{then}\ \cmd_1'\ \m{else}\ \cmd_2': 
(X_1 \cup X_2\cup M_1\cup M_2) \setminus (M_1 \cap M_2)
}

\and
\inferrule*[right=I/O-Seq]{ X;M;I \Vdash_\rio \iota : X';M';I' \\ 
X';M';I'\Vdash_\rio \cmd \SeqStepsto{ }\cmd' : X''}{
  X;M;I \Vdash_\rio \iota;\cmd\SeqStepsto{ } \iota;\cmd' : X'' 
}

\and
\inferrule*[right=Collect-CP]{\emptyset;\emptyset;\emptyset \Vdash_\rio \cmd 
\SeqStepsto{} \cmd': X'}{
  X;M;I \Vdash_\rio \m{checkpoint()};\cmd \SeqStepsto{} \m{checkpoint(X')};\cmd' : X
}
\and
\inferrule*[right=Collect-$\iota$]{ X;M;I \Vdash_\rio  \iota : X';M';I' }{
  X;M;I \Vdash_\rio \iota \SeqStepsto{} \iota: X'
}
\end{mathpar}

We write $X;M\Vtaint \iota: X';M'$ to mean 
that given that control is tainted, must-write set $M$, and exclusive may
set $X$,
variables that must be written to in $\iota$ are in $M'$, variables that may (exclusive) 
be written to in $\iota$ are in $X'$. 
We write $X;M \Vtaint \cmd \SeqStepsto{} \cmd': X';M'$ to mean that 
given must-write set $M$, and exclusive may set $X$, program $\cmd$, and control is 
tainted, 
$\cmd'$ is a rewritten program where variables that must be written to in $\cmd$ up to 
the next checkpoint are in $M'$, and
and variables that may (exclusive) be written to in $\cmd$ up to 
the next checkpoint are in $X'$.

\noindent\framebox{$X;M \Vtaint \iota: X';M' $}
\begin{mathpar}
  \inferrule*[right=I/O-skip-tainted]{ }{
    X;M \Vtaint \m{skip} : X;M
  }
  \and
\inferrule*[right=I/O-Get-tainted]{ }{
  X;M \Vtaint x:= \m{IN}() : X;M\cup x  }
  \and
\inferrule*[right=I/O-Assign-tainted]{ }{
  X;M\Vtaint x:= e : X;M \cup x 
}
\and
\inferrule*[right=I/O-Arr-tainted]{ }{
  X;M \Vtaint a[e] := e' : X \cup a;M 
}
\end{mathpar}

\noindent \framebox{$X;M \Vtaint \cmd\longrightarrow \cmd': X';M'$}
\begin{mathpar}
\inferrule*[right=I/O-If-tainted]{
X;M \Vtaint \cmd_i \SeqStepsto{} \cmd_i': X_i;M_i \\ i \in [1, 2]}{
X;M \Vtaint \m{if}\ e\ \m{then}\ \cmd_1\ \m{else}\ \cmd_2 \SeqStepsto{ } 
\\
\m{if}\ e\ \m{then}\ \cmd_1'\ \m{else}\ \cmd_2': 
((X_1 \cup X_2 \cup M_1 \cup M_2) \setminus (M_1 \cap M_2)) ;  (M_1 \cap M_2)
}

\and
\inferrule*[right=I/O-Seq-tainted]{ X;M \Vtaint \iota : X';M' \\ 
X';M'\Vtaint \cmd \SeqStepsto{ }\cmd' : X'';M''}{
  X;M \Vtaint \iota;\cmd\SeqStepsto{ } \iota;\cmd' : X'';M'' 
}

\and
\inferrule*[right=Collect-CP-tainted]{\emptyset;\emptyset;\emptyset \Vdash_\rio 
\cmd \SeqStepsto{} \cmd': X'}{
  X;M\Vtaint \m{checkpoint()};\cmd \SeqStepsto{} \m{checkpoint(X')};\cmd' : X; M
}
\and
\inferrule*[right=Collect-$\iota$-tainted]{ X;M\Vtaint  \iota : X';M' }{
  X;M\Vtaint \iota \SeqStepsto{} \iota: X';M'
}

\end{mathpar}

\begin{lem}
\begin{enumerate}
\item If $\ee::M_1 \Vdash^\mt{MstWt} \iota: M_2$ and 
 $X;M_1 \Vtaint \iota: X';M'_2$ 
then $M_2 = M'_2$
\item If $\ee::M_1 \Vdash^\mt{MstWt} \cmd: M_2$ and 
 $X;M_1 \Vtaint \cmd\longrightarrow \cmd': X';M'_2$ 
then $M_2 = M'_2$
\end{enumerate}
\end{lem}
\begin{proofsketch}
By induction over the structure of $\ee$.
\end{proofsketch}

\begin{lem}~\\
\begin{itemize}
\item If $X;M;I \Vdash_\rio \iota: X'; M'; I'$ and all of the locations in $X$
  are array and all of the locations in $M$ are variables
  then  all of the locations in $X'$
  are array and all of the locations in $M'$ are variables and $X'  \supseteq X$
\item If $X;M;I \Vdash_\rio \cmd \SeqStepsto{} \cmd': X'$ and all of the locations in $X$
  are array and all of the locations in $M$ are variables then $X'
  \supseteq X$
\item If $X;M \Vtaint \iota: X';M' $ and all of the locations in $X$
  are array and all of the locations in $M$ are variables
  then all of the locations in $X'$
  are array and all of the locations in $M'$ are variables $X'  \supseteq X$ and $M'\cup X' \supseteq
  X\cup M$.
\item If $X;M \Vtaint \cmd\longrightarrow \cmd': X';M'$ and  
all of the locations in $X$
  are array and all of the locations in $M$ are variables
  then $X'\cap M'=\emptyset$, $X'  \supseteq X$, $M'\cup X' \supseteq
  X\cup M$.
\end{itemize}
\end{lem}
\begin{proofsketch}
By induction over the structure of the checkpointed locations collection derivations.
\end{proofsketch}

\begin{lem}~\\
\begin{itemize}
\item If $X;M;I \Vdash_\rio \iota: X'; M'; I'$ and all of the locations in $X$
  are array and all of the locations in $M$ are variables
  then $\forall N\supseteq X'$, $N; I; M \Vdash_\rio \iota: I'; M'$ 
\item If $X;M;I \Vdash_\rio \cmd \SeqStepsto{} \cmd': X'$ and all of the locations in $X$
  are array and all of the locations in $M$ are variables  then 
then $\forall N\supseteq X'$, $N; I; M \Vdash_\rio \cmd': \m{ok}$ 

\item If $X;M \Vtaint \iota: X';M' $ and all of the locations in $X$
  are array and all of the locations in $M$ are variables  
  then $\forall N_c\supseteq X'$, $\forall M_c$ s.t.
  $N_c\cup M_c \supseteq X'\cup M'$, 
  $N_c; M_c \Vtaint \iota: \m{ok}$ 

\item If $X;M \Vtaint \cmd\longrightarrow \cmd': X';M'$and all of the locations in $X$
  are array and all of the locations in $M$ are variables 
  then $\forall N_c\supseteq X'$, $\forall M_c$ s.t.
  $N_c\cup M_c \supseteq X'\cup M'$, 
  $N_c; M_c \Vtaint \cmd': \m{ok}$ 
\end{itemize}
\end{lem}
\begin{proofsketch}
By induction over the structure of the checkpointed locations collection derivations.
\end{proofsketch}

\subsection{Tainting Semantics} 
We augment values with tainted values denoted $v^T$. The expression
evaluation now propagates the taint tag. 
\[
\begin{array}{llcl}
\textit{values} & \val & \bnfdef & v \bnfalt v^t
\\ \textit{Configuration} & \sigma& \bnfdef & (\timestamp,  \nvmem, \vmem, \cmd)
\end{array}
\]

We write $(\timestamp,  \nvmem, \vmem, \cmd) \SeqStepsto{O} (\timestamp', 
\nvmem', \vmem', \cmd')$ to denote the small-step operational
semantics of the core calculus. The rules are summarized below. 

~\\\noindent\framebox{$(\timestamp,  \nvmem, \vmem, \cmd) \SeqStepsto{O} (\timestamp', 
\nvmem', \vmem', \cmd')$}

\begin{mathpar}

\inferrule*[right=I/O-Tnt-CheckPoint]{ }{
  (\timestamp,  \nvmem, \vmem, \m{checkpoint}(\omega);\cmd) 
 \SeqStepsto{\m{checkpoint}}  (\timestamp +1,  \nvmem, \vmem, \cmd) 
}
\and
\inferrule*[right=I/O-Tnt-Skip]{ }{
  (\timestamp,  \nvmem, \vmem, \m{skip};\cmd) 
 \SeqStepsto{}  (\timestamp,  \nvmem, \vmem, \cmd) 
}
\and
\inferrule*[right=I/O-Tnt-NV-Assign]{x \in \m{dom}(N)\\N, V \vdash e \Downarrow_{r} \val }{
  (\timestamp,  \nvmem, \vmem, x:=e) 
 \SeqStepsto{[r]}  (\timestamp + 1,  \nvmem[x\mapsto \val], \vmem, \m{skip}) 
}
\and
\inferrule*[right=I/O-Tnt-Assign-Arr]{N, V \vdash e \Downarrow_{r} \val \\N, V \vdash e' \Downarrow_{r'} \val' }{
  (\timestamp,  \nvmem, \vmem, a[e]:=e') 
 \SeqStepsto{[r, r']}  (\timestamp +1,  \nvmem[a[\val]\mapsto \val'], \vmem, \m{skip}) 
}
\and
\inferrule*[right=I/O-Tnt-V-Assign]{x \in \m{dom}(V)\\N, V \vdash e \Downarrow_{r} \val }{
  (\timestamp,  \nvmem, \vmem, x:=e) 
 \SeqStepsto{[r]}  (\timestamp +1,  \nvmem, \vmem[x\mapsto \val], \m{skip}) 
}
\and
\inferrule*[right=I/O-Tnt-NV-Assign-In]{x \in \m{dom}(N) }{
  (\timestamp,  \nvmem, \vmem, x:=\m{IN}()) 
 \SeqStepsto{\m{in}(\timestamp)}  (\timestamp + 1,  \nvmem[x\mapsto \m{in}(\timestamp)^t], \vmem, \m{skip}) 
}

\and
\inferrule*[right=I/O-Tnt-Assign-Arr-In]{N, V \vdash e \Downarrow_{r} \val }{
  (\timestamp,  \nvmem, \vmem, a[e]:=\m{IN}()) 
 \SeqStepsto{[r], \m{in}(\timestamp)}  (\timestamp +1,  \nvmem[a[\val]\mapsto \m{in}(\timestamp)^t], \vmem, \m{skip}) 
}
\and
\inferrule*[right=I/O-Tnt-V-Assign-In]{x \in \m{dom}(V)
}{
  (\timestamp,  \nvmem, \vmem, x:=\m{IN}()) 
 \SeqStepsto{\m{in}(\timestamp)}  (\timestamp +1,  \nvmem, \vmem[x\mapsto \m{in}(\timestamp)^t], \m{skip}) 
}
\and
\inferrule*[right=I/O-Tnt-Seq]{  (\timestamp,  \nvmem, \vmem, i) 
 \SeqStepsto{o}  (\timestamp + 1,  \nvmem', \vmem', \m{skip})}{
  (\timestamp,  \nvmem, \vmem, i;c) 
 \SeqStepsto{o}  (\timestamp + 1,  \nvmem', \vmem', c) 
}
\and
\inferrule*[right=I/O-Tnt-IfT]{N, V \vdash e \Downarrow_{r} \val \\
  \val =\m{true} ~\mbox{or}~ \m{true}^t}{
  (\timestamp,  \nvmem, \vmem, \m{if}\ e\ \m{then}\ \cmd_1\ \m{else}\ \cmd_2) 
 \SeqStepsto{[r]}  (\timestamp + 1,  \nvmem, \vmem, \cmd_1) 
}
\end{mathpar}
\begin{mathpar}
\inferrule*[right=I/O-Tnt-IfF]{N, V \vdash e \Downarrow_{r} \val\\ 
\val = \m{false}~\mbox{or}~\m{false}^t}{
  (\timestamp,  \nvmem, \vmem, \m{if}\ e\ \m{then}\ \cmd_1\ \m{else}\ \cmd_2) 
 \SeqStepsto{[r]}  (\timestamp+1,  \nvmem, \vmem, \cmd_2) 
}
\and
\inferrule*[right=I/O-Tnt-Sleep]{ \timestamp' > \timestamp}{
  (\timestamp,  \nvmem, \vmem, \cmd) 
 \SeqStepsto{}  (\timestamp',  \nvmem, \vmem, \cmd) 
}
\end{mathpar}

\subsection{Auxiliary Definitions}

\paragraph{Idempotent reads}
 
\begin{mathpar}

\inferrule*[right=I-Rb-Base]{ }{
  O \oleqm O
}
\and 
\inferrule*[right=I-Rb-Ind]{ O_1' \oleqm O_2}{
  O_1, \m{reboot}, O_1' \oleqm O_2
}
\\
\inferrule*[right=Cp-Base]{ O_1\oleqm O_2}{
  O_1 \oleqmc O_2
}
\and
\inferrule*[right=Cp-Ind]{O_1 \oleqm O_2 \\ O_1' \oleqmc O_2'}{
  O_1, \m{checkpoint}, O_1' \oleqmc O_2, O_2'
}

\end{mathpar}

\paragraph{Related non-volatile memory}

\begin{defn}[Related non-volatile memories at the same execution point]
  \label{def:app-io-same_exec}
  $\tau, N_0, V_0, c_0, c', \inputs \vdash \nint \sim \ncont$ iff
\begin{itemize}
\item $\m{dom}(\nint)= \m{dom}(\ncont)$ and
\item $\forall \loc \in \nint$ s.t. $\nint(\loc) \neq \ncont(\loc)$, 
 \begin{itemize}
\item  $\loc\in\mt{MFstWt}(N_0, V_0, c_0)$
\item let $\{T\} = \runof{\sigma,\inputs, c'}$ where
  $\sigma=(\timestamp, N_0, V_0, c_0)$ and the last state of $T$ is
  $(\timestamp',  \nint, V, c')$
\begin{itemize}
\item  $\loc\in\mt{MstWt}(\nint, V, c')$ 
\item  $\loc\notin \mt{Wt}(T)$
\end{itemize}
\end{itemize}
 
\end{itemize}

\end{defn} 
  
\begin{defn}[Relating memories between current and initial execution point]
\label{def:app-io-initial}  
$N_c, N_{\mt{rb}}, V, c \vdash \nint \sim \ncont$ iff 
$\m{dom}(\nint)= \m{dom}(\ncont)$, 
\begin{itemize}
\item $\m{dom}(\nint)= \m{dom}(\ncont)$ and
\item $\forall \loc \in \nint$ s.t. $\nint(\loc) 
\neq \ncont(\loc)$, $\loc\in
  N_c \cup \mt{MFstWt}(N_{\mt{rb}}, V, \cmd)$
\end{itemize}

\end{defn} 

\begin{defn}[Related configurations]
$\timestamp_0, \inputs, \nrb \vdash (\tau_1,  \context, \nint, \vmem_1, \cmd_1) \sim (\tau_2,
 \ncont, \vmem_2, \cmd_2)$ 
iff $\context = (\nvmem_c, \vmem_0, \cmd_0)$, 
$\nvmem_c \subseteq \nrb$
and  $\timestamp_0,\nrb, \vmem_0,
\cmd_0, \cmd_1, \inputs \vdash \nint \sim \ncont$,
$\tau_1=\tau_2$, $\vmem_1=\vmem_2$, and $\cmd_1=\cmd_2$.
\end{defn} 

\begin{defn}[Erased configuration]
$\erase{(\tau, \context, \nvmem, \vmem, \cmd)} = (\tau, \nvmem, \vmem, \cmd)$ 
\end{defn} 

\begin{lem}[Solid to Dash]\label{lem:io-mem-relation-relation}
If $N_c, N_0, V_0, c_0\vdash \nint \sim \ncont$, $N_c \subseteq \nint$, $N_c; \emptyset; \emptyset \Vdash 
c_0: \m{ok}$, $N_c \subseteq \ncont$, $N_c \subseteq \nint$ and $N_c \subseteq \nvmem_0$
then $\forall \tau$, $\tau,\nint, V_0, c_0, c_0, \emptyset \vdash \nint \sim \ncont$. 
\end{lem}
\begin{proofsketch} 
By examining the two relations and Lemmas~\ref{lem:io-wt-must-version} and ~\ref{lem:io-wt-must-version-arr} (writes are 
either must-writes or versioned). 
\end{proofsketch}

\begin{lem}[Dash to Solid]\label{lem:io-mem-same-to-init}
  If $\tau, \nint, V_0, c_0, c_0, \emptyset \vdash \nint \sim \ncont$,  and $N_c; \emptyset; \emptyset \Vdash 
  c_0: \m{ok}$,  $N_c\subseteq \nint$ and  $N_c \subseteq \ncont$,
  then $N_c, \nint, V_0, c_0 \vdash \nint \sim \ncont$
  \end{lem}
  \begin{proofsketch} 
  By examining the two relations. 
  \end{proofsketch}

  \begin{lem}[Solid to Solid after Reboot]\label{lem:io-solid-reboot}
    If $N_{\mt{ckpt}},N_{\mt{rb_0}}, V_0, c_0 \vdash \nint \sim \ncont$
    , $N_{\mt{rb}} = \nint \lhd N_{\mt{ckpt}}$,
    and $N_{\mt{ckpt}}; \emptyset; \emptyset \Vdash 
    c_0: \m{ok}$, 
     then $N_{\mt{ckpt}},N_{\mt{rb}}, V_0, c_0 \vdash N_{\mt{rb}} \sim \ncont$
  \end{lem}
  \begin{proof}
  By examining the two relations and Lemma~\ref{lem:io-wt-must-version}.
  \end{proof}

\subsection{Correctness Proofs}
The correctness theorem follows from the following lemma.
\begin{lem}[Correctness]
\label{lem:io-correctness-basic}
If
  $T= (\timestamp, \context, \nvmem, \vmem, \cmd) \MStepsto{O_1} \Sigma$,
$\context = (\nvmem_0, \vmem, \cmd)$,  $\nvmem_0\subseteq \nvmem$,
$ \nrb = \m{nearestRb(T)}$, 
and  $\nvmem_0, \emptyset, \emptyset \Vdash_\war \cmd: \m{ok}$,
 $\nvmem_0, \emptyset, \emptyset \Vdash_\rio \cmd: \m{ok}$ 
then $\exists O_2, \timestamp_2, \sigma$ s.t. 
\begin{enumerate}
\item   $(\timestamp_2,\nvmem, \vmem, \cmd) \MSeqStepsto{O_2}
  \sigma$, $\timestamp_2, O_2\bnfalt_\m{in}, \nrb \vdash \Sigma \sim
  \sigma$, 
and $\timestamp_2\geq\timestamp$, $O_1\oleqmc O_2$.
and
\item $\forall T' = \Sigma \MStepsto{O} \Sigma'$,
 $\mt{CP}(\Sigma' )$ 
and $T'$ does not contain checkpoints or reboots
implies
$\sigma \MSeqStepsto{O} \erase{(\Sigma')}$
\end{enumerate}
\end{lem}
\ifproofs
\begin{proof} 
By induction on the number of checkpoints in $O_1$.
\begin{description}
\item[Base case:] $O_1$ does not include any checkpoint, directly
  apply Lemma~\ref{lem:io-one-cp-m-f}.
\item[Inductive case:] $O_1$ contains $k+1$ checkpoints where $k\geq
  0$
\begin{tabbing}
By assumption, 
\\\qquad \= (1)\quad \= 
$T= (\timestamp,\context, \nvmem, \vmem, \cmd) \MStepsto{O_1}
(\timestamp_1, \context, \nvmem_1, \vmem_1, \m{checkpoint}(\omega);\cmd_1) 
\Stepsto{\m{checkpoint}} \Sigma' \MStepsto{O_2} \Sigma''$,
\\\>\> where $O_1$ does not contain checkpoints
\\By Lemma~\ref{lem:io-one-cp-m-f}, exists $O'_1$  and $\timestamp'$ s.t.
\\\>(2)\>$(\timestamp',\nvmem, \vmem, \cmd) \MSeqStepsto{O'_1}
(\timestamp_1,\nvmem_1, \vmem_1, \m{checkpoint}(\omega);\cmd_1) $ and $O_1 \oleqm
O'_1$ and $\timestamp' \geq \timestamp$
\\By \rulename{I/O-CP-CheckPoint} and (1)
\\\>(3)\> $\Sigma' =(\timestamp_1 +1, \context_1, \nvmem_1, \vmem_1, \cmd_1)$ where
$\context_1 =  (\proj{\nvmem_1}{\omega}, \vmem_1, \cmd_1)$
\\By \rulename{I/O-Tnt-CheckPoint} and (2)
\\\>(4)\> $(\timestamp_1, \nvmem_1, \vmem_1, \m{checkpoint}(\omega);\cmd_1)
\SeqStepsto{} (\timestamp_1 +1, \nvmem_1, \vmem_1, \cmd_1)$ 
\\By  $\nvmem_0, \emptyset, \emptyset \Vdash_\rio \cmd: \m{ok}$  and 
Lemma~\ref{lem:io-cmd-runtime-ok}
\\\>(5)\> $\exists$ $I'$, $M$ s.t. $\nvmem_0, I', M \vDash_\rio \m{checkpoint}(\omega);\cmd_1:\m{ok}$
\\\>\>or $\exists$ $M$ s.t. $\nvmem_0, M \Vtaint \m{checkpoint}(\omega);\cmd_1:\m{ok}$
\\By inversion of (5)
\\\>(6)\>  $\omega, \emptyset, , \emptyset\Vdash_\rio \cmd_1: \m{ok}$
\\By Lemma~\ref{lem:cmd-runtime-ok} and $\nvmem_0, \emptyset, \emptyset \Vdash_\war \cmd: \m{ok}$
\\\>(7)\> $\exists W', R' $ s.t. $\nvmem_0,  W', R' \Vdash_\war \m{checkpoint}(\omega);\cmd_1: \m{ok}$
\\By inversion of (7)
\\\> (8)\> $\omega, \emptyset, \emptyset \Vdash_\war \cmd_1: \m{ok}$
\\By I.H. on the tail of $T$ starting from $\Sigma'$, (3), (6), (8),
exists $O'_2$ and $\timestamp_2$ s.t.
\\\>(9)\>  $(\timestamp_2,\nvmem_1, \vmem_1, \cmd_1) \MSeqStepsto{O'_2}
  \sigma''$ and $\timestamp_2, O'_2\bnfalt_\m{in}, \nvmem_1  \vdash \Sigma'' \sim \sigma''$ 
  and $ O_2 \oleqmc O_2'$ and $\timestamp_2 \geq (\timestamp_1 + 1)$ and

\\\>(10)\>
 $\Sigma'' \MStepsto{O_3} (\timestamp_3,\context_3, \nvmem_3, \vmem_3, \m{checkpoint}(\omega');\cmd_3)$
implies
\\\>\>$\sigma'' \MSeqStepsto{O_3} (\timestamp_3, \nvmem_3, \vmem_3, \m{checkpoint}(\omega');\cmd_3)$,
\\By (2) and (9) and \rulename{Cp-Ind}
\\\>(11)\>\ $O_1, \m{checkpoint}, O_2 \oleqmc O'_1, O'_2$ 
\\By \rulename{I/O-Tnt-Sleep} and $\timestamp_2 \geq (\timestamp_1 + 1)$
\\\>(12) \> $(\timestamp_1+1,\nvmem_1, \vmem_1, \cmd_1) 
 \SeqStepsto{} (\timestamp_2,\nvmem_1, \vmem_1, \cmd_1)$
\\By connecting the executions the conclusion holds
\end{tabbing}
\end{description}
\end{proof}
\fi

\begin{lem}\label{lem:io-cmd-runtime-ok} 
  If $\nvmem_0, \emptyset , \emptyset\Vdash_\rio \cmd: \m{ok}$ and
  $(\timestamp, \nvmem, \vmem, \cmd)\MSeqStepsto{O} 
  (\timestamp',  \nvmem', \vmem', \cmd')$, and $O$ does not
  contain checkpoints, then either $\exists$ $I'$, $M$ s.t.
  $\nvmem_0, I' , M\Vdash_\rio \cmd': \m{ok}$ or 
$\exists$ $M$ s.t.
  $\nvmem_0, M \Vtaint \cmd': \m{ok}$. 
\end{lem} 
\begin{proofsketch}
By induction over the derivation $\nvmem_0, W, R \Vdash \cmd:
\m{ok}$. 
\end{proofsketch}

\begin{lem}[One checkpoint, multiple failures]
\label{lem:io-one-cp-m-f}
If
  $T = (\timestamp, \context, \nvmem_1, \vmem, \cmd) \MStepsto{O_1} \Sigma$
  where $T$ does not execute checkpoint, $ \nrb = \m{nearestRb(T)}$
 $\context = (\nvmem_0, \vmem, \cmd)$, $\nvmem_0\subseteq\nvmem_1$, $\nvmem_0\subseteq\nvmem_2$,
 and $\nvmem_0, \nvmem_1, \vmem, \cmd \vdash \nvmem_1 \sim \nvmem_2$ 
 and  $\nvmem_0, \emptyset, \emptyset \Vdash_\war \cmd: \m{ok}$,  
 $\nvmem_0, \emptyset, \emptyset \Vdash_\rio \cmd: \m{ok}$ 
then
\begin{itemize}
\item  $\exists O_2, \timestamp_2, \sigma$ s.t. $(\timestamp_2,\nvmem_2, \vmem, \cmd) \MSeqStepsto{O_2}
  \sigma$ and $\timestamp_2, O_2\bnfalt_\m{in}, \nrb \vdash \Sigma \sim \sigma$ and $\timestamp_2\geq\timestamp$
  and $O_1\oleqm O_2$.
\item 
$\forall T' = \Sigma \MStepsto{O} \Sigma'$ where $\Sigma' =
(\timestamp',  \context', \nvmem', \vmem', \m{checkpoint}(\omega);\cmd')$ and $T'$
does not contain checkpoints
implies
$\sigma \MSeqStepsto{O} \erase{(\Sigma')}$,
\end{itemize}
\end{lem}
\ifproofs
\begin{proof}
By induction on the number of reboots in $O_1$.
\begin{description}
\item[Base case:] $O_1$ does not include any reboot, first apply
  Lemma~\ref{lem:io-mem-relation-relation} to show $\timestamp,N_1, V, c, c, \emptyset \vdash
  N_1 \sim N_2$; then directly
  apply Lemma~\ref{lem:io-one-f} and Lemma~\ref{lem:io-one-f-completion}.
\item[Inductive case:] $O_1$ contains $k+1$ reboots where $k\geq
  0$
\begin{tabbing}
By assumption
\\\qquad \= (1)\quad \= 
$T= (\timestamp, \context, \nvmem_1, \vmem, \cmd) \MStepsto{O'_1}
(\timestamp_1, \context, \nvmem_1', \vmem_1, \m{reboot}(\omega)) 
\Stepsto{\m{reboot}} \Sigma' \MStepsto{O''_1} \Sigma''$,
\\\>\> and $O'_1$ does not contain reboots or checkpoints
\\By assumption $\nvmem_0,\nvmem_1, \vmem, \cmd \vdash \nvmem_1 \sim \nvmem_2$
\\By Lemma~\ref{lem:io-partial-run} 
\\\>(2)\> $\nvmem_0, \nvmem_1, \vmem, \cmd \vdash \nvmem_1' \sim \nvmem_2$.  
It follows that  $\nvmem_0, \nvmem_1, \vmem, \cmd \vdash \nvmem_1' \lhd\nvmem_0 \sim \nvmem_2$
\\By (2) and Lemma~\ref{lem:io-solid-reboot} 
\\\>(2b)\> $\nvmem_0, \nvmem_1' \lhd\nvmem_0, \vmem, \cmd \vdash \nvmem_1' \lhd\nvmem_0 \sim \nvmem_2$
\\By \rulename{I/O-CP-Reboot} and (1)
\\\>(3)\> $\Sigma' = (\timestamp_1 + n, \context, \nvmem_1' \lhd \nvmem_0,\vmem,\cmd)$
\\By I.H. on all the tail of $T$ starting from $\Sigma'$ and (2), (2b), (3)
\\\> (4)\> exists $O_2$, $\timestamp'$, and $\sigma$ 
s.t. $(\timestamp', \nvmem_2, \vmem, \cmd) \MSeqStepsto{O_2}
  \sigma$  and $\timestamp', O_2\bnfalt_{\mt{in}}, \nrb \vdash \Sigma'' \sim \sigma $ 
\\\>\>  and $\timestamp' \geq \timestamp_1 + n $ and $ O''_1\oleqm O_2$
\\\> (5)\> $\Sigma'' \MStepsto{O} 
(\timestamp_3, \context', \nvmem', \vmem', \m{checkpoint}(\omega);\cmd')$
\\\>\>implies
$\sigma \MSeqStepsto{O} (\timestamp_3, \nvmem', \vmem',
\m{checkpoint}(\omega);\cmd')$ 
\\By (1), (4) and \rulename{Rb-Ind}
\\\>(6)\>$ O'_1, \m{reboot}, O''_1\oleqm O_2$
\\By (3) and (4) and time is monotonically increasing
\\\>(7)\>$\timestamp' \geq \timestamp_1 + n\geq \timestamp_1 \geq \timestamp$
\\By (5), (6), (7) the conclusion holds.
\end{tabbing}
\end{description}
\end{proof}
\fi

\begin{lem}[Partial run relates to initial state (Multi-steps)]
\label{lem:io-partial-run}
If
$T= (\timestamp, \context, \nvmem_1, \vmem, \cmd_0) \MStepsto{O}
(\timestamp',  \context,
\nvmem'_1, \vmem', \cmd)$,  $T$ does not
  execute checkpoint or reboot, $\nvmem_0; \emptyset;
  \emptyset\Vdash_\war \cmd_0: \m{ok}$,
$\nvmem_0, \emptyset, \emptyset \Vdash_\rio \cmd_0: \m{ok}$, 
$\context = (\nvmem_0, \vmem, \cmd_0)$,  
and $\nvmem_0, \nvmem_1, \vmem, \cmd_0 \vdash \nvmem_1 \sim \nvmem_2$ 
then 
 $\nvmem_0, \nvmem_1, \vmem,\cmd_0 \vdash \nvmem'_1 \sim \nvmem_2$. 
\end{lem}
\begin{proof}
By induction over the length of $T$. The base case is trivial. The
inductive case uses I.H. and Lemma~\ref{lem:io-partial-run-one-step}. 
\end{proof}

\begin{lem}[Relating syntactic and semantic read]~\label{lem:io-ok-read-same}
If 
$\ee:: \nvmem_0; W_0; R_0\Vdash \cmd_0: \m{ok}$, 
and $T = (\timestamp_0, N_1, V_0, \cmd_0) \MSeqStepsto{}
(\timestamp,  N, V, c)$, $T$ does not
contain any checkpoints
then $\exists$ $W$ and $R$ s.t.
$\nvmem_0; W; R \Vdash_\war \cmd: \m{ok}$, and 
$\mt{RD}(T) \cup R_0 \subseteq R$.
\end{lem}
\begin{proofsketch}
Induction over the derivation $\ee$. 
\end{proofsketch}

\begin{lem}[Relating syntactic must write and semantic must write]
\label{lem:io-ok-mst-wt-same}
If  $\ee:: M \Vdash^\mt{MstWt} \cmd: M'$, $x\in M'\setminus M$ 
then $\forall T$ s.t. $T = (\timestamp,  N, V, \cmd)\MSeqStepsto{} (\timestamp',  N', V',
\m{checkpoint}(\omega);c')$ and $|T|\geq 1$, $x\in\mt{Wt}(T)$.
\end{lem}
\begin{proofsketch} By induction over the structure of $\ee$.
\end{proofsketch}

\begin{lem}[Variable Writes are either First Writes or Versioned]
\label{lem:io-wt-fst-version}
If  $\nvmem_0; \emptyset; \emptyset\Vdash_\war \cmd_0: \m{ok}$, 
given any trace $T_c$ from $\cmd_0$ to the nearest checkpoint:
 $T_c = T_0\cdot (\timestamp,  N, V,  x:=e; \cmd_1)  
 \MSeqStepsto{O}  (\timestamp',  N', V', \m{checkpoint}(\omega);c')$ 
then $x\in \mt{FstWt}(T_c)$ or $x\in \nvmem_0$.
\end{lem}
\begin{proofsketch}
The same as proof Lemma~\ref{lem:wt-fst-version}. 
\end{proofsketch}

\begin{lem}\label{lem:io-taint-exists-if}
If  $\ee::\nvmem_0; I ; M\Vdash_\rio \cmd_0: \m{ok}$, 
 exists $T$ s.t. $T =  (\timestamp_0,  N_1, V_1,
\cmd_0)\MSeqStepsto{}  (\timestamp_2,  N_2, V_2, \cmd_1)$  and 
$\nvmem_0, M_1 \Vtaint \cmd_1: \m{ok}$, then 
$T =  (\timestamp_0, N_1, V_1,
\cmd_0) \MSeqStepsto{}  (\timestamp_1,  N_1, V_1, \cmd)
\MSeqStepsto{}  (\timestamp_2,  N_2, V_2, \cmd_1)$ where $\cmd
= \m{if}\ e\ \m{then}\ \cmd_t\ \m{else}\ \cmd_f$ and 
 $M' \Vdash^{\mt{MstWt}} \cmd : M_1$ and $\exists \ee', I $
 s.t. $\ee'::\nvmem_0; I \Vdash_\rio \cmd: \m{ok}$ and $\ee'$ is a
 sub-derivation of $\ee$.
\end{lem}
\begin{proofsketch} By induction over the length of $T$. 
\end{proofsketch}

\begin{lem}[Variable Writes are either Must Writes or Versioned]
\label{lem:io-wt-must-version}
If  $\nvmem_0; \emptyset; \emptyset\Vdash_\rio \cmd_0: \m{ok}$, 
 exists $T$ s.t. $T =  (\timestamp_0, N_1, V_1,
\cmd_0)\MSeqStepsto{}  (\timestamp,  N_2, V_2,  x:=e; \cmd_1)$  
then  $\forall T_c$ s.t. $T_c = (\timestamp_0,  N, V,
\cmd_0)\MSeqStepsto{} (\timestamp',  N', V',
\m{checkpoint}(\omega);c')$, $x\in \mt{Wt}(T_c)\cup \nvmem_0$.
\end{lem}
\begin{proof}
 By Lemma~\ref{lem:io-cmd-runtime-ok} there are two cases:
 (I):  $\exists$ $I$ $M$ s.t.
  $\nvmem_0, I; M \Vdash_\rio x:=e; \cmd_1: \m{ok}$ or  (II):
$\exists$ $M$ s.t.  $\nvmem_0, M \Vtaint x:=e; \cmd_1: \m{ok}$. 
\\By assumption, let $T_0$ be the trace s.t. 
$T_0 =  (\timestamp_0, N_1, V_1,
\cmd_0)\MSeqStepsto{}  (\timestamp,  N_2, V_2,  x:=e; \cmd_1)$  
\begin{description}
\item[Case I: ] $\nvmem_0, I, M \Vdash_\rio x:=e; \cmd_1: \m{ok}$
\begin{tabbing}
By Lemma~\ref{lem:not-tainted-deterministic}, 
\\\quad \= (I1)\quad \=  $\forall T$ s.t. $T = (\timestamp_0, N, V,
\cmd_0)\MSeqStepsto{} (\timestamp',  N', V',
\m{checkpoint}(\omega);c')$, 
\\\>\> $T = (\timestamp, N_1, V_1, \cmd_0) \MSeqStepsto{}
(\timestamp'',  N'_2, V'_2,  x:=e; \cmd_1) \MSeqStepsto{} (\timestamp',  N', V',
\m{checkpoint}(\omega);c')$
\\By (I1) and operation semantics and the definition of $\mt{Wt}$
\\\>(I2)\> $x\in\mt{Wt}(T)$
\end{tabbing}

\item[Case II: ] $\exists$ $M$ s.t.  $\ee::\nvmem_0, M \Vtaint x:=e; \cmd_1: \m{ok}$. 
\begin{tabbing}
By inversion of $\ee$
\\\quad \= (II1)\quad \=  $x\in \nvmem_0$ (conclusion holds) or $x\in M$
\\ We continue with the case where $x\in M$
\\By Lemma~\ref{lem:io-taint-exists-if}
\\\>(II2)\> $T_0 =  (\timestamp_0, N_1, V_1,
\cmd_0) \MSeqStepsto{}  (\timestamp_1,  N_1, V_1, \cmd)
\MSeqStepsto{}  (\timestamp,  N_2, V_2, x:=e;\cmd_1)$ 
\\\>\>where $\cmd = \m{if}\ e\ \m{then}\ \cmd_t\ \m{else}\ \cmd_f$ 
\\\>(II3)\>  and $M' \Vdash^{\mt{MstWt}} \cmd : M$ 
\\\>(II4)\> and $\exists I $ s.t. $\nvmem_0; I; M' \Vdash_\rio \cmd:
\m{ok}$.
\\By Lemma~\ref{lem:not-tainted-deterministic}
\\ \>(II5)\> $\forall T_c$ s.t. $T_c = (\timestamp_0, N, V,
\cmd_0)\MSeqStepsto{} (\timestamp',  N', V',
\m{checkpoint}(\omega);c')$, 
\\\>\> $T_c = T_{c1}\cdot T_{c2}$,  $T_{c1} = (\timestamp_0, N, V,
\cmd_0) \MSeqStepsto{}  (\timestamp'_1,  N'_1, V'_1, \cmd)$, 
\\\>\>$T_{c2} =  (\timestamp'_1,  N'_1, V'_1, \cmd) 
 \MSeqStepsto{} (\timestamp',  N', V', \m{checkpoint}(\omega);c')$, 
\\\>\>where $\cmd = \m{if}\ e\ \m{then}\ \cmd_t\ \m{else}\ \cmd_f$ 
\\\>(II6)\> and $\forall x\in M'$, $x\in \mt{Wt}(T_{c1})$
\\By Lemma~\ref{lem:io-ok-mst-wt-same}, (II3), and $x\in M \setminus M'$
\\\>(II7)\> $x\in\mt{Wt}(T_{c2})$
\\By (II6) and (II7), $x\in\mt{Wt}(T_{c})$
\end{tabbing}
\end{description}
\end{proof}

\begin{lem}[Array Writes are either First Writes or Versioned]
\label{lem:io-wt-fst-version-arr}
If  $\nvmem_0; \emptyset; \emptyset\Vdash_\war \cmd_0: \m{ok}$, 
given any trace $T_c$ from $\cmd_0$ to the nearest checkpoint:
 $T_c = T_0\cdot (\timestamp,  N, V,  a[e]:=e'; \cmd_1)  
 \MSeqStepsto{O}  (\timestamp',  N', V',
 \m{checkpoint}(\omega);c')$,  $N,V \vdash e \Downarrow_r v$
then $a[v]\in \mt{Wt}(T_c)$ implies $a[v]\in \mt{FstWt}(T_c)$ or $a\in \nvmem_0$.
\end{lem}
\begin{proofsketch}
We consider two cases: (I) $a[v]\in \mt{RD}(T_0)\cup \mt{rd}(e) \cup \mt{rd}(e')$ 
and (II) $a[v]\notin \mt{RD}(T_0)\cup \mt{rd}(e) \cup \mt{rd}(e')$ 
\begin{description}
\item[Case (I)] $a[v]\in \mt{RD}(T_0)\cup \mt{rd}(e) \cup \mt{rd}(e')$ 
\begin{tabbing}
$a[v]\in \mt{RD}(T_0)\cup \mt{rd}(e)\cup \mt{rd}(e')$
\\By Lemma~\ref{lem:io-ok-read-same} and (2) and (4)
\\\quad\=(I1)\quad\= $\exists$ $W$ and $R$ s.t. 
$\ee':: \nvmem_0; W; R \Vdash_\war a[e]:=e'; \cmd_1: \m{ok}$,
 and 
 \\\>(I2)\>$\mt{RD}(T_0)  \subseteq R$
\\By inversion of $\ee'$ and (I2), $a\in \nvmem_0$
\end{tabbing}

\item[Case (II)] $a[v]\notin \mt{RD}(T_0)\cup \mt{rd}(e)\cup \mt{rd}(e')$
\\By definition of $\mt{FstWt}$, $a[v]\in \mt{FstWt}(T_c)$
\end{description}
\end{proofsketch}

\begin{lem}[Array Writes are either Must Writes or Versioned]
\label{lem:io-wt-must-version-arr}
If  $\nvmem_0; \emptyset; \emptyset\Vdash_\rio \cmd_0: \m{ok}$, 
 exists $T_0$ s.t. $T_0=  (\timestamp_0, N_1, V_1,
\cmd_0)\MSeqStepsto{}  (\timestamp,  N_2, V_2,  a[e]:=e';
\cmd_1)$  ,
 $N_2,V_2 \vdash e \Downarrow_r v$
then  $\forall T_c$ s.t. $T_c = (\timestamp_0, N, V,
\cmd_0)\MSeqStepsto{} (\timestamp',  N', V',
\m{checkpoint}(\omega);c')$ 
$a[v]\in \mt{Wt}(T_c)$ or $a\in\nvmem_0$.
\end{lem}
\begin{proof}
By Lemma~\ref{lem:io-cmd-runtime-ok} there are two cases:
 (I):  $\exists$ $I$ $M$ s.t.
  $\nvmem_0, I, M \Vdash_\rio a[e]:=e'; \cmd_1: \m{ok}$ or  (II):
$\exists$ $M$ s.t.  $\nvmem_0, M \Vtaint a[e]:=e'; \cmd_1: \m{ok}$. 
\\By assumption, let $T_0$ be the trace s.t. 
$T_0 =  (\timestamp_0, N_1, V_1,
\cmd_0)\MSeqStepsto{}  (\timestamp,  N_2, V_2,  a[e]:=e'; \cmd_1)$  
\begin{description}
\item[Case I: ] $\nvmem_0, I, M \Vdash_\rio a[e]:=e'; \cmd_1: \m{ok}$
\begin{tabbing}
By Lemma~\ref{lem:not-tainted-deterministic}, 
\\\quad \= (I1)\quad \=  $\forall T$ s.t. $T = (\timestamp_0, N, V,
\cmd_0)\MSeqStepsto{} (\timestamp',  N', V',
\m{checkpoint}(\omega);c')$, 
\\\>\> $T = (\timestamp, N_1, V_1, \cmd_0) \MSeqStepsto{}
(\timestamp'',  N'_2, V'_2,  a[e]:=e'; \cmd_1) \MSeqStepsto{} (\timestamp',  N', V',
\m{checkpoint}(\omega);c')$
\\\>(I2)\> and $N_2\tntrel N'_2$ and $V_2\tntrel V'_2$ 
\\By Lemma~\ref{lem:not-tainted-eq-val} and (I2)
\\\>(I3)\>  $N'_2,V'_2 \vdash e \Downarrow_r v$
\\By (I1) and operation semantics and the definition of $\mt{Wt}$
\\\>(I4)\> $a[v]\in\mt{Wt}(T)$
\end{tabbing}

\item[Case II: ] $\exists$ $M$ s.t.  $\ee::\nvmem_0, M \Vtaint a[e']:=e; \cmd_1: \m{ok}$. 
\\ By inversion of $\ee$,  $a\in \nvmem_0$, conclusion holds
\end{description}
\end{proof}

We write $\sigma\MSeqStepsto{} \m{CP}$ to denote the trace from $\sigma$
to the nearest checkpoint. We write $N_1 \tntrel N_2$ to mean that
$N_1$ and $N_2$ only differ in tainted value. 

\begin{lem}[Taint Diff]\label{lem:instruction-taint-diff}
If $N_0\tntrel N'_0$, $V_0\tntrel V'_0$, 
$(\timestamp_0, N_0, V_0, \iota)\SeqStepsto{}  (\timestamp_0+1,
N_1, V_1, \m{skip})$, then $(\timestamp_0, N'_0, V'_0, \iota)\SeqStepsto{}  (\timestamp_0+1,
N'_1, V'_1, \m{skip})$ and $N_1\tntrel N'_1$, $V_1\tntrel V'_1$. 
\end{lem}
\begin{proofsketch} By examining all the operational semantic rules. 
\end{proofsketch}

\begin{lem}[Relating Syntactic and  Semantics Taint]\label{lem:instruction-taint-semantics}
If
$N; I_0 \Vdash \iota: I_1$ and 
$\mt{taint}(N_0\cup V_0)  \subseteq I_0$, and
$(\timestamp_0, N_0, V_0, \iota)\SeqStepsto{}  (\timestamp_0+1,
N_1, V_1, \m{skip})$, then $\mt{taint}(N_1\cup V_1)  \subseteq I_1$.
\end{lem}
\begin{proofsketch} By examining all the operational semantic rules. 
\end{proofsketch}

\begin{lem}[Not tainted execution deterministic]\label{lem:not-tainted-deterministic}
If $\ee::\nvmem; I; M\Vdash_\rio \cmd_0: \m{ok}$,
$\ee'::\nvmem; I_1; M_1 \Vdash_\rio \cmd_1: \m{ok}$,  $\ee'$ is a
sub-derivation of $\ee$,
and exists $\timestamp_0$, $N_0$, $V_0$, $N_1$, $V_1$, $T$ 
s.t. $T =  (\timestamp_0, N_0, V_0,
\cmd_0)\MSeqStepsto{}  (\_,  N_1, V_1,  \cmd_1)$, 
and $\mt{taint}(N_0\cup V_0)  \subseteq I$
then  $\forall \timestamp_0, N'_0, V'_0, T_c$ s.t. $T_c = (\timestamp'_0, N'_0, V'_0,
\cmd_0)\MSeqStepsto{} \m{CP}$, 
and $N_0\tntrel N'_0$, $V_0\tntrel V'_0$,
it is the case that 
$T_c = T_{c1}\cdot T_{c2}$ where $T_{c1} = (\timestamp'_0, N'_0, V'_0, \cmd_0)
\MSeqStepsto{}  (\_,  N'_1, V'_1,  \cmd_1)$ and $T_{c2}= (\_,  N'_1, V'_1,  \cmd_1)
\MSeqStepsto{} \m{CP}$, and $N_1\tntrel N'_1$ and $V_1\tntrel V'_1$
and $\forall x\in M_1\setminus M$, $x\in \mt{Wt}(T_{c1})$.
\end{lem}
\begin{proof}
By induction over the distance between $\ee$ and $\ee'$.  For the base
case: $\ee = \ee'$, the conclusion holds trivially. For the inductive
case, we case on $\ee$.

\begin{description}
\item[$\ee$ ends in \rulename{RIO-Seq}]  
\begin{tabbing}
\\By assumptions
\\\quad \=(1)\quad \= 
$\ee = \inferrule*[right=RIO-Seq]{
\ee_1::N;I; M \Vdash \iota: I_1, M_1 \\ 
\ee_2::N;I_1;M_1 \Vdash_\rio \cmd_1: \m{ok}}{
  N;I ; M\Vdash_\rio \iota;\cmd_1: \m{ok}}$
\\\>\> and $\ee'$ is a sub-derivation of $\ee_2$
\\By assumptions, exists $T_0$ s.t.
\\\>(2)\> $T_ 0=  (\timestamp_0, N_0, V_0,
\iota;\cmd_1)\MSeqStepsto{}  (\timestamp,  N_2, V_2,  \cmd)$ 
\\By operational semantic rules and (2)
\\\>(3)\> $T_0 =  (\timestamp_0, N_0, V_0, \iota;\cmd_1) 
\SeqStepsto{} (\timestamp_0+1, N_1, V_1, \cmd_1)
\MSeqStepsto{}  (\timestamp,  N_2, V_2,  \cmd)$ 
\\By Lemma~\ref{lem:instruction-taint-semantics}
\\\>(4)\> $\mt{taint}(N_1\cup V_1)  \subseteq I_1$
\\By I.H. on $\ee_2$, (4)
\\\>(5)\> $\forall \timestamp_1, T_1, N'_1, V'_1$ s.t. 
$T_1 = (\timestamp_1, N'_1, V'_1, \cmd_1)\MSeqStepsto{} \m{CP}$, 
\\\>\> and $N_1\tntrel N'_1$ and $V_1\tntrel V'_1$
\\\>\> $T_1 = (\timestamp_1, N'_1, V'_1, \cmd_1)
\MSeqStepsto{}  (\timestamp_2,  N'_2, V'_2,  \cmd)
\MSeqStepsto{} \m{CP}$, 
\\\>\> and $N_2\tntrel N'_2$ and $V_2\tntrel V'_2$ 
\\Given any time $\timestamp'_0$, $N'_0$,  $V'_0$ s.t. $N_0\tntrel
N'_0$ and $V_0\tntrel V'_0$
\\Given any $T'_0 = (\timestamp'_0, N'_0 , V'_0, \iota;\cmd_1)
\MSeqStepsto{}\m{CP}$
\\By operational semantic rules 
\\\>(6)\> $T'_0 = (\timestamp'_0, N'_0 , V'_0, \iota;\cmd_1)
\SeqStepsto{} (\timestamp'_0+1, N''_1, V''_1, \cmd_1)\cdot T'_1$
\\\>\> and $T'_1 = (\timestamp'_0+1, N''_1, V''_1, \cmd_1)\MSeqStepsto{}\m{CP}$
\\By Lemma~\ref{lem:instruction-taint-diff} and (3) and (5)
\\\>(7)\> $N_1\tntrel N''_1$ and  $V_1\tntrel V''_1$ 
\\By (5), (7), and (6)
\\\>(8)\> $T'_1 = (\timestamp'_0+1, N''_1, V''_1, \cmd_1)
\MSeqStepsto{}  (\timestamp'_2,  N'_2, V'_2,  \cmd)
\MSeqStepsto{} \m{CP}$
\\By (6) and (8), the conclusion holds
\end{tabbing}

\item[$\ee$ ends in \rulename{RIO-If-NDep}]
\begin{tabbing}
\\By assumptions
\\\quad \=(1)\quad \= 
$\ee = \inferrule*[right=RIO-If-no-dep]{\ee_1::I \cap rd(e) = \emptyset 
\\  \ee_2::N;I ;M\Vdash_\rio \cmd_i: \m{ok}   \\ i \in [1, 2] }{
 N;I ;M\Vdash_\rio \m{if}\ e\ \m{then}\ \cmd_1 \m{else}\ \cmd_2: \m{ok}}$
\\By assumptions, exists $T_0$ s.t.
\\\>(2)\> $T_ 0=  (\timestamp_0, N_0, V_0,
\m{if}\ e\ \m{then}\ \cmd_1\ \m{else}\ \cmd_2)\MSeqStepsto{}  (\timestamp,  N_2, V_2,  \cmd)$ 
\\We show the case where the
true branch is taken; the proof for the other case is the same.
\\By operational semantic rules and (2), 
\\\>(3)\> $T_0 =  (\timestamp_0, N_0, V_0, \m{if}\ e\ \m{then}\ \cmd_1\ \m{else}\ \cmd_2)
\SeqStepsto{} (\timestamp_0+1, N_0, V_0, \cmd_1)
\MSeqStepsto{}  (\timestamp,  N_2, V_2,  \cmd)$, 
\\\>(4)\> and $N_0, V_0 \vdash e \Downarrow_{r} \m{true}$
\\By I.H. on $\ee_2$ 
\\\>(5)\> $\forall \timestamp_1, T_1, N'_1, V'_1$ s.t. 
$T_1 = (\timestamp_1, N'_1, V'_1, \cmd_1)\MSeqStepsto{} \m{CP}$, 
\\\>\> and $N_0\tntrel N'_1$ and $V_0\tntrel V'_1$
\\\>\> $T_1 = (\timestamp_1, N'_1, V'_1, \cmd_1)
\MSeqStepsto{}  (\timestamp_2,  N'_2, V'_2,  \cmd)
\MSeqStepsto{} \m{CP}$, 
\\\>\> and $N_2\tntrel N'_2$ and $V_2\tntrel V'_2$ 
\\Given any time $\timestamp'_0$, $N'_0$,  $V'_0$ s.t. $N_0\tntrel
N'_0$ and $V_0\tntrel V'_0$
\\Given any $T'_0 = (\timestamp'_0, N'_0 , V'_0, \m{if}\ e\ \m{then}\ \cmd_1\ \m{else}\ \cmd_2)
\MSeqStepsto{}\m{CP}$
\\By assumption $\mt{taint}(N_0\cup V_0)  \subseteq I$ 
\\\>(6)\> $\mt{taint}(N'_0\cup V'_0)  \subseteq I$
\\By Lemma~\ref{lem:not-tainted-eq-val}, $\ee_1$, and (4)
\\\>(7)\> and $N'_0, V'_0 \vdash e \Downarrow_{r} \m{true}$
\\By operational semantic rules 
\\\>(8)\> $T'_0 = (\timestamp'_0, N'_0 , V'_0, \iota; \m{if}\ e\ \m{then}\ \cmd_1\ \m{else}\ \cmd_2)
\SeqStepsto{} (\timestamp'_0+1, N'_0, V'_0, \cmd_1)\cdot T'_1$
\\\>\> and $T'_1 = (\timestamp'_0+1, N'_0, V'_0, \cmd_1)\MSeqStepsto{}\m{CP}$
\\By (5), (8)
\\\>(9)\> $T'_1 = (\timestamp'_0+1, N'_0, V'_0, \cmd_1)
\MSeqStepsto{}  (\timestamp'_2,  N'_2, V'_2,  \cmd)
\MSeqStepsto{} \m{CP}$
\\By (8) and (9), the conclusion holds
\end{tabbing}

\item[$\ee$ ends in \rulename{RIO-If-Dep} or \rulename{RIO-$\iota$}]   
~\\These two cases are the base cases where $\ee=\ee'$.
\end{description}
\end{proof}

\begin{lem}[Not tainted val eq]\label{lem:not-tainted-eq-val}
  If $N_1$ and $N_2$ only differs in tainted values, and $V_1$ and
  $V_2$ only differs in tainted values, $\forall \loc\in \mt{rd}(e)$,
  $\nvmem_1\cup\vmem_1(\loc)$ is not tainted, and
  $\nvmem_1,\vmem_1 \vdash e \Downarrow_{r_1} v_1$
  $\nvmem_2, \vmem_2 \vdash e \Downarrow_{r_2} v_2$ then $r_1 = r_2$
  and $v_1 = v_2$.
\end{lem}
\begin{proofsketch}
By induction over the structure of $e$. Note that when $a^n\in
\mt{rd}(e)$, we mean all of $a[1]$ to $a[n]$ in $\mt{rd}(e)$. 
\end{proofsketch}

\begin{lem}[Partial run relates to initial state (One step)]
\label{lem:io-partial-run-one-step}
 If  all of the following hold
\begin{itemize}
\item $T= (\timestamp, \context, \nvmem, \vmem, \cmd)
  \Stepsto{o} (\timestamp+1,  \context,
\nvmem', \vmem', \cmd')$,  where  $\context = (\nvmem_0, \vmem_0, \cmd_0)$,  
\item $o$ is not checkpoint or reboot,  
\item  $\nvmem_0; \emptyset; \emptyset\Vdash_\war \cmd_0: \m{ok}$,  
   $\nvmem_0; \emptyset ; \emptyset\Vdash_\rio \cmd_0: \m{ok}$, 
\item $T_0 = (\timestamp_0, N_1, V_0, \cmd_0) \MSeqStepsto{}
  (\timestamp,  N, V, \cmd)$, $T_0$ does not
contain any checkpoints, 
\item $\nvmem_0, \nvmem_1, \vmem_0, \cmd_0 \vdash \nvmem\sim \nvmem_2$ 
\end{itemize}
then 
 $\nvmem_0, \nvmem_1, \vmem_0, \cmd_0 \vdash \nvmem' \sim \nvmem_2$. 
\end{lem}
\ifproofs
\begin{proof} 
We case on  $T$. We only show the cases where
 non-volatile memory is updated; as the conclusion trivially holds in
 other cases. 
\begin{description}
\item [Case] $T$ ends in a variable assignment (when $x:=e$ is the
  last instruction, $\cmd_1 = \m{skip}$)
\begin{tabbing}
By assumption
\\\qquad \= (1)~~ \= $
  (\timestamp, \context, \nvmem, \vmem, x:=e; \cmd_1) 
 \Stepsto{[r]}  (\timestamp+1, \context, \nvmem[x\mapsto \val], \vmem, \cmd_1) $
\\\>\> $x \in \m{dom}(N)$, and $N, V \vdash e \Downarrow_{r} \val$
\\\>(2)\> $\nvmem_0; \emptyset; \emptyset\Vdash_\war \cmd_0: \m{ok}$
\\\>(3)\> $\nvmem_0; \emptyset; \emptyset\Vdash_\rio \cmd_0: \m{ok}$
\\We only need to show that $x\in \mt{MFstWt}(N_1, V_0, \cmd_0)\cup
\nvmem_0$. 
\\By Lemma~\ref{lem:io-wt-must-version}, 
\\\>(4)\>$x\in \mt{MustWt}(N_1, V_0, \cmd_0)\cup \nvmem_0$. 
\\By (4) and Lemma~\ref{lem:io-wt-fst-version}, 
\\\>(5)\> $x\in \mt{MFstWt}(N_1, V_0, \cmd_0)\cup \nvmem_0$, 
\\The conclusion holds from (4), (5)
\end{tabbing}

\item[Case] $T$ ends in array assignment 
\begin{tabbing}
By assumption
\\\qquad \= (1)~~ \= $(\timestamp, \context, \nvmem, \vmem, a[e]:=e'; \cmd_1) 
 \Stepsto{[r,r']}  (\timestamp+1, \context, \nvmem[a[v]\mapsto v'], \vmem, \cmd_1) $
\\\>\> $N, V \vdash e \Downarrow_{r} v $ and $N, V \vdash e' \Downarrow_{r'} v'$
\\\>(2)\> $\nvmem_0; \emptyset; \emptyset\Vdash_\war \cmd_0: \m{ok}$
\\\>(3)\> $\nvmem_0; \emptyset; \emptyset\Vdash_\rio \cmd_0: \m{ok}$
\\We only need to show that $a[v]\in \mt{MFstWt}(N_1, V_0, \cmd_0)\cup
\nvmem_0$. 
\\By Lemma~\ref{lem:io-wt-must-version-arr}, 
\\\>(4)\>$a[v]\in \mt{MustWt}(N_1, V_0, \cmd_0)\cup \nvmem_0$. 
\\By Lemma~\ref{lem:io-wt-fst-version-arr}, 
\\\>(5)\> $a[v]\in \mt{MFstWt}(N_1, V_0, \cmd_0)$ or $a\in\nvmem_0$, 
\\The conclusion holds from (4), (5)
\end{tabbing}

\item[Case] $T$ ends in input assignment  
\begin{tabbing}
By assumption
\\\qquad \= (1)~~ \= $
  (\timestamp,\context, \nvmem, \vmem, x:=\m{IN}(); \cmd_1) \Stepsto{\m{in}(\timestamp)}
  (\timestamp +1, \context, \nvmem[x\mapsto \m{in}(\timestamp)^t], \vmem, \cmd_1)$
and $x \in \m{dom}(N)$
\\\>(2)\> $\nvmem_0; \emptyset; \emptyset\Vdash_\war \cmd_0: \m{ok}$
\\\>(3)\> $\nvmem_0; \emptyset; \emptyset\Vdash_\rio \cmd_0: \m{ok}$
\\We only need to show that $x\in \mt{MFstWt}(N_1, V_0, \cmd_0)\cup  \nvmem_0$. 
\\By Lemma~\ref{lem:io-wt-must-version}, 
\\\>(4)\>$x\in \mt{MustWt}(N_1, V_0, \cmd_0)\cup \nvmem_0$. 
\\By (4) and Lemma~\ref{lem:io-wt-fst-version}, 
\\\>(5)\> $x\in \mt{MFstWt}(N_1, V_0, \cmd_0)\cup \nvmem_0$, 
\\The conclusion holds from (4), (5)
\end{tabbing}
\end{description}
\end{proof}
\fi

\begin{lem}[One failure]
\label{lem:io-one-f}
If
\begin{itemize}
\item $T= (\timestamp, \context, \nvmem_1, \vmem, \cmd) 
\MStepsto{O} (\timestamp',  \context, \nvmem'_1, \vmem', \cmd')$,  $T$ does not
  execute checkpoint or reboot,
\item $\context = (\nvmem_0, \vmem_0, \cmd_0)$,  and 
\item $T_0 = (\timestamp_0, N, V_0, \cmd_0) \MSeqStepsto{O_0}
  (\timestamp,  N_1, V, \cmd)$, $T_0$ does not
contain any checkpoints, 
\item $\timestamp_0,\nvmem,\vmem_0, \cmd_0, \cmd, O_0\bnfalt_\m{in} \vdash \nvmem_1 \sim \nvmem_2$ 
\end{itemize}
then 
$ (\timestamp,  \nvmem_2, \vmem, \cmd) \MSeqStepsto{O}
(\timestamp',  \nvmem'_2, \vmem', \cmd')$, 
 and $\timestamp_0,\nvmem, \vmem_0, \cmd_0, \cmd', (O_0\cdot O)\bnfalt_\m{in}  \vdash \nvmem'_1 \sim \nvmem'_2$. 
\end{lem}
\begin{proof} 
By induction over the length of $T$, apply Lemma~\ref{lem:io-related-to-related-one-step}.
\end{proof}

\begin{lem}[related NV step to equal NV by checkpoint]
\label{lem:io-one-f-completion}
  If   $T= (\timestamp, \context, \nvmem_1, \vmem, \cmd) 
 \MStepsto{O} \Sigma'$, $\mt{CP}(\Sigma')$ 
and $T$ does not
  execute checkpoint or reboot,
 $\context = (\nvmem_0, \vmem_0, \cmd_0)$, 
$T_0 = (\timestamp_0, N, V_0, \cmd_0) \MSeqStepsto{O_0}  
  (\timestamp,  N_1, V, \cmd)$, $T_0$ does not
contain any checkpoints, 
 and $\timestamp_0,\nvmem, \vmem_0, \cmd_0, \cmd, O_0\bnfalt_\m{in} \vdash \nvmem_1 \sim \nvmem_2$, 
 then
  $ (\timestamp, \nvmem_2, \vmem, \cmd) \MSeqStepsto{O} \erase{\Sigma'}$
\end{lem}
\ifproofs
\begin{proof} 
By induction over the length of $T$.
\begin{description}
\item[Base case:] $|T| = 0$
\begin{tabbing}
By assumption
\\\qquad \= (1)~~ \= $\cmd= \m{checkpoint}(\omega);\cmd'$
\\ \>(2)\> $\forall \loc \in N_1$ s.t. $N_1(\loc) \neq N_2(\loc)$,  
$\loc \in \mt{MFstWt}(N, \vmem_0, \cmd_0)$ and 
\\\>\> $\loc \in \mt{MstWt}(\nvmem_1, \vmem,
\m{checkpoint}(\omega);\cmd') = \emptyset$
\\By (2) and $\nexists \loc$ s.t. $\loc \in \emptyset$ 
\\\>(3)\> $N_1 = N_2$
\\ The continuous powered execution also takes $0$ steps and the
conclusion holds
\end{tabbing}

\item[Inductive case:]
~
\begin{tabbing}
By assumption
\\\qquad \= (1)~~ \= 
$T= (\timestamp, \context, \nvmem_1, \vmem, \cmd) \Stepsto{o} 
(\timestamp_1, \context, \nvmem'_1, \vmem_1, \cmd_1) \MStepsto{O} 
(\timestamp', \context,  \nvmem', \vmem', \m{checkpoint}(\omega);\cmd')$, 
\\By Lemma~\ref{lem:io-related-to-related-one-step}
\\\>(2)\>  $(\timestamp,\nvmem_2, \vmem, \cmd) \SeqStepsto{o} (\timestamp_1, \nvmem'_2,
\vmem_1, \cmd_1) $ and $\timestamp_0,\nvmem, \vmem_0, \cmd_0, \cmd_1, o\bnfalt_{\mt{in}} 
\vdash \nvmem'_1 \sim \nvmem'_2$.
\\By I.H. on the tail of $T$ 
\\\>(3)\>$(\timestamp_1,\nvmem'_2,\vmem_1, \cmd_1) \MSeqStepsto{O} (\timestamp', \nvmem',
\vmem', \m{checkpoint}(\omega);\cmd')$
\\By (2) and (3)
\\\> The conclusion holds
\end{tabbing}
\end{description}
\end{proof}
\fi

\begin{lem}[related NVs step to related NVs (One step)] 
\label{lem:io-related-to-related-one-step}
  If   $T= (\timestamp, \context, \nvmem_1, \vmem, \cmd) 
 \Stepsto{o} (\timestamp+1, \context,
  \nvmem'_1, \vmem', \cmd')$, and $T$ does not
  execute checkpoint or reboot,
 $\context = (\nvmem_0, \vmem_0, \cmd_0)$, 
$T_0 = (\timestamp_0, N, \vmem_0, \cmd_0) \MSeqStepsto{O_0}  
(\timestamp,  N_1, \vmem, \cmd)$, $T_0$ does not
contain any checkpoints, 
 and $\timestamp_0,\nvmem,  \vmem_0, \cmd_0, \cmd, O_0\bnfalt_\m{in} \vdash \nvmem_1 \sim \nvmem_2$, 
 then
  $ (\timestamp, \nvmem_2, \vmem, \cmd) \SeqStepsto{o} 
 (\timestamp+1, \nvmem'_2, \vmem', \cmd')$
and $\timestamp_0,\nvmem, \vmem_0, \cmd_0, \cmd',  (O_0\cdot o) \bnfalt_\m{in}\vdash \nvmem'_1 \sim \nvmem'_2$.
\end{lem}
\ifproofs
\begin{proof}
By examining the structure of $T$. 
\begin{description}
\item[Case:] $T$ ends in \rulename{I/O-CP-Skip} rule. 
\begin{tabbing}
By assumption
\\\qquad \= (1)~~ \= 
$T= (\timestamp, \context, \nvmem_1, \vmem, \m{skip};\cmd) 
 \Stepsto{}  (\timestamp+1, \context, \nvmem_1, \vmem, \cmd) $ and
\\\>(2)\> $\context = (\nvmem_0, \vmem_0, \cmd_0)$,
  and $\timestamp_0,\nvmem, \vmem_0, \cmd_0, (\m{skip};\cmd), O_0\bnfalt_\m{in}  \vdash \nvmem_1 \sim \nvmem_2$
\\ By  \rulename{I/O-Tnt-Skip} rule
\\\>(3)\> $ (\timestamp, \nvmem_2, \vmem, \m{skip};\cmd)  \SeqStepsto{}
(\timestamp+1, \nvmem_2, \vmem, \cmd)$
\\By (2) and $\m{skip}$ does not write to store
\\\>(4)\>  $\timestamp_0,\nvmem, \vmem_0, \cmd_0, \cmd, O_0\bnfalt_\m{in} \vdash \nvmem_1 \sim \nvmem_2$
\end{tabbing}

\item[Case:] $T$ ends in \rulename{I/O-CP-Seq} rule and uses
  \rulename{I/O-CP-NV-Assign} or ends in \rulename{I/O-CP-NV-Assign}
  rule. The proof for both of these two cases are the same except that
  in the latter, the resulting command is $\m{skip}$.
\begin{tabbing}
By assumption
\\\qquad \= (1)~~ \= 
$T= (\timestamp, \context, \nvmem_1, \vmem, x:=e;\cmd) 
 \Stepsto{[r_1]}  (\timestamp +1,\context, \nvmem_1[x \mapsto v_1], \vmem, \cmd) $ and
\\\>(2)\> $\nvmem_1; \vmem \vdash e \Downarrow_{r_1} v_1$ and 
\\\>(3)\> $\context = (\nvmem_0, \vmem_0, \cmd_0)$ and
  $\timestamp_0, \nvmem, \vmem_0, \cmd_0, (x:=e;\cmd), O_0\bnfalt_\m{in} \vdash \nvmem_1 \sim \nvmem_2$
\\ By  \rulename{I/O-Tnt-NV-Assign} rule
\\\>(4)\> $ (\timestamp, \nvmem_2, \vmem, x:=e;\cmd)  \SeqStepsto{[r_2]}
(\timestamp+1, \nvmem_2[x \mapsto v_2], \vmem, \cmd)$ and
\\\>(5)\> $\nvmem_2; \vmem \vdash e \Downarrow_{r_2} v_2$
\\ By Lemma~\ref{lem:io-related-n-eq-var} and Lemma~\ref{lem:eq-var-eq-eval}
\\\>(6)\> $v_1 = v_2$ and $r_1 = r_2$
\\By (3) and (6), and 
\\\>(7)\>  $\timestamp_0,\nvmem, \vmem_0, \cmd_0, \cmd,O_0\bnfalt_\m{in} 
 \vdash \nvmem_1[x \mapsto v_1] \sim \nvmem_2[x \mapsto v_2]$
\end{tabbing}

\item[Case:] $T$ ends in  \rulename{I/O-CP-Seq} rule and uses
  \rulename{I/O-CP-Arr-Assign} rule or ends in \rulename{I/O-CP-Arr-Assign} rule. 
\begin{tabbing}
By assumption
\\\qquad \= (1)~~ \= 
$T= (\timestamp, \context, \nvmem_1, \vmem, a[e']:=e;\cmd) 
 \Stepsto{[r_1',r_1]}  (\timestamp+1,\context, \nvmem_1[ (a[v_1']) \mapsto v_1], \vmem, \cmd) $ and
\\\>(2)\> $\nvmem_1; \vmem \vdash e \Downarrow_{r_1} v_1$ 
and $\nvmem_1; \vmem \vdash e' \Downarrow_{r_1'} v_1'$ 
\\\>(3)\> $\context = (\nvmem_0, \vmem_0, \cmd_0)$ and
  $\timestamp_0,\nvmem, \vmem_0, \cmd_0, (a[e']:=e;\cmd), O_0\bnfalt_\m{in} \vdash \nvmem_1 \sim \nvmem_2$
\\ By  \rulename{I/O-Tnt-Assign-Arr} rule
\\\>(4)\> $ (\timestamp, \nvmem_2, \vmem, a[e']:=e;\cmd)  \SeqStepsto{[r_2',r_2]}
(\timestamp+1, \nvmem_2[(a[v_{2}']) \mapsto v_2], \vmem, \cmd)$ and
\\\>(5)\> $\nvmem_2; \vmem \vdash e' \Downarrow_{r_2'} v_{2}'$ and $\nvmem_2; \vmem \vdash e \Downarrow_{r_2} v_2$
\\ By Lemma~\ref{lem:io-related-n-eq-var} and Lemma~\ref{lem:eq-var-eq-eval}
\\\>(6)\> $v_1 = v_2$ and $v_1' = v_2'$ and $r_1 = r_2$ and $r_1' = r_2'$
\\By (3) and (6)
\\\>(7)\>  $\timestamp_0,\nvmem, \vmem_0, \cmd_0, \cmd, O_0\bnfalt_\m{in}  \vdash \nvmem_1[a[v_1'] \mapsto v_1] \sim \nvmem_2[a[v_2'] \mapsto v_2]$
\end{tabbing}

\item[Case:] $T$ ends in  \rulename{I/O-CP-Seq} rule and uses
  \rulename{I/O-CP-V-Assign} rule or ends in \rulename{I/O-CP-V-Assign} rule.  
\begin{tabbing}
By assumption 
\\\qquad \= (1)~~ \= 
$T= (\timestamp, \context, \nvmem_1, \vmem, x:=e;\cmd) 
 \Stepsto{[r_1]}  (\timestamp+1, \context, \nvmem_1, \vmem[x\mapsto v_1], \cmd) $
 and $\nvmem_1; \vmem \vdash e \Downarrow_{r_1} v_1$
\\\>(2)\> $\context = (\nvmem_0, \vmem_0, \cmd_0)$,
and  $\timestamp_0,\nvmem,  \vmem_0, \cmd_0, (x:=e;\cmd), O_0\bnfalt_\m{in} \vdash \nvmem_1 \sim \nvmem_2$
\\ By  \rulename{I/O-Tnt-V-Assign} rule
\\\>(3)\> $ (\timestamp, \nvmem_2, \vmem, x:=e;\cmd)  \SeqStepsto{[r_2]}
(\timestamp+1, \nvmem_2, \vmem[x\mapsto v_2] , \cmd)$ and  $\nvmem_2; \vmem \vdash e
\Downarrow_{r_2} v_2$
\\ By Lemma~\ref{lem:io-related-n-eq-var} and Lemma~\ref{lem:eq-var-eq-eval}
\\\>(4)\> $v_1 = v_2$ and $r_1 = r_2$
\\By (4), $\vmem[x\mapsto v_1] =\vmem[x\mapsto v_2]$ 
\\By (2) and \rulename{CP-V-Assign} does not write to non-volatile store
\\\>(4)\>  $\timestamp_0,\nvmem,  \vmem_0, \cmd_0, \cmd, O_0\bnfalt_\m{in}  \vdash \nvmem_1 \sim \nvmem_2$
\end{tabbing}

\item[Case:] $T$ ends in \rulename{I/O-CP-Seq} rule and uses
  \rulename{I/O-CP-NV-Assign-In} or ends in \rulename{I/O-CP-NV-Assign-In}
  rule. 
\begin{tabbing}
By assumption
\\\qquad \= (1)~~ \= 
$T= (\timestamp, \context, \nvmem_1, \vmem, x:=\mt{IN}();\cmd) 
 \Stepsto{\mt{in}(\timestamp)}  
 (\timestamp +1, \context, \nvmem_1[x \mapsto \mt{in}(\timestamp)^t], \vmem, \cmd) $
  and 
\\\>(2)\> $\context = (\nvmem_0, \vmem_0, \cmd_0)$ and
  $\timestamp_0,\nvmem, \vmem_0, \cmd_0, (x:=\mt{IN}();\cmd), O_0\bnfalt_\m{in} \vdash \nvmem_1 \sim \nvmem_2$
\\ By  \rulename{I/O-Tnt-NV-Assign-In} rule
\\\>(3)\> $ (\timestamp, \nvmem_2, \vmem, x =\mt{IN}();\cmd)  \SeqStepsto{\mt{in}(\timestamp)}
(\timestamp+1,\nvmem_2[x \mapsto \mt{in}(\timestamp)^t], \vmem, \cmd)$ and
\\By (1) and (3)
\\\>(7)\>  $\timestamp_0,\nvmem, \vmem_0, \cmd_0, \cmd,(O_0\cdot \mt{in}(\timestamp))\bnfalt_\m{in} 
 \vdash \nvmem_1[x \mapsto \mt{in}(\timestamp)^t] \sim \nvmem_2[x \mapsto \mt{in}(\timestamp)^t]$
\end{tabbing}

\item[Case:] $T$ ends in  \rulename{I/O-CP-Seq} rule and uses
  \rulename{I/O-CP-Arr-Assign-In} rule or ends in \rulename{I/O-CP-Arr-Assign-In} rule. 
\begin{tabbing}
By assumption
\\\qquad \= (1)~~ \= 
$T= (\timestamp, \context, \nvmem_1, \vmem, a[e']:=\mt{IN}();\cmd) 
 \Stepsto{[r_1'],\mt{in}(\timestamp)}  
 (\timestamp+1, \context, \nvmem_1[ (a[v_1']) \mapsto \mt{in}(\timestamp)^t], \vmem, \cmd) $ and
\\\>(2)\> $\nvmem_1; \vmem \vdash e' \Downarrow_{r_1'} v_1'$ 
\\\>(3)\> $\context = (\nvmem_0, \vmem_0, \cmd_0)$ and
  $\timestamp_0,\nvmem, \vmem_0, \cmd_0, (a[e']:=\mt{IN}();\cmd), O_0\bnfalt_\m{in} \vdash \nvmem_1 \sim \nvmem_2$
\\ By  \rulename{I/O-Tnt-Assign-Arr-In} rule
\\\>(4)\> $ (\timestamp,\nvmem_2, \vmem, a[e']:=\mt{IN}();\cmd)  \SeqStepsto{[r_2'],\mt{in}(\timestamp)}
(\timestamp+1, \nvmem_2[(a[v_{2}']) \mapsto \mt{in}(\timestamp)^t], \vmem, \cmd)$ and
\\\>(5)\> $\nvmem_2; \vmem \vdash e' \Downarrow_{r_2'} v_{2}'$
\\ By Lemma~\ref{lem:io-related-n-eq-var} and Lemma~\ref{lem:eq-var-eq-eval}
\\\>(6)\> $v_1' = v_2'$ and $r_1' = r_2'$
\\By (3), (4) and (6)
\\\>(7)\>  $\timestamp_0,\nvmem, \vmem_0, \cmd_0, \cmd, (O_0\cdot \mt{in}(\timestamp))\bnfalt_\m{in}  
\vdash \nvmem_1[a[v_1'] \mapsto \mt{in}(\timestamp)^t] \sim \nvmem_2[a[v_2'] \mapsto \mt{in}(\timestamp)^t]$
\end{tabbing}

\item[Case:] $T$ ends in  \rulename{I/O-CP-Seq} rule and uses
  \rulename{I/O-CP-V-Assign-In} rule or ends in \rulename{I/O-CP-V-Assign-In} rule.  
\begin{tabbing}
By assumption 
\\\qquad \= (1)~~ \= 
$T= (\timestamp, \context, \nvmem_1, \vmem, x:=\mt{IN}();\cmd) 
 \Stepsto{\mt{in}(\timestamp)}  
 (\timestamp+1,\context, \nvmem_1, \vmem[x\mapsto \mt{in}(\timestamp)^t], \cmd) $
\\\>(2)\> $\context = (\nvmem_0, \vmem_0, \cmd_0)$,
and  $\timestamp_0,\nvmem,  \vmem_0, \cmd_0, (x:=\mt{IN}();\cmd), O_0\bnfalt_\m{in} \vdash \nvmem_1 \sim \nvmem_2$
\\ By  \rulename{I/O-Tnt-V-Assign-In} rule
\\\>(3)\> $ (\timestamp, \nvmem_2, \vmem, x:=\mt{IN}();\cmd)  \SeqStepsto{\mt{in}(\timestamp)}
(\timestamp+1, \nvmem_2, \vmem[x\mapsto \mt{in}(\timestamp)^t] , \cmd)$ 
\\By (2) and \rulename{CP-V-Assign} does not write to non-volatile store
\\\>(4)\>  $\timestamp_0,\nvmem,  \vmem_0, \cmd_0, \cmd, (O_0\cdot \mt{in}(\timestamp))\bnfalt_\m{in}  
\vdash \nvmem_1 \sim \nvmem_2$
\end{tabbing}

\item[Case:] $T$ ends in \rulename{I/O-CP-If-T} rule. 
\begin{tabbing}
By assumption
\\\qquad \= (1)~~ \= 
$T= (\timestamp,\context, \nvmem_1, \vmem, \m{if}~e~\m{then}~c_1~\m{else}~c_2) 
 \Stepsto{[r_1]}  (\timestamp+1,\context, \nvmem_1, \vmem, \cmd_1) $ and $\nvmem_1,
 \vmem  \vdash e \Downarrow_{r_1} \m{true}$ and
\\\>(2)\> $\context = (\nvmem_0, \vmem_0, \cmd_0)$, and 
  $\timestamp_0,\nvmem, \vmem_0, \cmd_0, \m{if}~e~\m{then}~c_1~\m{else}~c_2, O_0\bnfalt_\m{in} 
  \vdash \nvmem_1 \sim \nvmem_2$
\\Let's assume
\\\> (3)\>  $\nvmem_2, \vmem  \vdash e \Downarrow_{r_2} v$
\\By Lemma~\ref{lem:io-related-n-eq-var} and Lemma~\ref{lem:eq-var-eq-eval}
\\\>(4)\> $v = \m{true}$ and $r_1 = r_2$
\\ By  \rulename{I/O-Tnt-If-T} rule and (4)
\\\>(5)\> $ (\timestamp, \nvmem_2, \vmem, \m{if}~e~\m{then}~c_1~\m{else}~c_2)  
\SeqStepsto{[r_2]} (\timestamp +1, \nvmem_2, \vmem, \cmd_1)$
\\ By (5) and \rulename{CP-If-T} doesn't write to storage, and the  must
write set of if statement
\\  is the intersection of the must writes in both branches, 
\\\>(6)\>  $\timestamp_0,\nvmem, \vmem_0, \cmd_0, \cmd_1, O_0\vdash \nvmem_1 \sim \nvmem_2$
\end{tabbing}

\item[Case:] $T$ ends in   \rulename{CP-If-F} rule. Similar to the
  previous case. 
\end{description}
\end{proof}
\fi

\begin{lem}[Eq val in read locations] \label{lem:io-related-n-eq-var}
If 
$T= (\timestamp_0,\nvmem, \vmem_0, \cmd_0) 
 \MSeqStepsto{O} (\timestamp,  \nvmem_1, \vmem', \iota;\cmd)$,  $T$ does not
  execute checkpoint,
$\timestamp_0,\nvmem,  \vmem_0, \cmd_0, (\iota;\cmd),  O\bnfalt_\m{in}\vdash \nvmem_1 \sim \nvmem_2$ 
where
$\iota\neq \m{checkpoint}$, 
and $\loc\in \mt{rd}(\iota)\cap \nvmem_1$,
then $\nvmem_1(\loc) = \nvmem_2(\loc)$.
\end{lem}
\ifproofs
\begin{proof}
We assume that  $\nvmem_1(\loc) \neq \nvmem_2(\loc)$ then derive a
contradiction.
\begin{tabbing}
By assumption, 
\\\qquad \= (1)~~ \= $\nvmem_1(\loc) \neq \nvmem_2(\loc)$  and
$\timestamp_0, \nvmem, \vmem_0, \cmd_0, (\iota;\cmd)\vdash \nvmem_1 \sim \nvmem_2$.
\\ By definition~\ref{def:io-same_exec}
\\\>(2)\> 
\\\>\>$\forall x \in N_1$ s.t. $N_1(x) \neq N_2(x)$, 
 $x\in\mt{MFstWt}(N, V_0, \cmd_0)$, $x\in\mt{MstWt}(N_1, V', \iota;\cmd)$ and $x\notin \mt{Wt}(T')$
 \\\>
where $\{T'\} = \runof{\sigma, O\bnfalt_{\mt{in}}, (\iota;\cmd)}$,
  $\sigma=(\timestamp_0, \nvmem, V_0, \cmd_0)$ and the last state of $T'$ is
  $(\timestamp,N_1, V', \iota;\cmd)$
\\By (1) and (2)
\\\>(3)\> $\loc\in\mt{MFstWt}(N, V_0, \cmd_0)$, 
and $\loc\notin \mt{Wt}(T')$
\\By $\loc\in \mt{rd}(\iota)$, $\loc \notin \mt{WT}(T')$, 
and the definitions of $\mt{FstWt}$ and $\mt{MFstWt}$
\\\>(4)\> $\loc\notin \mt{MFstWt}(N, V_0, c_0)$
\\(3) and (4) derive a contradiction, so the conclusion holds.
\end{tabbing}
\end{proof}
\fi

\end{appendices}
\newpage
\bibliography{ms}

\end{document}